\def\llncs{0}
\def\fullpage{1}
\def\anonymous{0}

\def\notxfont{0}
\def\submission{0}
\def\cameraready{0}

\ifnum\submission=1
\def\anonymous{0}
\def\llncs{1}

\else
\fi

\ifnum\cameraready=1
\def\submission{1}
\def\llncs{1}

\def\anonymous{0}
\else
\fi

\ifnum\anonymous=1
\if0
\def\llncs{0}
\fi

\else
\fi

\ifnum\llncs=1
	\documentclass[runningheads,a4paper]{llncs}
\else
	\documentclass[letterpaper,hmargin=1.05in,vmargin=1.05in,11pt]{article}
			\ifnum\fullpage=1
		\usepackage{fullpage}
		\fi
\fi

\usepackage[%
  colorlinks=true,
  citecolor=blue,
  pagebackref=true
]{hyperref}

\usepackage{amsmath, amsfonts, amssymb, mathtools,amscd}

\usepackage{amsthm}

\usepackage{fancyhdr}
\usepackage{lmodern}
\usepackage[T1]{fontenc}
\usepackage[utf8]{inputenc}

\usepackage{arydshln} % In order to use \hdashline
\usepackage{url}
\usepackage{ifthen}
\usepackage{bm}
\usepackage{multirow}
\usepackage[dvips]{graphicx}
\usepackage[usenames]{color}
\usepackage{xcolor,colortbl} % Leave out in case the usepackage cannot be found. Not critical
\usepackage{threeparttable}
\usepackage{comment}
\usepackage{paralist,verbatim}
\usepackage{cases}
\usepackage{booktabs}
\usepackage{braket}
\usepackage{cancel} 
\usepackage{ascmac} 
\usepackage{framed}
\usepackage{physics}
\usepackage{authblk}
\usepackage{pifont}
\definecolor{darkblue}{rgb}{0,0,0.6}
\definecolor{darkgreen}{rgb}{0,0.5,0}
\definecolor{maroon}{rgb}{0.5,0.1,0.1}
\definecolor{dpurple}{rgb}{0.2,0,0.65}

\usepackage[capitalise,noabbrev]{cleveref}
\usepackage[absolute]{textpos}
\usepackage[final]{microtype}
\usepackage[absolute]{textpos}
\usepackage{everypage}

\newtheoremstyle{thicktheorem}%
{\topsep}
{\topsep}
{\itshape}{}%
{\bfseries}%
{.}
{ }%
{\thmname{#1}\thmnumber{ #2}%
		\thmnote{ (#3)}%
}

\newtheoremstyle{remark}%name
{\topsep}
{\topsep}
	{}%body font
	{}%indent amount
	{}%theorem head font
	{.}%punctuation after theorem head
	{ }%space after theorem head
	{\textit{\thmname{#1}}\thmnumber{ #2}%theorem head specs
			\thmnote{ (#3)}%
	}

\ifnum\llncs=0
	\theoremstyle{thicktheorem}
	\newtheorem{theorem}{Theorem}[section]
	\newtheorem{lemma}[theorem]{Lemma}
    
	\newtheorem{corollary}[theorem]{Corollary}
	
	\newtheorem{definition}[theorem]{Definition}
	\newtheorem{game}[theorem]{Game}

	\theoremstyle{remark}
	
	\newtheorem{remark}[theorem]{Remark}

\else
\fi

	\crefname{theorem}{Theorem}{Theorems}
	\crefname{assumption}{Assumption}{Assumptions}
	\crefname{construction}{Construction}{Constructions}
	\crefname{corollary}{Corollary}{Corollaries}
	\crefname{conjecture}{Conjecture}{Conjectures}
	\crefname{definition}{Definition}{Definitions}
	\crefname{exmaple}{Example}{Examples}
	\crefname{experiment}{Experiment}{Experiments}
	\crefname{counterexample}{Counterexample}{Counterexamples}
	\crefname{lemma}{Lemma}{Lemmata}
	\crefname{observation}{Observation}{Observations}
	\crefname{proposition}{Proposition}{Propositions}
	\crefname{remark}{Remark}{Remarks}
	\crefname{claim}{Claim}{Claims}
	\crefname{fact}{Fact}{Facts}
	\crefname{note}{Note}{Notes}

\ifnum\llncs=1
 \crefname{appendix}{App.}{Appendices}
 \crefname{section}{Sec.}{Sections}
\else
\fi

\ifnum\llncs=1
\renewcommand*{\backref}[1]{}
\else
	\renewcommand*{\backref}[1]{(Cited on page~#1.)}
	\ifnum\notxfont=1
	\else
		\usepackage{newtxtext}
	\fi
\fi

\newcommand{\mor}[1]{$\ll$\textsf{\color{red} Tomoyuki: { #1}}$\gg$}
\newcommand{\minki}[1]{$\ll$\textsf{\color{red} Minki: { #1}}$\gg$}
\newcommand{\takashi}[1]{$\ll$\textsf{\color{red} Takashi: { #1}}$\gg$}

\newcommand{\state}{\mathsf{state}}

\newcommand{\Sim}{\algo{Sim}}

\newcommand{\Samp}{\algo{Samp}}
\newcommand{\Ver}{\algo{Ver}}
\newcommand{\ans}{\mathsf{ans}}
\newcommand{\puzz}{\mathsf{puzz}}

\newcommand{\cA}{\mathcal{A}}
\newcommand{\cB}{\mathcal{B}}
\newcommand{\cC}{\mathcal{C}}
\newcommand{\cD}{\mathcal{D}}
\newcommand{\cE}{\mathcal{E}}

\newcommand{\cO}{\mathcal{O}}
\newcommand{\cP}{\mathcal{P}}
\newcommand{\cQ}{\mathcal{Q}}
\newcommand{\cR}{\mathcal{R}}

\newcommand{\cV}{\mathcal{V}}

\def\makeuppercase#1{
\expandafter\newcommand\csname tl#1\endcsname{\widetilde{#1}}
}

\def\makelowercase#1{
\expandafter\newcommand\csname tl#1\endcsname{\widetilde{#1}}
}

\newcommand{\secp}{\lambda}

\newcommand*{\pk}{\keys{pk}}

\newcommand*{\keys}[1]{\mathsf{#1}}

\newcommand*{\algo}[1]{\ensuremath{\mathsf{#1}}}

\newenvironment{boxfig}[2]{\begin{figure}[#1]\fbox{\begin{minipage}{0.97\linewidth}
                        \vspace{0.2em}
                        \makebox[0.025\linewidth]{}
                        \begin{minipage}{0.95\linewidth}
            {{
                        #2 }}
                        \end{minipage}
                        \vspace{0.2em}
                        \end{minipage}}}{\end{figure}}

\newcommand{\bit}{\{0,1\}}

\newcommand{\negl}{{\mathsf{negl}}}

\newcommand{\poly}{{\mathrm{poly}}}

\newcommand{\TD}{\mathsf{TD}}
\usepackage{tikz}

\makeatletter
\@addtoreset{theorem}{section}

\let\c@lemma\c@theorem

\let\c@corollary\c@theorem

\let\c@definition\c@theorem

\makeatother

\ifnum\llncs=1
\let\oldvec\vec% Store \vec in \oldvec
\let\vec\oldvec% Restore \vec from \oldvec
\makeatletter
\renewcommand*\l@author[2]{}
\renewcommand*\l@title[2]{}
\makeatletter

\else
\fi

\theoremstyle{remark}
\title{
\textbf{
Proofs of Quantum Memory
} 
}

\begin{document}
%\author{}
%\institute{}

\ifnum\anonymous=1
\author{\empty}
\institute{\empty}
\else
%
%  For camera ready version.
%
\ifnum\llncs=1
\index{Hhan, Minki}
\index{Morimae, Tomoyuki}
\index{Okinaka, Yasuaki}
\index{Yamakawa, Takashi}

\author{
	Minki Hhan\inst{1}
	\and
	Tomoyuki Morimae\inst{2}
	\and
	Yasuaki Okinaka\inst{2}
	\and
	Takashi Yamakawa\inst{3,4,2}
}

\authorrunning{M. Hhan et al.}

\institute{
	School of Computing,
	KAIST, Daejeon, Republic of Korea\\
	\email{minkihhan@kaist.ac.kr}
	\and
	Yukawa Institute for Theoretical Physics,
	Kyoto University, Kyoto, Japan\\
	\email{\{tomoyuki.morimae,yasuaki.okinaka\}@yukawa.kyoto-u.ac.jp}
	\and
	NTT Social Informatics Laboratories, Tokyo, Japan\\
	\email{takashi.yamakawa@ntt.com}
	\and
	NTT Research Center for Theoretical Quantum Information,
	Atsugi, Japan
}

\else
%
%   For full/eprint version, etc.
%

\author[1]{Minki Hhan}
\author[2]{Tomoyuki Morimae}
\author[2]{Yasuaki Okinaka}
\author[3,4,2]{Takashi Yamakawa}
\affil[1]{{\small 
School of Computing,
KAIST, Daejeon, Republic of Korea}\authorcr{\small minkihhan@kaist.ac.kr}}
\affil[2]{{\small Yukawa Institute for Theoretical Physics, Kyoto University, Kyoto, Japan}\authorcr{\small \{tomoyuki.morimae,yasuaki.okinaka\}@yukawa.kyoto-u.ac.jp}}
\affil[3]{{\small NTT Social Informatics Laboratories, Tokyo, Japan}\authorcr{\small takashi.yamakawa@ntt.com}}
\affil[4]{{\small NTT Research Center for Theoretical Quantum Information, Atsugi, Japan}}

\renewcommand\Authands{, }
\fi %%%%% END OF LNCS branch
\fi

\ifnum\llncs=1
\date{}
\else
\date{}
\fi

\maketitle
%\ifnum\llncs=0
%\thispagestyle{fancy}
%\rhead{YITP-24-125}
%\else
%\fi

%\pagenumbering{gobble} % Turn off page numbering temporarily

\begin{abstract}
With the rapid advances in quantum computer architectures and the emerging prospect of large-scale quantum memory, it is becoming essential to 
classically verify that remote devices genuinely allocate the promised quantum memory with a specified number of qubits. In this paper, we introduce a new concept, \emph{proofs of quantum memory (PoQM)}. 
A PoQM is an interactive protocol between a classical probabilistic polynomial-time (PPT) verifier and a quantum polynomial-time (QPT) prover over a classical channel
where the verifier can verify that the prover has possessed a quantum memory with a specified number of qubits for a specified period of time.
PoQM generalize the well-studied notion of
proofs of quantumness (PoQ) [Brakerski, Christiano, Mahadev, Vazirani, and Vidick, JACM 2021]
where a classical verifier can verify that the prover is not classical.
Our contributions are summarized as follows:
\begin{itemize}
    \item 
We introduce a formal definition of PoQM.
We also introduce a variant of PoQM, which we call inefficient-verifier PoQM (IV-PoQM), where
the verifier's final computation to make the decision is not necessarily efficient.
Clearly, PoQM imply IV-PoQM.
\item
We construct PoQM based on the hardness of LWE.
Specifically,
we give two constructions of PoQM. The first one is of two-round (i.e., four-message) and has negligible soundness error under
the subexponential-hardness of LWE. The second one is of polynomial-round and has inverse-polynomial soundness error 
under the polynomial-hardness of LWE.
\item 
As a lowerbound of IV-PoQM (and therefore PoQM), we show that IV-PoQM imply one-way puzzles.
Moreover, we show that
a certain restricted version of PoQM implies quantum computation classical communication (QCCC) key exchange, 
which suggests the difficulty of black-box constructing PoQM from one-way functions.
\item
We show that constant-round PoQ imply PoQM or
single-round PoQ (with a quantum verifier). 
Single-round PoQ are ``trivial'' PoQ in the sense that the verifier asks the prover to solve
a classical problem which is
quantumly easy but classically hard. The result therefore demonstrates that
PoQM capture ``genuinely-interactive'' PoQ.
\item
We show that if constant-round IV-PoQ that are black-box constructed from quantumly-secure falsifiable assumptions exist
then IV-PoQM exist. 
This result implies that IV-PoQM can be constructed from quantumly-secure constant-round statistically-hiding commitments
(and therefore from quantumly-secure collision-resistant hash functions).
\end{itemize}
\end{abstract}

\ifnum\cameraready=1
\else
\ifnum\llncs=1
\else
\newpage
  \setcounter{tocdepth}{2}      % sections in table if depth < i
  \setcounter{secnumdepth}{2}   % sections numbered if depth < i
  \tableofcontents
  \pagenumbering{arabic}
  \setcounter{page}{0}          % set the table contents page as 0-th page
  \thispagestyle{empty}
  \clearpage
\fi
\fi

\section{Introduction}
Imagine a quantum computing startup claiming that it has built a quantum processor equipped with a 100-million-qubit quantum memory.
How could a classical investor verify such a claim?

Proofs of quantumness (PoQ)~\cite{JACM:BCMVV21} are insufficient for this purpose.
A proof of quantumness is an interactive protocol between a classical probabilistic polynomial-time (PPT) verifier and a quantum polynomial-time (QPT) prover over a classical channel.
Completeness is that if the prover behaves honestly, the verifier accepts with high probability, and soundness is that the verifier rejects with high probability if
the prover is PPT. 
Using PoQ, a classical investor could confirm that the quantum startup is doing something at least non-classical, 
but it cannot verify that the startup can manipulate a 100-million-qubit quantum memory.
%Moreover, the classical investor cannot confirm that the quantum startup can keep the quantum coherence for 100 hours.

With the rapid advances in quantum computer architectures and the emerging prospect of large-scale quantum memory, 
the following question is becoming essential:
\begin{center}
\emph{Can a classical verifier verify that a remote prover has possessed a quantum memory with a specified number of qubits for a specified period of time?}    
\end{center}
%\item {\bf Verification of the coherence time}:
%\emph{Can a classical verifier verify that a remote prover has kept a quantum coherence for a specified period of time?}    

\subsection{Our Results}
In this paper, we address this question by introducing a new concept, which we call \emph{proofs of quantum memory} (PoQM).
A PoQM is an interactive protocol between a PPT verifier and a QPT prover over a classical channel 
where the verifier can verify that the prover has kept a specified number of qubits for a specified time period.

Our contributions are summarized as follows:\footnote{The results in the fourth and fifth items actually hold up to i.o.-soundness, but for simplicity, here
we omit it.}
\begin{enumerate}
    \item
{\bf Definitions.} 
We give a formal definition of PoQM.
We also define a variant of PoQM, which we call inefficient-verifier PoQM (IV-PoQM), where
the verifier's final computation to make the decision is not necessarily efficient.
Clearly, PoQM imply IV-PoQM.
\item
{\bf Constructions.}
We provide two constructions of PoQM from the hardness of LWE.
The first one is of two-round (i.e., four-message) and has negligible soundness error under
the subexponential-hardness of LWE. The second one is of polynomial-round and has inverse-polynomial soundness error 
under the polynomial-hardness of LWE.
\item
{\bf Lowerbounds.}
We show a lowerbound of IV-PoQM (and therefore PoQM). 
We show that IV-PoQM imply one-way puzzles (OWPuzzs),
which are a natural quantum analogue of one-way functions (OWFs) introduced by Khurana and Tomer~\cite{STOC:KhuTom24}.
We also show that a certain restricted version of PoQM, which includes one of our constructions, implies quantum key-exchange over a classical channel (QCCC KE),
which suggests that black-box constructions of PoQM from OWFs will be difficult. 
\item 
{\bf PoQM capture genuinely-interactive PoQ.} 
%\takashi{I feel that the term “non-trivial PoQ” is somewhat misleading, as the notion of triviality is subjective. How about "genuinely interactive PoQ", for example? (ChatGPT also suggested  "round-irreducible PoQ" or "non-degenerate PoQ".)
%Also, shouldn't we replace "are" with "imply"? 
%Also, I'm wondering whether we should add "constant-round" for accuracy. (But that may be too long?)
%}\minki{I agreed. Perhaps ``Constant-round interactive PoQ imply PoQM''? (though it also ambiguously uses interactive as non-non-interactive)}
We show that constant-round PoQ imply PoQM or single-round\footnote{In this paper, a single round is a set of two messages, first one being from the verifier to the prover and the second one from the prover to the verifier.} quantum-verifier PoQ (QV-PoQ). Here, QV-PoQ are PoQ with a QPT verifier, while the communication remains classical.
Single-round (QV-)PoQ are ``trivial'' PoQ in the sense that the verifier asks the prover to solve
a classical problem which is
quantumly easy but classically hard. 
Therefore, the result
demonstrates that
PoQM capture ``genuinely-interactive'' PoQ.
%The result also answers the following question: what is a cryptographic application of PoQ?
%\cite{STOC:MSY25} showed that PoQ imply classically-secure OWPuzzs. However, no \emph{quantumly-secure} cryptographic primitive that is implied by PoQ
%was known.\footnote{Recently, \cite{TomZha25} showed that (constant-round) PoQ imply single-round PoQ, OWFs, or signature tokens.} 
%The result shows that (constant-round) PoQ has the quantumly-secure cryptographic application, namely, PoQM (unless they imply trivial QV-PoQ). 
%\takashi{This is not quite true since QV-PoQ is not PoQ.}\mor{QV wo tsuketa. How about this?}

\item  
{\bf Constant-round IV-PoQ with black-box reduction imply IV-PoQM.}
%\takashi{I'm not sure why having a black-box reduction is regarded as "meaningful".
%Can't we simply say that "Constant-round IV-PoQ with a black-box reduction imply IV-PoQM"?}
%\cite{C:MorYam24} introduced
%inefficient-verifier PoQ (IV-PoQ), which are a variant of PoQ where the verifier's final computation to make the decision
%is not necessarily efficient.
%In a similar way, we introduce inefficient-verifier PoQM (IV-PoQM) where the verifier's final computation is not necessarily efficient. 
We show that if constant-round IV-PoQ\footnote{Inefficient-verifier PoQ (IV-PoQ)~\cite{C:MorYam24} is a variant of PoQ where the verifier's final computation to make the decision is not necessarily efficient.} that can be black-box constructed from quantumly-secure falsifiable assumptions\footnote{In this paper, we define falsifiable assumptions in such a way that the challenger is not necessarily efficient.} exist
then IV-PoQM exist. 
The assumption is natural, because
many known constructions of (IV-)PoQ are the black-box ones from quantumly-secure falsifiable assumptions.
%This result has the following conceptually interesting implication:
%When we construct an (IV-)PoQ from a falsifiable assumption, the assumption should be quantumly-hard one, because
%if it is quantumly easy, the assumption itself is already an (IV-)PoQ and therefore there is no point in making
%a new (IV-)PoQ from the assumption.
%The above result therefore suggests that non-trivial\mor{``meaningful'' wo non-trivial ni kaeta. How about this?} constructions of IV-PoQ already construct IV-PoQM.
%The result also leads to a new feasibility result of IV-PoQ:
In particular,
\cite{C:MorYam24} constructed constant-round IV-PoQ from statistically-hiding commitments in a black-box way.
Therefore,
the result also implies that IV-PoQM can be constructed from quantumly-secure constant-round statistically-hiding commitments
(and therefore from quantumly-secure collision-resistant hash functions).
\end{enumerate}

In the following, we provide more details.

\paragraph{\bf Formal definition of PoQM.}
We formally define PoQM as follows.\footnote{Our definition is based on (classical) proofs of space~\cite{C:DFKP15} and 
quantum proofs of space~\cite{Maxfield}. For details, see~\cref{sec:related}.} (Also see \cref{fig:PoQM}.) 
Let $\alpha,\beta:\mathbb{N}\to[0,1]$ be any functions. Let $m_1,m_2:\mathbb{N}\to\mathbb{N}$ be any (polynomially-bounded\footnote{This means $m_1,m_2=O(\secp^c)$ for some constant $c>0$. This condition is occasionally omitted if it is clear from the context.}) functions such that $m_1>m_2$.
An $(\alpha,\beta,m_1,m_2)$-PoQM 
 is a tuple $(\cV_1,\cP_1,\cV_2,\cP_2)$ of 
interactive algorithms over a classical channel.
$\cV_1$ and $\cV_2$ are PPT, and $\cP_1$ and $\cP_2$ are QPT.
The interaction consists of two phases, the initialization phase and the execution phase. 
In the initialization phase, both $\cV_1$ and $\cP_1$ take the security parameter $1^\secp$ as input,
and interact over a classical channel. $\cV_1$ outputs a classical bit string $v$,
and $\cP_1$ outputs a classical bit string $\state$ and an $m_1$-qubit quantum state $\sigma_\state$.\footnote{During the operation of the initialization phase, the prover can use more than $m_1$ qubits.}

In the execution phase,
$\cV_2$ takes $v$ as input, and $\cP_2$ takes $(\state,\sigma_\state)$ as input.
They interact over a classical channel, and $\cV_2$ outputs $\top$ or $\bot$.
$\alpha$-completeness requires that $\cV_2$ outputs $\top$ with probability at least $\alpha$, 
that is, the honest prover with an $m_1$-qubit memory is accepted with high probability.
On the other hand, $(\beta,m_2)$-soundness is defined as follows. 
Let $\cP_1^*$ be a QPT algorithm that interacts with $\cV_1$, and
outputs a classical bit string $s$ and an $m_2$-qubit state $\rho$. 
Let $\cP_2^*$ be a QPT algorithm that takes $(s,\rho)$ as input, and interacts with $\cV_2$. 
Then for any such $(\cP_1^*,\cP_2^*)$, $\cV_2$ outputs $\top$ with probability at most $\beta$.
Note that we do not make any upperbound for the size of the classical bit string $s$: it can be any polynomial length. 
Moreover, during the computation, $\cP^*_1$ and $\cP_2^*$ can use any polynomially-large quantum memory.

This intuitively means that any malicious prover that can preserve at most $m_2$-qubit quantum memory
during the interval between the initialization phase and the execution phase
cannot be accepted by the verifier. In other words, if the verifier accepts,
the verifier can verify that the prover has possessed at least $(m_2+1)$-qubit quantum memory
during the interval between the initialization phase and the execution phase.
%Moreover, the verifier can enforce a time delay $\Delta$ between two phases in order to 
%ensure that the successful verification implies the coherent time $\Delta$ of the prover's device.

%\minki{It seems $m$-PoQM verifies $(m+1)$-qubit, right? I think it is more natural that $m$-PoQM verifies $m$-qubit. }
%\mor{You are right, but what is in our mind is our protocol can verify that "$m$-qubit is not enough". I think the current definition is simpler.}
%\minki{I see, I think your (original, ie without $m_1$) definition is clearer as $(m+1)$ qubits are usually not sufficient for honest protocols. As long as we make it clear, I agree to use the original definition. Another advantage of the original definition shows why ours are different from certifying a qubit below.}

%eprint

\usetikzlibrary{positioning} % for position relative to node
\usetikzlibrary{calc} % for computing coordinates
\usetikzlibrary {quotes}
\usetikzlibrary{intersections,arrows.meta}
\tikzset{>=latex} % set default arrow head as latex

% TIKZ STUFF
\tikzstyle{mysmallarrow}=[->,black,line width=1.6]
\tikzstyle{mydbarrow}=[<->,black,line width=1.6]
\tikzstyle{newarrow}=[<->,red,line width=1.6]
\tikzstyle{newsinglearrow}=[->,red,line width=1.6]
\tikzstyle{carrow}=[<->,red,line width=1.6]
% JET CATEGORIES with straight lines
\begin{figure}
\begin{center}
\begin{tikzpicture}[scale=0.95, every edge quotes/.style = {font=\footnotesize, fill=white}]
  \def\W{7.5}     
  \def\Vgap{0.9}  
  \def\RowY{0.0}  
  \def\RowYB{-4.3}

  \node[font=\normalsize\bfseries, anchor=south west] at (-\W/0.85, \RowYB + 3.2) {Initialization Phase};
\node[font=\normalsize\bfseries, anchor=south west] at (-\W/0.85, \RowYB - 0.2) {Execution Phase};

  \draw[dashed,thick, name path=endini] (-\W/0.9, \RowYB+2.1) -- (\W/1.2, \RowYB+2.1);
  \draw[dashed,thick, name path=startexe] (-\W/0.9, \RowYB+1.4) -- (\W/1.2, \RowYB+1.4);

  \node[align=center] (V1) at (-\W/2,\RowY) {$\cV_1$\\(PPT)};
  \node[align=center] (P1) at (\W/2,\RowY) { $\cP_1$\\(QPT)};
%  \node[align=center] (V1) at (-\W/2,\RowY) {PPT \\ $\cV_1$};
%  \node[align=center] (P1) at (\W/2,\RowY) {QPT \\ $\cP_1$};

  \draw[mydbarrow] (V1) -- node[above]{Classical communication} (P1);

  \node (v)   at ($(V1.south)+(0,-\Vgap)$) {$v$};
  \node (stq) at ($(P1.south)+(0,-\Vgap)$) {$(\state,\sigma_{\state})$};
  \draw[mysmallarrow] (V1.south) -- (v);      
  \draw[mysmallarrow] (P1.south) -- (stq);

  \node[align=center] (V2) at (-\W/2,\RowYB) {$\cV_2$\\(PPT)};
  \node[align=center] (P2) at (\W/2,\RowYB) { $\cP_2$\\(QPT)};
%  \node[align=center] (V2) at (-\W/2,\RowYB) {PPT \\ $\cV_2$};
%  \node[align=center] (P2) at (\W/2,\RowYB) {QPT \\ $\cP_2$};

  \draw[mydbarrow] (V2) -- node[above]{Classical communication} (P2);

  \draw[mysmallarrow] (v)   -- (V2.north);   
  \draw[mysmallarrow] (stq) -- (P2.north);   

  \node (decision) at ($(V2.south)+(0,-\Vgap/1.1)$) {$\top/\bot$};
  \draw[mysmallarrow] (V2.south) -- (decision);

  \draw[<->, thick]
(-\W/0.9 - 0.3, \RowYB+1.4) --
node[left] {$\Delta$}
(-\W/0.9 - 0.3, \RowYB+2.1);

\end{tikzpicture}
\end{center}
\caption{A PoQM consists of two phases: the initialization phase and the execution phase. 
At the end of the initialization phase, $\cV_1$ outputs a classical bit string $v$, 
and $\cP_1$ outputs a classical bit string $\state$ and an $m_1$-qubit quantum state $\sigma_\state$. 
The verifier can enforce a time delay $\Delta$ between the initialization phase and the execution phase. 
At the beginning of the execution phase, $\cV_2$ takes $v$ as input, and 
$\cP_2$ takes $(\state,\sigma_\state)$ as input. 
At the end of the execution phase, $\cV_2$ outputs $\top$ or $\bot$.}
\label{fig:PoQM}
\end{figure}

\paragraph{\bf Relation to PoQ.}
We observe that PoQM generalize the notion of PoQ:
%\minki{We can combine the following theorem and corollary to reduce the main body if needed. (ie we can just use Corollary 1.2 as it is more general)}
\if0
\begin{theorem}
Let $\alpha,\beta:\mathbb{N}\to[0,1]$ be any functions such that $\alpha(\secp)>\beta(\secp)$ for all $\secp\in\mathbb{N}$. Let $m_1:\mathbb{N}\to\mathbb{N}$ be any function. 
    If $(\alpha,\beta,m_1,0)$-PoQM exist, then $(\alpha,\beta)$-PoQ exist.
\end{theorem}
The proof is given in \cref{sec:PoQMtoPoQ}. Here, an $(\alpha,\beta)$-PoQ is a PoQ with completeness $\alpha$ and soundness $\beta$.
Because an $(\alpha,\beta,m_1,m_2)$-PoQM is trivially an $(\alpha,\beta,m_1,m_2-1)$-PoQM for any $m_2\ge1$, we also obtain the following corollary.
\fi
\begin{theorem}
Let $\alpha,\beta:\mathbb{N}\to[0,1]$ be any functions such that $\alpha(\secp)>\beta(\secp)$ for all $\secp\in\mathbb{N}$. 
Let $m_1,m_2:\mathbb{N}\to\mathbb{N}$ be any functions such that $m_1>m_2$.
    If $(\alpha,\beta,m_1,m_2)$-PoQM exist, then $(\alpha,\beta)$-PoQ exist.
\end{theorem}
%The proof is given in \cref{sec:PoQMtoPoQ}.\mor{koredato lemma3.4nosyoumei ni tobu}

\paragraph{\bf Constructions of PoQM.}
We provide two constructions of PoQM based on the hardness of LWE.
The first construction is based on the subexponential hardness of LWE.
\begin{theorem} \label{thm:subexpLWE}
Let $m_2:\mathbb{N}\to\mathbb{N}$ be any polynomially-bounded function.
Assuming the subexponential hardness of LWE, two-round $(1-\negl,\negl,m_1,m_2)$-PoQM exist with some polynomial $m_1$.\footnote{The concrete polynomial for $m_1$ depends on the choice of parameters of LWE.}
\end{theorem}
%The proof is given in \cref{sec:12puztoPoQM}.

The second construction is based on the polynomial hardness of LWE.
\begin{theorem} \label{thm:polyLWE}
Let $p$ be any polynomial. Let $m_2:\mathbb{N}\to\mathbb{N}$ be any polynomially-bounded function such that $m_2(\secp)=\omega(\log(\secp))$.
Assuming the polynomial hardness of LWE, $r$-round $(1-\negl,1/p,\lceil 9.1m_2\rceil,m_2)$-PoQM exist with a certain polynomial $r$.
\end{theorem}
%The proof is given in \cref{sec:appRSPtoPoQM}. 
These two results are incomparable. 
The first construction is of two-round and with negligible soundness, while the assumption, the subexponential hardness of LWE, is stronger.
On the other hand, the second construction is based on the polynomial hardness of LWE, but it is of $\poly$-round and soundness is only $1/\poly$.\footnote{Parallel repetitions are non-trivial in PoQM, and we do not know how to do it.}

% The other construction is from remote state preparations (RSPs). 
% \begin{theorem}
%     If RSPs exist, then $\poly$-round PoQM exist.
% \end{theorem}

% An RSP \cite{FOCS:GheVid19} is an interactive protocol between a PPT verifier and a QPT prover over a classical channel. At the end of the interaction, the prover has a quantum state. The verifier has information about the quantum state, but the QPT prover does not have it. We use verifiable RSPs, which allow the verifier to confirm that the prover’s quantum state is close to the desired state.

% Because RSPs are constructed from the standard hardness of LWE, we obtain the following corollary.
% \begin{corollary}
%     Assuming standard hardness of LWE, $\poly$-round $(1-\negl,1/\poly,\poly)$-PoQM exist.
% \end{corollary}

\paragraph{\bf Lowerbounds of (IV-)PoQM.}
We show that one-way puzzles (OWPuzzs)~\cite{STOC:KhuTom24}\footnote{In this paper, if we just say OWPuzzs, they are quantumly-secure ones. We explicitly add ``classically-secure'' if they are only classically-secure.} are a lowerbound 
of IV-PoQM (and therefore PoQM):
\begin{theorem}
\label{thm:OWPuzz}
Let $\alpha,\beta:\mathbb{N}\to[0,1]$ be any functions such that $\alpha(\secp)-\beta(\secp)\ge1/\poly(\secp)$ for all sufficiently large $\secp\in\mathbb{N}$.
Let $m_1,m_2:\mathbb{N}\to\mathbb{N}$ be any functions such that $m_1>m_2$.
If $(\alpha,\beta,m_1,m_2)$-IV-PoQM exist, then OWPuzzs exist. 
\end{theorem}
%The proof is given in \cref{sec:proof_PoQMtoStatepuz}. 
OWPuzzs are a natural quantum analogue of OWFs. A OWPuzz is a pair $(\Samp,\Ver)$ of algorithms.
$\Samp$ is a QPT algorithm that takes the security parameter $1^\secp$ as input and outputs classical bit strings $\puzz$ and $\ans$.
$\Ver$ is an unbounded algorithm that takes $(\puzz,\ans)$ as input and outputs $\top$ or $\bot$.
Correctness is 
\begin{align}
\Pr\bigl[\top\gets\Ver(\puzz,\ans):(\puzz,\ans)\gets\Samp(1^\secp)\bigr]\ge1-\negl(\secp)
\end{align}
and security is
\begin{align}
\Pr\bigl[\top\gets\Ver(\puzz,\ans'):(\puzz,\ans)\gets\Samp(1^\secp),\allowbreak \ans' \gets \cA(1^\secp,\puzz)\bigr]\le\negl(\secp)
\end{align}
for any QPT adversary $\cA$.
OWPuzzs are implied by many quantum cryptographic primitives, and imply non-interactive bit commitments
and multiparty computations 
\cite{STOC:KhuTom24,C:MorYam22,AC:Yan22,C:AnaQiaYue22}.

\if0
\minki{The above is the lower bound of PoQ, right? What about the below one? (I'll be surprised if there's a lower bound of PoQM stronger than PoQ, as it somehow use QM different from quantumness.)}
\mor{I do not understand this comment.}
\minki{What I'm saying is that all lower bounds are the lower bounds for $(\alpha,\beta,m_1,0)$-PoQM, i.e. $m_2=0$. This is very natural given the other (quantum) primitives usually do not care about quantum memory bound. My question is: Can we either obtain some interesting primitives from $(\alpha,\beta,m_1,m_2)$-PoQM for some $m_2>0$ (that is not obtained from $m_2=0$) or separate $(\alpha,\beta,m_1,m_2)$-PoQM and other interesting primitive?}
\fi

\cite{STOC:MSY25} showed that PoQ imply classically-secure OWPuzzs.
An interesting open question is whether any quantumly-secure quantum cryptographic primitive is implied by PoQ.
\cref{thm:OWPuzz} does not solve the open problem, but gives a ``partial answer'' in the sense that
a generalization (i.e., PoQM) of PoQ implies the quantumly-secure quantum cryptographic primitive, namely, OWPuzzs.

If we consider a restricted version of PoQM (which we call \emph{predictable PoQM}), 
then we obtain 
a potentially stronger\footnote{\cite{cryptoeprint:2024/1707} showed that QCCC KE imply OWPuzzs.} lowerbound:
\begin{theorem}
\label{thm:KE}
    Let $m_1:\mathbb{N}\to\mathbb{N}$ be any function.  Let $p$ be any polynomial. If $(1-\negl,1-\frac{1}{p},m_1,0)$-predictable PoQM exist, then QCCC KE exist.
\end{theorem}
%\begin{theorem}
%\label{thm:KE}
%Let $m_1,m_2:\mathbb{N}\to\mathbb{N}$ be any functions. Let $\alpha:\mathbb{N}\to[0,1]$ be any function. Let $c_1$ and $c_2$ be any constants such that $c_1>c_2>0$. Let $p(\secp)\coloneqq\secp^{c_1}$ and $q(\secp)\coloneqq\secp^{c_2}$. If $(\alpha,\alpha-\frac{1}{q},m_1,m_2)$-predictable PoQM with extraction probability $1-\frac{1}{p}$ exist, then quantum computation classical communication key-exchange (QCCC KE) exist.   
%\end{theorem} 
%The proof is given in \cref{sec:urPoQMtoKE}.\mor{App no basyo syuusei}
Here, an $(\alpha,\beta,m_1,m_2)$-predictable PoQM  
is an $(\alpha,\beta,m_1,m_2)$-PoQM where the execution phase is of single round, and 
the final verification algorithm is in the following special form: $\cV_2$ computes a bit string by itself
and accepts if and only if the bit string is equal to the bit string sent from $\cP_2$.
Our construction of PoQM based on the polynomial hardness of LWE satisfies this property.
\if0
Here, an $(\alpha,\beta,m_1,m_2)$-predictable PoQM with extraction probability $\gamma$ 
is an $(\alpha,\beta,m_1,m_2)$-PoQM where the execution phase is of single round (i.e., of two message), and 
$\cP_2$'s message in the execution
phase can be
computed in QPT given the output of $\cV_1$ with probability at least $\gamma$.
Our construction of PoQM based on the polynomial hardness of LWE satisfies this property.
\fi
A QCCC KE is a key exchange in the quantum computation and classical communication (QCCC) setting, i.e., 
Alice and Bob are QPT but all communications are classical.

\if0
These two lowerbounds provide interesting insights to the following two open problems about PoQ:
\begin{enumerate}
\item 
Is there any 
quantum cryptographic lowerbound for PoQ?
\item 
Can PoQ be constructed from OWFs?
%If PoQ are shown to imply KE, it suggests that PoQ will not be constructed from OWFs in a black-box way, because...\mor{cite}.
\end{enumerate}
A reason why we are interested in the first open problem is the following.
\cite{STOC:MSY25} showed that PoQ imply classically-secure OWPuzzs. However, classically-secure quantum cryptographic
primitives are not ``true'' quantum cryptographic primitives. We therefore want to know whether
quantumly-secure cryptographic primitive are obtained from PoQ.
However, it is not a good question, because the question is slightly 

Although we do not know how to solve the first open problem, 
we could consider an ``easier'' version of the problem: Is there any quantum cryptographic lowerbound for 
a generalization of PoQ, namely, PoQM?
\cref{thm:OWPuzz}
shows that the answer is yes.
\fi

PoQ have been constructed from several assumptions including the hardness of LWE~\cite{JACM:BCMVV21}, 
noisy-trapdoor claw-free functions~\cite{JACM:BCMVV21},
trapdoor two-to-one collision resistant hash functions~\cite{NatPhys:KMCVY22},
full-domain trapdoor permutations~\cite{ITCS:MorYam23}, 
iO plus OWFs~\cite{C:ShmZha25}, etc.
However, these assumptions are relatively stronger, 
and an important open problem is whether PoQ can be constructed from much weaker
assumptions, for example, OWFs.\footnote{
Tomer and Zhandry~\cite{STOC:TomZha26} recently showed that constant-round PoQ will not be constructed from iO plus OWPs in a black-box way.}
Given the fact that PoQM are a generalization of PoQ, it is also important to ask
whether PoQM can be constructed from OWFs.
\cref{thm:KE}
gives a negative evidence: at least the restricted version of PoQM
will not be constructed from OWFs in a black-box way, because there is evidence that
QCCC KE will not be constructed from OWFs in a black-box way~\cite{ITCS:LLLL25,C:ACCFLM22,C:LLLL24}.\footnote{Note that, however, 
these impossibility results do not exclude the most general QCCC KE, because
the results are only for the perfectly-correct case.}

\paragraph{\bf PoQM capture genuinely-interactive PoQ.} 
Single-round (QV-)PoQ can be considered as ``trivial'' ones, because the verifier sends a classical message to the prover, which is 
a quantumly-easy but classically-hard problem, and asks the prover to solve it.
(For example, sending an integer to the prover and asking the prover to return its factor.)
On the other hand, there are ``genuinely-interactive'' PoQ, such as those that use the rewinding~\cite{JACM:BCMVV21}.
The formal definition of PoQ just includes these two types of PoQ,
and therefore the definition cannot distinguish them.
Can we capture only genuinely-interactive PoQ?

We show that PoQM somehow capture such genuinely-interactive PoQ.
We show the following theorem, which says that $\ell$-round PoQ imply $\ell$-round PoQM,
$(\ell-1)$-round PoQ, or single-round quantum-verifier PoQ (QV-PoQ).
\begin{theorem}
\label{thm:PoQandPoQM}
Let $\ell:\mathbb{N}\to\mathbb{N}$ be any polynomial.
Let $\alpha,\beta:\mathbb{N}\to[0,1]$ be any functions such that $\alpha(\secp)-\beta(\secp)\ge\frac{1}{\poly(\secp)}$ for all $\secp\in\mathbb{N}$.  
Let $\epsilon,\delta:\mathbb{N}\to[0,1]$ be any functions such that 
$\epsilon(\secp)\ge\frac{1}{\poly(\secp)}$, $\delta(\secp)\ge\frac{1}{\poly(\secp)}$
and
$\left(\alpha(\secp)-\epsilon(\secp)\right)-\left(\beta(\secp)+\delta(\secp)\right)\ge\frac{1}{\poly(\secp)}$ for all $\secp\in\mathbb{N}$.
If $\ell$-round $(\alpha,\beta)$-PoQ exist, 
then
one of the following primitives exists:
\begin{itemize}
\item 
$\ell$-round $(\alpha,\alpha-\epsilon,\poly,0)$-PoQM (with i.o.-soundness).
%\item 
%$(\ell-1)$-round $(c-\epsilon,s)$-PoQ exist.
\item
single-round $(\alpha-\epsilon,\beta+\delta)$-QV-PoQ (with i.o.-soundness).
\item 
$(\ell-1)$-round $(\beta+\delta,\beta)$-PoQ.
\end{itemize}
\end{theorem}
%The proof is given in \cref{sec:nontrivialPoQ}.
By repeating this theorem, we obtain the following corollary.\footnote{We can repeat the theorem only constant times, because of some technical 
reasons, and therefore the corollary holds only for the constant-round case. For details, see \cref{remark:Kabir}.}

\begin{corollary} \label{coro:PoQandPoQM}
If constant-round PoQ exist, then one of the following primitives exists.
\begin{itemize}
    \item 
constant-round PoQM with $m_2=0$ (with i.o.-soundness).
\item 
single-round QV-PoQ (with i.o.-soundness).
%\item 
%single-round PoQ.
\end{itemize}
%All primitives in this corollary have an inverse-polynomial completeness-soundness gap.
%\takashi{I still think that the third item is redundant and rather confusing since single-round PoQ implies single-round QV-PoQ with i.o. soundness. }
\end{corollary}
%The proof is given in \cref{sec:nontrivialPoQ}.
This corollary demonstrates
that PoQM capture ``genuinely-interactive'' PoQ.

We also note that the proofs of \cref{thm:PoQandPoQM,coro:PoQandPoQM}
also work for the inefficient verifier cases, and therefore we obtain similar results
for the IV setting.

\paragraph{\bf Constant-round IV-PoQ with black-box reduction imply IV-PoQM.}
Cryptographic primitives are constructed from assumptions. 
In particular, we often focus on falsifiable assumptions~\cite{C:naor03,STOC:GenWic11}.
A falsifiable assumption is a pair $(\cC,\mathsf{thres})$ of a (potentially-inefficient) challenger $\cC$ and a constant $\mathsf{thres}\in[0,1)$. 
The challenger interacts with an adversary $\cA$ and outputs $\top$ or $\bot$. 
We say that an adversary $\cA$ breaks the assumption if there exists a polynomial 
$p$ such that 
        $\Pr[\top\gets\langle \cC(1^\secp),\cA(1^\secp)\rangle]\ge\mathsf{thres}+\frac{1}{p(\secp)}$
for infinitely many $\secp\in\mathbb{N}$.
Falsifiable assumptions capture a wide class of cryptographic assumptions including 
OWFs, public-key encryption, oblivious transfers, the hardness of LWE, indistinguishability obfuscation, etc.

We say that an (IV-)PoQ is black-box constructed from a falsifiable assumption 
if there exists a QPT reduction algorithm $\cR$ such that for any (not-necessarily-efficient)
classical adversary $\cP'$ that breaks soundness of the (IV-)PoQ, 
$\cR^{\cP'}$ breaks the assumption.

%However, PoQ are sound only against PPT provers, whereas PoQM are sound against QPT provers. 
%To our knowledge, there is no known generic way to derive quantumly-secure primitives from classically-secure primitives. This leaves a gap between classical and quantum security guarantees. To address this gap, Tomer and Zhandry \cite{TomZha25} considered PoQ with black-box reduction to a quantumly-secure falsifiable assumption and showed that they imply OWPuzzs. Here, falsifiable assumptions are modeled as an interactive game between a (potentially inefficient) challenger and adversary, and, as explained in \cite{STOC:GenWic11,C:MorYam24}, the definition of falsifiable assumptions captures a wide class of assumptions such as one-way functions, public-key encryption, oblivious transfers, the hardness of LWE, indistinguishability obfuscation etc. 

We show the following result.\footnote{Showing a similar result for the non IV case is open.}
\begin{theorem} \label{thm:introIVPoQM}
Let $\ell$ be any polynomial. Let $\alpha,\beta:\mathbb{N}\to[0,1]$ be any functions such that $\alpha(\secp)-\beta(\secp)\ge\frac{1}{\poly(\secp)}$ for all $\secp\in\mathbb{N}$. Let $\epsilon:\mathbb{N}\to[0,1]$ be any function such that $\epsilon(\secp)\ge\frac{1}{\poly(\secp)}$ and $\left(\alpha(\secp)-\epsilon(\secp)\right)-\beta(\secp)\ge\frac{1}{\poly(\secp)}$ for all $\secp\in\mathbb{N}$.  Suppose that there exist $\ell$-round $(\alpha,\beta)$-IV-PoQ that are black-box constructed from a falsifiable assumption.
Then one of the following exists. 
\begin{itemize}
   \item $\ell$-round $(\alpha,\alpha-\epsilon,\poly,0)$-IV-PoQM (with i.o.-soundness).
   \item $(\ell-1)$-round $(\alpha-\epsilon,\beta)$-IV-PoQ that are black-box constructed from the falsifiable assumption. 
\end{itemize}
\end{theorem}

Repeatedly applying \cref{thm:introIVPoQM}, we have the following corollary.\footnote{Again, we can show the corollary only for the constant-round case, because we can repeat \cref{thm:introIVPoQM} only constant times.} 
\begin{corollary} \label{coro:introivPoQM}
If there exist constant-round IV-PoQ that are black-box constructed from a falsifiable assumption, 
then one of the following exists.
\begin{itemize}
   \item constant-round IV-PoQM (with i.o.-soundness).
   \item single-round IV-PoQ that are black-box constructed from the falsifiable assumption.  
\end{itemize}
\end{corollary}

However, \cite{STOC:TomZha26} showed that single-round IV-PoQ cannot be black-box reduced to a quantumly-secure 
falsifiable assumption. 
Therefore, we obtain the following corollary:
\begin{corollary} \label{coro:introupgrade}
If there exist constant-round IV-PoQ that are black-box constructed from a quantumly-secure falsifiable assumption, 
then there exist constant-round IV-PoQM (with i.o.-soundness).
\end{corollary}
\if0
This result has the following conceptually interesting implication:
When we construct an (IV-)PoQ from a falsifiable assumption, the assumption should be quantumly-hard one, because
if it is quantumly easy, the assumption itself is already an (IV-)PoQ and therefore there is no point in making
a new (IV-)PoQ from the assumption.
\cref{coro:introupgrade} therefore demonstrates that ``meaningful'' constructions of IV-PoQ already imply IV-PoQM.
\fi

This corollary leads to a new feasibility result of IV-PoQM: 
\cite{C:MorYam24} constructed constant-round IV-PoQ from classically-secure constant-round statistically-hiding commitments in a black-box way,
and the proof holds even for quantumly-secure commitments. 
Therefore \cref{coro:introupgrade} implies the following.
\begin{corollary}
    If quantumly-secure constant-round statistically-hiding commitments exist,  then there exist constant-round IV-PoQM (with i.o.-soundness). 
\end{corollary}

Furthermore, constant-round statistically-hiding commitments can be obtained from 
collision-resistant hash functions \cite{C:HalMic96}, thus we also obtain the following corollary.
\begin{corollary}
If quantumly-secure collision-resistant hash functions exist,  then there exist constant-round IV-PoQM (with i.o.-soundness).      
\end{corollary}

\if0
The definition of PoQ is overly broad and contains non-interactive one like factoring \cite{FOCS:Shor94} and Yamakawa-Zhandry \cite{FOCS:YamZha22}. This makes it difficult to understand its relationship with other primitives. Therefore, rewinding-type PoQ should be separated from PoQ, and in some sense, the fact that PoQM imply StatePuzzs suggests that PoQM achieves the separation.
\fi

\subsection{Technical Overview}
Here we provide a high-level overview of our proofs.

\if0
\paragraph{Amplification of PoQM.}
For our constructions, we need the following amplification lemma:
\begin{lemma}
\label{lem:amp}
An $(\alpha,\beta,0)$-PoQM is an $(\alpha,2^m\beta,m)$-PoQM. 
\end{lemma}
To show it, we use a lemma in \cite{C:BonZha13}, 
which says the following: let $\cA$ be any algorithm and $p$ be the probability that $\cA$ accepts.
Let $\cA'$ be an algorithm that is the same as $\cA$ except that a $k$-outcome measurement is done at any point in the algorithm.
Let $p'$ be the probability that $\cA'$ accepts.
Then, $p\le kp'$. By using the lemma of \cite{C:BonZha13}, we can show \cref{lem:amp} 
as follows.
If a PoQM is not $(2^m\beta,m)$-sound, then due to this lemma, there exist QPT adversaries that can hold $0$-qubit and make the verifier accept with probability at least $\beta$. However, this violates the $(\beta,0)$-soundness.  
\fi

\paragraph{\bf PoQM based on the polynomial hardness of LWE. (\cref{thm:polyLWE})}
Let us first explain our construction of 
$(1-\negl,1/p,\lceil 9.1m_2\rceil,m_2)$-PoQM for any polynomial $p$ and 
any polynomially-bounded function $m_2:\mathbb{N}\to\mathbb{N}$ such that $m_2(\secp)=\omega(\log(\secp))$
based on the polynomial hardness of LWE. 
The basic idea of our construction is simple:
First, let us consider the following ``information-theoretically-secure'' $(1-\negl,\negl,n,0)$-PoQM~\cite{Maxfield}:
\begin{itemize}
    \item
{\it Initialization phase.}
$\cV_1$ generates the state $\sigma\coloneqq\bigotimes_{i=1}^{n} H^{\theta_i}|x_i\rangle$ with random $(x,\theta)\in\bit^n\times\bit^n$ by itself 
and sends the state to $\cP_1$ over a quantum channel. 
(Here, $x_i$ and $\theta_i$ are the $i$th bit of $x$ and $\theta$, respectively. $H$ is the Hadamard operator.)
$\cV_1$ outputs $(x,\theta)$.
$\cP_1$ outputs $\sigma$.
\item
{\it Execution phase.}
$\cV_2$ takes $(x,\theta)$ as input.
$\cP_2$ takes $\sigma$ as input.
$\cV_2$ sends $\theta$ to $\cP_2$, and $\cP_2$ measures $i$th qubit of $\sigma$ in
the computational (Hadamard) basis if $\theta_i=0$ ($\theta_i=1)$ for all $i\in[n]$.
Let $x_i'$ be the measurement result on the $i$th qubit. $\cP_2$ sends $\{x_i'\}_{i\in[n]}$ to $\cV_2$.
If $x_i=x_i'$ for all $i\in[n]$, $\cV_2$ accepts.
Otherwise, $\cV_2$ rejects.
\end{itemize}
This information-theoretically-secure PoQM does not achieve our goal, because of the following two reasons:
\begin{enumerate}
    \item 
    It is only the case when $m_2=0$. We want to construct $(1-\negl,1/\poly,\allowbreak\lceil9.1m_2\rceil,m_2)$-PoQM for any polynomial $m_2$.
    \item 
    Both the verifier and the channel are quantum.
\end{enumerate}

The first issue is solved 
by using a lemma of \cite{C:BonZha13}.
The lemma says the following: 
Let $\cA$ be a quantum algorithm that outputs a classical bit string.
Let $\cA'$ be the algorithm that is the same as $\cA$ except that a $k$-outcome measurement is done at any step.
Then, $\Pr[x\gets\cA]\le k\Pr[x\gets\cA']$ for any $x$.\footnote{$\Pr[x\gets\cA]$ is the probability that $\cA$ outputs $x$.}
Using this lemma, we can show that an $(\alpha,\beta,m_1,0)$-PoQM is an $(\alpha,2^{m_2}\beta,m_1,m_2)$-PoQM
for any $\alpha,\beta,m_1,m_2$ as follows:
Let $(\cV_1,\cP_1,\cV_2,\cP_2)$ be an 
$(\alpha,\beta,m_1,0)$-PoQM.
Assume that it is not an
$(\alpha,2^{m_2}\beta,m_1,m_2)$-PoQM.
Then there exists a pair $(\cP_1^*,\cP_2^*)$ of QPT adversaries such that
$\cP_1^*$ outputs a classical bit string $s$ and an $m_2$-qubit state $\rho$,
$\cP_2^*$ takes $(s,\rho)$ as input,
and the probability that $\cV_2$ outputs $\top$ is strictly larger than $2^{m_2}\beta$. 
Let us define another QPT adversary $\cP_1^{**}$ as follows:
$\cP_1^{**}$ runs $\cP_1^*$ but it measures all qubits of the output state $\rho$ of $\cP_1^*$ in the computational basis,
and outputs the measurement result.
Then, by using the lemma of \cite{C:BonZha13},
the probability that $(\cP_1^{**},P_2^*)$ is accepted is strictly larger than
$2^{-m_2}\times 2^{m_2}\beta=\beta$,
which is the contradiction.

Therefore, we want to show that soundness (i.e., $\beta$) of the above information-theoretically-secure PoQM is 
$\beta=2^{-m_2}\times\negl$.\footnote{At this stage, we can achieve $\negl$-soundness, but as we will see soon,
due to the $1/\poly$-soundness of RSPs, what we finally get is only $1/\poly$-soundness.}
In order to show it,
we use the lemma, called 
\emph{the LOCC leakage property for BB84 states}, which was introduced in \cite{cryptoeprint:2024/1876} for another purpose, namely,
constructions of leakage-resilient encryption and signatures. 
This lemma shows the following. Let us consider the following security game:
\begin{enumerate}
    \item 
A challenger generates a random BB84 state $\sigma\coloneqq\bigotimes_{i=1}^n H^{\theta_i}|x_i\rangle$ with random $(x,\theta)\in\bit^n\times\bit^n$.
\item 
An adversary sends the challenger a classical description $E$ of a quantum algorithm that takes a quantum state as input and outputs a classical bit string.
\item 
The challenger runs $\eta\gets E(\sigma)$\footnote{$\eta\gets E(\sigma)$ means that the algorithm $E$ is run on input $\sigma$, and the output is $\eta$.}, and sends the classical bit string $\eta$ to the adversary.
\item 
The challenger sends $\theta$ to the adversary.
\item 
The adversary returns a bit string $x'\in\bit^n$ to the challenger.
\item 
The challenger accepts if $x=x'$, and rejects if $x\neq x'$.
\end{enumerate}
The lemma says that the probability that the challenger accepts is
at most
$ 2^{-\frac{\xi}{2}\cdot n + 2^{-n}}$,
where $\xi\coloneqq-\log(\frac{1}{2}+\frac{1}{2\sqrt{2}})>0.22$.
By considering the above $\eta$ as the output of $\cP_1^*$,
and taking $n=\lceil9.1m_2\rceil$,
we can show that the probability that the verifier accepts in 
the information-theoretically-secure PoQM is 
$2^{-\frac{\xi}{2}\lceil9.1m_2\rceil+2^{-\lceil9.1m_2\rceil}}\le 2^{-m_2}\times 2^{-0.001m_2+2^{-9.1m_2}}=2^{-m_2}\times\negl$.
(Note that $m_2(\secp)=\omega(\log(\secp))$ by assumption.)

The second issue that both the verifier and the channel are quantum in the information-theoretically-secure PoQM is solved
by using verifiable remote state preparations (RSPs)~\cite{FOCS:GheVid19,ITCS:Zhang25}. 
An RSP is a two-party protocol between a PPT sender and a QPT receiver. 
They interact over a classical channel. The receiver outputs a
random BB84 state, and the sender outputs its classical description. 
By using this, we can replace the quantum verifier and the quantum channel with
a PPT verifier and a classical channel.
Because the RSP of \cite{ITCS:Zhang25} requires poly round of communication, and it 
achieves only $1/\poly$-soundness, the final PoQM we obtain is of poly-round, and
has only $1/\poly$-soundness.\footnote{One might hope to avoid the polynomial-round overhead by replacing RSP with Oblivious State Preparation (OSP)~\cite{C:BarKhu25}. However, since our proof crucially relies on verifiability, which OSP does not provide, it is unclear how to construct PoQM from OSP.}

\paragraph{\bf PoQM from the subexponential hardness of LWE. (\cref{thm:subexpLWE})}
The above construction requires polynomial round of communication and achieves only $1/\poly$-soundness. 
To complement this, for any polynomially-bounded function $m_2:\mathbb{N}\to\mathbb{N}$, we next construct two-round $(1-\negl,\negl,m_1,m_2)$-PoQM with some polynomial $m_1$ from the subexponential hardness of LWE. 

The construction is based on
$1$-of-$2^k$ puzzles \cite{ITCS:LLQ22}. $1$-of-$2$ puzzles were first introduced in \cite{AFT:RadSat19} to study \emph{semi-quantum money},
which is a variant of quantum money that can be minted and verified classically. 
\cite{ITCS:LLQ22} extended $1$-of-$2$ puzzles to 1-of-$2^k$ puzzles to construct
classically verifiable position verification. 
%They introduced the notion of $1$-of-$2$ non-local soundness, and showed that $1$-of-$2$ puzzles with hardness defined in \cite{AFT:RadSat19} imply $1$-of-$2$ puzzles with $1$-of-$2$ non-local soundness. Moreover, such puzzles can be amplified to $1$-of-$2^k$ puzzles with $1$-of-$2^k$ non-local soundness by parallel repetition. We use $1$-of-$2^k$ puzzles with $1$-of-$2^k$ non-local soundness, and for simplicity, we refer to $1$-of-$2^k$ non-local soundness as soundness.
A $1$-of-$2^k$ puzzle consists of four algorithms $(\mathsf{KeyGen, Obligate, Solve, Ver})$. $\mathsf{KeyGen}$ is a PPT algorithm that takes the security parameter $1^\secp$ as input and outputs a public key $\mathsf{pk}$ and a secret key $\mathsf{sk}$. $\mathsf{Obligate}$ is a QPT algorithm that takes $\mathsf{pk}$ as input and outputs a bit string $y$ and a quantum state $\rho$. $\mathsf{Solve}$ is a QPT algorithm that takes $\mathsf{pk},y,\rho$ and a randomly chosen $k$-bit string $\mathsf{ch}$ as input and outputs a classical answer $\mathsf{ans}$. $\mathsf{Ver}$ is a polynomial-time classical deterministic algorithm that takes $\mathsf{sk},y,\mathsf{ch}$ and $\mathsf{ans}$ as input and outputs $\top$ or $\bot$.
 Completeness is that $\mathsf{Ver}$ outputs $\top$ with probability at least $1-\negl(\secp)$. 
 In order to define soundness, we consider the following security game between a set $(\cA,\cB,\cC)$ of adversaries and a challenger $\mathsf{Chal}$.

\begin{enumerate}
        \item $\mathsf{Chal}$ runs $(\mathsf{pk,sk})\gets\mathsf{KeyGen}(1^\secp)$.
        \item $\cA$ receives the public key $\mathsf{pk}$ and outputs a bit string $y$ and a quantum state $\sigma_{\mathbf{BC}}$ over two registers $\mathbf{B}$ and $\mathbf{C}$. 
        \item
        $\cA$ sends $y$ to $\mathsf{Chal}$. 
        $\cA$ sends the register $\mathbf{B}$ to $\cB$.
        $\cA$ sends the register $\mathbf{C}$ to $\cC$. 
%        Here, $\sigma_\mathbf{B}$ $(\sigma_\mathbf{C})$ is the quantum state on the register $\mathbf{B}$ $(\mathbf{C})$.
        \item $\mathsf{Chal}$ samples $\mathsf{ch}\gets\bit^{k(\secp)}$ and sends $\mathsf{ch}$ to both $\cB$ and $\cC$.
        \item $\cB$ outputs an answer $\mathsf{ans}_\cB$, and sends it to $\mathsf{Chal}$.
        $\cC$ outputs an answer $\mathsf{ans}_\cC$, and sends it to $\mathsf{Chal}$.
        \item $\mathsf{Chal}$ outputs $\top$ if and only if
        \begin{align}
            \mathsf{Ver}(\mathsf{sk},y,\mathsf{ch},\mathsf{ans}_\cB)=\top \land \mathsf{Ver}(\mathsf{sk},y,\mathsf{ch},\mathsf{ans}_\cC)=\top.
        \end{align}
\end{enumerate}
With this security game, we define $c$-soundness as follows: for any set $(\cA,\cB,\cC)$ of non-uniform QPT adversaries, 
\begin{align}
\Pr[\top\gets\mathsf{Chal}] \le2^{-k(\secp)}+\negl(2^{\secp^c}). 
\end{align}

We want to construct a two-round $(1-\negl,\negl,m_1,m_2)$-PoQM.
Let $c>0$ be any constant such that 
$m_2(\secp)=O(\secp^c)$.
Set $k(\secp)=\omega(\secp^c)$.
We construct a two-round $(1-\negl,\negl,m_1,m_2)$-PoQM from 1-of-$2^k$ puzzles with $c$-soundness as follows.
\begin{itemize}
    \item 
{\it Initialization phase.} 
$\cV_1$ runs $(\mathsf{pk,sk})\gets\mathsf{KeyGen}(1^\secp)$ and sends $\mathsf{pk}$ to $\cP_1$. 
$\cP_1$ runs $(y,\rho)\gets\mathsf{Obligate}(\mathsf{pk})$ and sends $y$ to $\cV_1$.
$\cV_1$ outputs $(\mathsf{sk},y)$, and $\cP_1$ outputs $(\mathsf{pk},y,\rho)$. 
\item
{\it Execution phase.} 
$\cV_2$ takes $(\mathsf{sk},y)$ as input, and $\cP_2$ takes $(\mathsf{pk},y,\rho)$ as input. 
$\cV_2$ samples random $k$-bit string $\mathsf{ch}$ and sends it to $\cP_2$. 
$\cP_2$ runs $\mathsf{ans}\gets\mathsf{Solve}(\mathsf{pk},y,\rho,\mathsf{ch})$, and sends $\mathsf{ans}$ to $\cV_2$. 
$\cV_2$ runs $\top/\bot\gets\mathsf{Ver}(\mathsf{sk},y,\mathsf{ch},\mathsf{ans})$ and outputs the output. 
\end{itemize}

Thus constructed PoQM is 
$(1-\negl,\epsilon,m_1,0)$-PoQM with $\epsilon(\secp)=(2^{-k(\secp)}+\negl(2^{\secp^c}))^\frac{1}{2}$, where $m_1(\secp)$ 
is the number of qubits of the output quantum state $\rho$ of $\mathsf{Obligate}$.
The reason is as follows.
Assume that it is not $(\epsilon,0)$-sound.\footnote{We remind the reader that $(\epsilon,0)$-soundness means 
that the soundness is $\epsilon$ and $m_2=0$. In other words, any QPT adversary that has at most $m_2$-qubit quantum memory during the interval
between the initialization phase and the execution phase is accepted with probability at most $\epsilon$.} 
Then,
there exists a pair $(\cP_1^*,\cP_2^*)$ of QPT adversaries such that $\cP_1^*$ outputs only a classical bit string $s$,
$\cP_2^*$ takes only the classical bit string as input, and $\cV_2$ outputs $\top$ with probability at least $\epsilon$.
From such $(\cP_1^*,\cP_2^*)$, we can construct a set $(\cA,\cB,\cC)$ of adversaries for the 1-of-$2^k$ puzzle whose winning probability is
strictly larger than $\epsilon^2$ as follows:
Given $\pk$, $\cA$ runs $s\gets\cP_1^*(\pk)$, sends $s$ to both $\cB$ and $\cC$. 
$\cB$ and $\cC$ run $\ans\gets\cP_2^*(s,\mathsf{ch})$ and send $\ans$ to $\mathsf{Chal}$, respectively.
Because $\epsilon^2(\secp)=((2^{-k(\secp)}+\negl(2^{\secp^c}))^\frac{1}{2})^2= 2^{-k(\secp)}+\negl(2^{\secp^c})$, $c$-soundness is broken.

Assuming the subexponential hardness of LWE, for any constant $c>0$ and for any polynomial $k$, there exist $1$-of-$2^k$ puzzles with $c$-soundness~\cite{ITCS:LLQ22}. Therefore, assuming the subexponential hardness of LWE, for any $c>0$, we obtain $(1-\negl,\epsilon,m_1,0)$-PoQM, where $\epsilon(\secp)=\negl(2^{\secp^c})$.

Finally, 
by using the lemma of \cite{C:BonZha13},
thus constructed $(1-\negl,\epsilon,m_1,0)$-PoQM is
$(1-\negl,2^{m_2}\epsilon,m_1,m_2)$-PoQM.
Because $2^{m_2(\secp)}\epsilon(\secp)=2^{O(\secp^c)}\negl(2^{\secp^c})\allowbreak=\negl(2^{\secp^c})=\negl(\secp)$, we finally obtain
$(1-\negl,\negl,m_1,m_2)$-PoQM.

\begin{remark}
    Classical-communication public-key encryption with secure key leasing (PKE-SKL)~\cite{EC:CGJL25,AC:PWYZ24} also implies PoQM with $m_2=0$. However, the known construction of classical-communication PKE-SKL relies on the LWE assumption, and therefore the resulting PoQM still requires the LWE assumption. Moreover, if we want to amplify $m_2$ to an arbitrary polynomial, we need the subexponential hardness of LWE to withstand the exponential security loss. Thus, deriving PoQM from classical-communication PKE-SKL does not provide any advantage over the above construction.
\end{remark}

\paragraph{\bf IV-PoQM imply OWPuzzs. (\cref{thm:OWPuzz})}
As a lowerbound of IV-PoQM, we show that IV-PoQM imply OWPuzzs. 
Because OWPuzzs are existentially equivalent to state puzzles (StatePuzzs)~\cite{STOC:KhuTom25}, we actually construct
StatePuzzs from IV-PoQM.
A StatePuzz is a QPT algorithm $\mathsf{Samp}$ that takes $1^\secp$ as input and outputs a pair $(s,\ket{\psi_s})$ of 
a bit string $s$ and a pure quantum state $\ket{\psi_s}$. The security is that given $s$, no QPT algorithm can output
a quantum state that is close to $\ket{\psi_s}$.
Our construction 
of StatePuzzs from IV-PoQM is as follows:
The classical puzzle $s$ of the StatePuzz is the classical output $\state$ of $\cP_1$ 
and the transcript $\tau$ of the initialization phase.
The quantum answer $\ket{\psi_s}$ of the StatePuzz is the (purified) quantum output $\sigma_\state$ of $\cP_1$.
Intuitively, if thus constructed StatePuzz is not secure, there exists a QPT adversary $\cA$ that, given $s=(\state,\tau)$, can output a quantum state
$\rho$ that is close to $\sigma_\state$. Then, the following adversary $(\cP_1^*,\cP_2^*)$ can break the soundness of the IV-PoQM:
$\cP_1^*$ outputs $s=(\state,\tau)$. $\cP_2^*$ takes $s$ as input, 
runs $\rho\gets\cA(s)$, and runs $\cP_2$ on $(\state,\rho)$.

\paragraph{\bf Predictable PoQM imply QCCC KE. (\cref{thm:KE})}
We also show that predictable PoQM imply QCCC KE. 
Our construction of QCCC KE is as follows:
Alice and Bob run the initialization phase of the predictable PoQM: 
Alice runs $\cP_1$ and Bob runs $\cV_1$.
Alice and Bob next run the execution phase of the predictable PoQM: 
Alice runs $\cP_2$, does not send the output to Bob, and uses the output as her key.
Bob runs $\cV_2$, computes the output of Alice, and uses the output as his key.
Because of the completeness of the predictable PoQM, their keys are equal with high probability,
which shows the correctness of the KE.
To show the security, assume that there is a QPT adversary Eve who, given the transcript of the interaction between Alice and Bob,
can compute Alice's key. By using such Eve, we can construct an adversary $(\cP_1^*,\cP_2^*)$ that breaks soundness of the predictable PoQM as follows:
$\cP_1^*$ runs $\cP_1$ and outputs the transcript of the initialization phase $\tau$. 
$\cP_2^*$ takes $\tau$ as input, runs Eve on input $\tau$ and $\cV_2$'s message, and outputs the output of Eve.

\paragraph{\bf $\ell$-round PoQ imply $(\ell-1)$-round PoQ, single-round QV-PoQ or $\ell$-round PoQM. (\cref{thm:PoQandPoQM})} 
Assume that an $\ell$-round PoQ $\langle\cV,\cP\rangle$ exists.
%Without loss of generality, in each round, $\cV$ first sends a message to $\cP$, and then $\cP$ returns a message to $\cV$.
We consider the following two cases:
\begin{itemize}
    \item 
There exists a QPT prover $\cP^*$ that is accepted by $\cV$ with high probability and its internal state after the end of the $(\ell-1)$th round
is classical.
\item 
There is no such prover.
\end{itemize}
The second case immediately means the existence of PoQM by considering the first $(\ell-1)$ rounds as the initialization phase
and the last round as the execution phase.

Therefore, let us assume that the first case happens.
Then, by using $\cP^*$, we construct the following single-round QV-PoQ $\langle \cV',\cP'\rangle$:
\begin{enumerate}
    \item $\cV'$ simulates the interaction between $\cV$ and $\cP^*$ 
    by itself
    until $\cV$ outputs its $\ell$th message. 
    \item
    $\cV'$ sends the transcript and $\cP^*$'s classical internal state to $\cP'$.
    \item $\cP'$ computes $\cP^*$'s $\ell$th message and sends it to $\cV'$.
    \item $\cV'$ runs the last step of $\cV$ to make the decision, and outputs the decision. 
\end{enumerate}
By the assumption, it is complete, i.e., $\cP'$ is accepted by $\cV'$ with high probability.
If it is sound, i.e., no PPT prover can be accepted by $\cV'$ with high probability, 
we obtain single-round QV-PoQ. 

Therefore, let us assume that it is not sound. That is, there exists a PPT prover $\cA$ that is accepted by $\cV'$ with high probability. 
By using $\cA$, we construct the following $(\ell-1)$-round PoQ $\langle \cV'',\cP''\rangle$:
\begin{enumerate}
    \item $\cV''$ and $\cP''$ run $\cV$ and $\cP^*$, respectively, 
    until the end of the $(\ell-1)$th round. 
    When $\cP''$ sends the $(\ell-1)$th message to $\cV''$, it also sends its classical internal state.
    \item $\cV''$ computes $\ell$th message of $\cV$ and, by using it, runs $\cA$. 
    \item
    By using the output of $\cA$, $\cV''$ runs the final step of $\cV$ to compute the decision. 
    \item 
    Output the decision.
\end{enumerate}
By the assumption, it is complete, i.e., $\cP''$ is accepted by $\cV''$ with high probability.
We can also show that it is sound. To show it, assume that there is a PPT adversary $\cB$ that is accepted by $\cV''$ with high probability.
Then, $\cB$ plus $\cA$ break soundness of the original $\ell$-round PoQ $\langle\cV,\cP\rangle$:
first $\cB$ is run by the $(\ell-1)$th round, and then $\cA$ is run in the $\ell$th round.
However, it is the contradiction.

\paragraph{\bf $\ell$-round IV-PoQ with black-box reduction imply $\ell$-round IV-PoQM or $(\ell-1)$-round IV-PoQ with black-box reduction. (\cref{thm:introIVPoQM})} 
   Let $(\cV,\cP)$ be an $\ell$-round IV-PoQ that are black-box constructed from a falsifiable assumption. 
   This means that 
   there exists a QPT reduction algorithm $\cR$ that satisfies the following:
         for any (not-necessarily-efficient) classical prover $\cP'$ that is accepted by $\cV$ with high probability,
$\cR^{\cP'}$ breaks the falsifiable assumption.
%\begin{align} \label{eq:bbr}
%    \Pr[\top\gets\langle\cC(1^\secp),\cR^{\cA}(1^\secp)\rangle]\ge\mathsf{thres}+\frac{1}{p(\secp)}
%\end{align}
    
Let us consider the following two cases.
\begin{itemize}
\item 
There exists a QPT algorithm $\cP^*$ whose internal state after sending its $(\ell-1)$th message is classical
and it is accepted by $\cV$ with high probability.
\item
There is no such $\cP^*$.
%For any QPT algorithm $\cP^*$ whose internal state after sending its $(\ell-1)$th message is classical,
%$\Pr[\top\gets\langle \cV(1^\secp),\cP^*(1^\secp)\rangle]$
%is small.
\end{itemize}
In the second case,
we immediately obtain IV-PoQM by considering the first $(\ell-1)$ rounds of the IV-PoQ as the initialization phase
and the $\ell$th round as the execution phase.

Therefore let us consider the first case.
By using $\cP^*$, we construct an $(\ell-1)$-round IV-PoQ $(\widetilde{\cV},\widetilde{\cP})$ as follows:
\begin{enumerate}
    \item $\widetilde{\cV}$ and $\widetilde{\cP}$ run $\cV$ and $\cP^*$, respectively, until the end of the $(\ell-1)$th round.
    \item When $\widetilde{\cP}$ sends its $(\ell-1)$th message, it also sends its classical internal state $s$. 
%    \item $\widetilde{\cV}$ computes $\cV$'s $\ell$th message $m_\ell$, runs $\cB(s,m_\ell)$, and uses its output to compute the last step of $\cV$. 
\item $\widetilde{\cV}$ simulates the $\ell$th round between $\cV$ and $\cP^*$ by itself,
    and outputs the output of $\cV$.
\end{enumerate}
Its completeness clearly follows from the assumption that
$\cP^*$ is accepted by $\cV$ with high probability.

The remaining task is to show the existence of a QPT reduction algorithm $\cQ$
that satisfies the following:
for any (not-necessarily-efficient) classical algorithm $\widetilde{\cP}^{*}$ that is accepted by
$\widetilde{\cV}$ with high probability,
$\cQ^{\widetilde{\cP}^{*}}$ breaks the falsifiable assumption.

To show it, we define
a classical (inefficient) adversary $\cA$ interacting with $\cV$ as follows:
Run $\widetilde{\cP}^{*}$ until the $(\ell-1)$th round.
In the $\ell$th round, simulate the $\ell$th round of $\cP^*$.
%run, but do not send the classical internal state $s$ in the last round. Upon receiving the $\ell$th message $m_\ell$, run $\cB(s,m_\ell)$ and sends its output.

Clearly,
$\cA$ is accepted by
$\cV$ with high probability.
Therefore, 
    $\cR^{\cA}$ breaks the falsifiable assumption.

Let us define the QPT reduction algorithm $\cQ^{\widetilde{\cP}^{*}}$ as follows:
It is identical to $\cR^{\cA}$ except for the following:
when $\cR$ queries $\cA$'s inefficient part (namely, simulating the $\ell$th round of $\cP^*$),
the query is replaced with the efficient simulation of it.

Here, the last efficient simulation is always possible as is shown in \cite{STOC:TomZha26}.\footnote{\cite{STOC:TomZha26} showed that any quantum query to
a classical oracle can be efficiently simulated in QPT if the classical oracle's truth table is sampled from a distribution that
can be QPT samplable. See \cref{lem:TZ}.}
Therefore, the probability that
    $\cQ^{\widetilde{\cP}^{*}}$ wins the security game of the falsifiable assumption
    is close to that of
    $\cR^{\cA}$, which is large.
Hence $\cQ^{\widetilde{\cP}^{*}}$ breaks the falsifiable assumption.

\subsection{Related Works} \label{sec:related}

\cite{TCC:MalMou25} constructed information-theoretically-sound PoQ that are sound against \emph{classical}-memory-bounded classical provers.
They also constructed information-theoretically-secure claw generation that are secure against quantum-memory-bounded quantum provers.
%and an information-theoretically-sound classical verification of quantum computation (CVQC) that is sound against provers whose total memory (the classical memory plus the quantum memory) is bounded. 
%Their results are not PoQM, because in PoQM, the quantum memory is bounded while the classical one is not.
\cite{DBLP:journals/corr/abs-2205-04656} introduced classical verification of quantum depth. 
%A CVQD is an interactive protocol between a PPT verifier and a QPT prover  over a classical channel. Completeness is that if the prover can implement high-depth quantum circuits with classical polynomial-time computation, the verifier accepts with high probability. Soundness is that if the prover has only access to low-depth quantum circuits with classical polynomial-time computation, the verifier rejects with high probability.
%CVQD certifies that the prover has large quantum circuit depth, whereas PoQM certifies that the prover has large quantum memory. Thus, they are fundamentally different primitives. Still, they share a similarity: both impose restrictions only on the prover's quantum component, placing no constraints on the classical part.
%where a PPT verifier can certify that the prover has generated a qubit.  
%We do not know how to use it to verify the fact that the prover has possessed a certain amount of quantum memory for a certain time period.
%introduced...\mor{check}, which is a protocol for certifying the dimension of a quantum system. In this test, a verifier encodes an $n$‑bit string into $n$ qubits in either the computational or Hadamard basis, sends it to the prover, and later checks that the prover can recover most of the string. A high success probability on this task certifies, in an information-theoretic sense, that the prover maintained quantum memory during the challenge. 
Although these results share similar motivations, they would not be able to
classically verify that a prover has possessed a certain amount of quantum memory for a specified period of time.

\cite{chao2020quantumdimensiontestusing} constructed information-theoretically-sound quantum dimension tests. \cite{ICALP:FuWanZha23} constructed computationally-sound quantum dimension tests with a classical verifier. Dimension tests certify the amount of quantum memory available during the computation, but cannot certify the amount of quantum memory retained for a specified period of time.

\cite{C:BGKPV23} constructed a test of qubits protocol. More recently, \cite{cryptoeprint:2026/1215} introduced tests of non-commutation (ToNC), a weaker variant of qubit tests. As observed in \cite{cryptoeprint:2026/1215}, ToNC imply PoQM. Therefore, PoQM are weaker than qubit tests.

\cite{Maxfield} introduced quantum proofs of space, which are a quantum variant of proofs of space~\cite{C:DFKP15}.
Our definition of PoQM is based on them.
However, 
the verifier and the channel of \cite{Maxfield} are quantum.

\cite{STOC:TomZha26} showed that $\ell$-round PoQ imply $(\ell-1)$-round PoQ, OWFs, or signature tokens.
This result is similar to our \cref{thm:PoQandPoQM} capturing
``genuinely-interactive'' PoQ. However, their result is incomparable with our result.
(In particular, we do not know any relation between PoQM and signature tokens.)

\cite{STOC:TomZha26} showed that IV-PoQ (even of poly round) that are black-box constructed from
a quantumly-secure falsifiable assumption imply OWPuzzs.
\cref{coro:introupgrade} is stronger than their result in the sense that IV-PoQM are obtained, 
while it is weaker because it holds only for constant-round IV-PoQ.

\if0
\subsection{Open Problems}
\mor{Can we get PKE from PoQM?}
\fi

\section{Preliminaries}

\subsection{Basic Notations}
We use standard notations of quantum computing and cryptography. 
All polynomials in this paper have coefficients in $\mathbb{N}$. 
We use $\secp$ as the security parameter. For a  bit string $x$, $x_i$ denotes the $i$th bit of $x$. For two bit strings $x$ and $y$, $x\|y$ means the concatenation of them. $[n]$ means the set $\{1,2,\ldots,n\}$. $\lceil x\rceil$ means the minimum integer greater than or equal to $x$. $\negl$ is a negligible function, and $\poly$ is a polynomial. PPT stands for (classical) probabilistic polynomial-time and QPT stands for quantum polynomial-time. We refer to a non-uniform QPT algorithm as a QPT algorithm with polynomial-size quantum advice. 
For a set $S$, $x\gets S$ means that an element $x$ is chosen from $S$ uniformly at random.
For an algorithm $\cA$, $y\gets \cA(x)$ means that the algorithm $\cA$ outputs $y$ on input $x$. 
For an algorithm $\cA$ that takes a quantum state as input and outputs a quantum state, $\cA(\rho)$
often means the output state of $\cA$ on input $\rho$.
For two density matrices $\rho$ and $\sigma$, $\TD(\rho,\sigma)$ is their trace distance. 

For two interactive algorithms $\cA$ and $\cB$ over a classical channel, 
$\rho_{\mathbf{A},\mathbf{B}}\gets\langle\cA(x),\cB(y)\rangle$ means that 
$\cA$ and $\cB$ are executed on input $x$ and $y$, respectively, and
the final output state is a quantum state $\rho_{\mathbf{A},\mathbf{B}}$ over two registers $\mathbf{A}$ and $\mathbf{B}$,
where
$\cA$'s output register is $\mathbf{A}$ and $\cB$'s output register is $\mathbf{B}$. 

\subsection{Lemmata}
Here we explain two lemmata that we will use later.
\begin{lemma}[\cite{C:BonZha13}, Lemma 1] \label{lemma:bz}
    Let $\cA$ be a quantum algorithm that outputs a classical bit string. Let $\cA'$ be another quantum algorithm obtained from $\cA$ by pausing $\cA$ at an arbitrary stage of the execution, performing a measurement that obtains one of $k$ outcomes, and then resuming $\cA$. Then $\Pr[x\gets\cA']\ge\Pr[x\gets\cA]/k$ for any 
    bit string $x$.
\end{lemma}

\begin{lemma}[LOCC Leakage Property for BB84 States~\cite{cryptoeprint:2024/1876}, Theorem 10] \label{lem:LOCC_leak}
    Let us consider the following game between a (not-necessarily-polynomial-time) adversary $\cA$ and a challenger $\cC$:
    \begin{enumerate}
        \item $\cC$ samples $x,\theta\gets\{0,1\}^\secp$ and outputs 
        $|R_0\rangle\coloneqq\bigotimes_{i=1}^\secp H^{\theta_i}\ket{x_i}$.
        Here, $H$ is the Hadamard operator.
%        \item $\cA$ and $\cC$ interact over a classical channel. 
        \item 
        Let $N:\mathbb{N}\to\mathbb{N}$ be a function.
        For $i=1,2,...,N(\secp)$, $\cA$ and $\cC$ do the following.
        \begin{enumerate}
            \item 
        $\cA$ sends $\cC$ a classical description $E_i$ of a (not-necessarily-polynomial-time) quantum algorithm
        which takes a quantum state as input and outputs a classical bit string and a pure quantum state.
        \item
        $\cC$ runs the algorithm $E_i$ on input $\ket{R_{i-1}}$. 
        Let $(L_i,\ket{R_i})$ be the output, where $L_i$ is a classical bit string.
        \item 
        $\cC$ sends $L_i$ to $\cA$.
        \end{enumerate}
        \item $\cC$ sends $\theta$ to $\cA$.
        \item $\cA$ outputs $x'$ and sends it to $\cC$.
        \item $\cC$ outputs $\top$ if $x'=x$, and otherwise it outputs $\bot$.
    \end{enumerate}
    Then, for all sufficiently large $\secp\in\mathbb{N}$ and for any (not-necessarily-polynomial-time) adversary $\cA$,
        $\Pr[\top\gets\cC]\le 2^{-\frac{\xi}{2}\cdot\secp + 2^{-\secp}}$,
    where $\xi\coloneqq-\log(\frac{1}{2}+\frac{1}{2\sqrt{2}})>0.22$.
\end{lemma}

\subsection{Cryptography}
In this subsection, we explain several cryptographic primitives that we will use.

First, we recall the definition of proofs of quantumness (PoQ) introduced by \cite{JACM:BCMVV21}.
\begin{definition}[Proofs of Quantumness (PoQ)~\cite{JACM:BCMVV21}]
    An $(\alpha,\beta)$-proof of quantumness (PoQ) is a set $(\cV,\cP)$ of interactive algorithms over a classical channel. $\cV$ (verifier) is a PPT algorithm that takes $1^\secp$ as input and outputs $\top$ or $\bot$.  
    $\cP$ (prover) is a QPT algorithm that takes $1^\secp$ as input and outputs nothing. We require the following two properties.
    \paragraph{\bf{$\alpha$-completeness:}} For all sufficiently large $\secp\in\mathbb{N}$,
        $\Pr[\top\gets \langle\cV(1^\secp),\cP(1^\secp)\rangle] \ge \alpha(\secp)$.

    \paragraph{\bf{$\beta$-soundness:}} For any non-uniform PPT adversary $\cP^*$ and for all sufficiently large $\secp\in\mathbb{N}$, 
        $\Pr[\top\gets \langle\cV(1^\secp),\cP^*(1^\secp)\rangle] \le \beta(\secp)$.

    We say that i.o.-$\beta$-soundness is satisfied if the above inequality is satisfied for
    infinitely many $\secp\in\mathbb{N}$.

    % We say PoQ exist if $\ell$-round $(c,s)$-PoQ exist for some polynomial $\ell$ and
    % some $\alpha,\beta$ such that $\alpha(\secp)-\beta(\secp)\ge\frac{1}{\poly(\secp)}$ for all $\secp\in\mathbb{N}$.
\end{definition}

   \cite{C:MorYam24} introduced the following variant of PoQ, called inefficient-verifier proofs of quantumness (IV-PoQ).
\begin{definition}[Inefficient-Verifier Proofs of Quantumness (IV-PoQ) \cite{C:MorYam24}]
    An $(\alpha,\beta)$-inefficient-verifier proof of quantumness (IV-PoQ) is an $(\alpha,\beta)$-PoQ in which the verifier’s final computation to make the decision can be unbounded.
\end{definition}

We also consider another variant of PoQ where the verifier is QPT:
\begin{definition}[Quantum-Verifier Proofs of Quantumness (QV-PoQ)]
An $(\alpha,\beta)$-quantum-verifier proof of quantumness (QV-PoQ) is an $(\alpha,\beta)$-PoQ with a QPT verifier.
(Note that the communication between the verifier and the prover is still classical.)
\end{definition}

Next, we give the definition of state puzzles (StatePuzzs).

\begin{definition}[State Puzzles~\cite{STOC:KhuTom25}]
    Let $\epsilon:\mathbb{N}\to [0,1]$ be a function.
    An $\epsilon$-StatePuzz is a QPT algorithm $\Samp$ that takes $1^\secp$ as input and outputs a pair $(s,\ket{\psi_s})$ of a bit string $s$ and a 
    pure quantum state $\ket{\psi_s}$ satisfying the following property.
    \paragraph{\bf{Security:}} For any non-uniform QPT adversary $\cA$ that takes $s$ as input and outputs a quantum state, and for all sufficiently large $\secp\in\mathbb{N}$,
        $\underset{(s,\ket{\psi_s})\gets \Samp(1^\secp)}
        {\mathbb{E}}
        \langle\psi_s|\cA(s)|\psi_s\rangle\le 1- \epsilon(\secp)$.
    If $\epsilon(\secp)=1-\negl(\secp)$, we just call it a state puzzle.
\end{definition}

The following lemma is implicitly shown in \cite{STOC:KhuTom25}.

\begin{lemma} \label{thm:StatePuzz}
    Let $p$ be any polynomial. If $1/p$-StatePuzzs exist, then StatePuzzs exist.
\end{lemma}

% We recall the LWE assumption.

% \begin{definition}[Learning with Errors (LWE)~\cite{JACM:Regev09}]
%     Let $n,m,q:\mathbb{Z}\to\mathbb{Z}$ be functions and $\{\chi_\secp\}_{\secp\in\mathbb{N}}$ be a family of probability distributions over $\mathbb{Z}_{q(\secp)}$. 
%     The $(t,\epsilon)$-learning with errors (LWE) assumption is the following assumption.
%     For any non-uniform quantum adversary $\cA$ running in time $t(\secp)$,
%     \begin{align}
%         \left|\Pr[\top\gets\cA(A,A\cdot s+e)]-\Pr[\top\gets\cA(A,u)]\right|\le\epsilon(\secp), \label{LWE}
%     \end{align}
%     where $A\gets\mathbb{Z}^{n(\secp)\times m(\secp)}_{q(\secp)},s\gets\mathbb{Z}^{n(\secp)}_{q(\secp)},e\gets\chi^{m(\secp)}_\secp$ and $u\gets\mathbb{Z}^{m(\secp)}_{q(\secp)}$. 
    
%     We refer to the $(\poly(\secp),\negl(\secp))$-LWE assumption as standard hardness of LWE. For a constant $0<c<1$, we refer to the $(\poly(2^{\secp^c}),\negl(2^{\secp^c}))$-LWE assumption as $c$-subexponential hardness of LWE.
%     \mor{soudan}
% \end{definition}

We also define $1$-of-$2^k$ puzzles~\cite{ITCS:LLQ22}.

\begin{definition}[$1$-of-$2^k$ puzzles~\cite{ITCS:LLQ22}]
    Let $k$ be a polynomial. A $1$-of-$2^k$ puzzle is a set $(\mathsf{KeyGen, Obligate, Solve, Ver})$ of algorithms with the following syntax.
    \begin{itemize}
        \item $\mathsf{KeyGen}(1^\secp)\to(\mathsf{pk,sk}):$ A PPT algorithm that takes $1^\secp$ as input and outputs a public key $\mathsf{pk}$ and a secret key $\mathsf{sk}$.
        \item $\mathsf{Obligate}(\mathsf{pk})\to(y,\rho):$ A QPT algorithm that takes $\mathsf{pk}$ as input and outputs a bit string $y$ 
        and a quantum state $\rho$.
        \item $\mathsf{Solve}(\mathsf{pk},y,\rho,\mathsf{ch})\to\mathsf{ans}:$ A QPT algorithm that takes $\mathsf{pk},y,\rho$ and a challenge $k$-bit string $\mathsf{ch}$ as input and outputs a classical answer $\mathsf{ans}$.
        \item $\mathsf{Ver}(\mathsf{sk},y,\mathsf{ch},\mathsf{ans})\to\top/\bot:$ A polynomial-time classical deterministic algorithm that takes $\mathsf{sk},y,\mathsf{ch}$ and $\mathsf{ans}$ as input and outputs $\top$ or $\bot$.
    \end{itemize}
    We require the following properties.
    \paragraph{\bf Completeness:}
    \begin{align}
        \Pr\qty[\top\gets\mathsf{Ver}(\mathsf{sk},y,\mathsf{ch},\mathsf{ans}):
        \begin{array}{l}
            (\mathsf{pk,sk})\gets\mathsf{KeyGen}(1^\secp) \\
            (y,\rho)\gets\mathsf{Obligate}(\mathsf{pk}) \\
            \mathsf{ch}\gets \{0,1\}^{k(\secp)} \\
            \mathsf{ans}\gets\mathsf{Solve}(\mathsf{pk},y,\rho,\mathsf{ch})
        \end{array}
        ] \ge 1-\negl(\secp).
    \end{align}

    \paragraph{\bf $c$-soundness:}
    Let us consider the following game between a set $(\cA,\cB,\cC)$ of adversaries and a challenger $\mathsf{Chal}$:
    \begin{enumerate}
        \item $\mathsf{Chal}$ runs $(\mathsf{pk,sk})\gets\mathsf{KeyGen}(1^\secp)$.
        \item $\cA$ receives the public key $\mathsf{pk}$, and outputs a bit string $y$ and a quantum state $\sigma_{\mathbf{BC}}$ over two registers $\mathbf{B}$ and $\mathbf{C}$. 
        \item
        $\cA$ sends $y$ to $\mathsf{Chal}$. 
        $\cA$ sends $\cB$ the register $\mathbf{B}$.
        $\cA$ sends $\cC$ the register $\mathbf{C}$. 
        \item $\mathsf{Chal}$ samples $\mathsf{ch}\gets\bit^{k(\secp)}$ and sends $\mathsf{ch}$ to both $\cB$ and $\cC$.
        \item $\cB$ outputs an answer $\mathsf{ans}_\cB$ and sends it to $\mathsf{Chal}$.
        $\cC$ outputs an answer $\mathsf{ans}_\cC$ and sends it to $\mathsf{Chal}$.
        \item $\mathsf{Chal}$ outputs $\top$ if 
            $\mathsf{Ver}(\mathsf{sk},y,\mathsf{ch},\mathsf{ans}_\cB)=\top \land \mathsf{Ver}(\mathsf{sk},y,\mathsf{ch},\mathsf{ans}_\cC)=\top$.
        Otherwise, $\mathsf{Chal}$ outputs $\bot$.
    \end{enumerate}
   For any 
   set $(\cA,\cB,\cC)$
   of non-uniform QPT adversaries,
    \begin{align}
        \Pr[\top\gets\mathsf{Chal}] \le2^{-k(\secp)}+\negl(2^{\secp^c}). \label{Eq:1of2puzsound}
    \end{align}
\end{definition}

The following lemma is implicitly shown in \cite{ITCS:LLQ22}.\footnote{\cite{ITCS:LLQ22} implicitly showed that, for any $c>0$ , assuming $c$-subexponential hardness of LWE, $1$-of-$2^k$ puzzles with $c$-soundness exist. $c$-subexponential hardness of LWE roughly means that any quantum algorithm running in time $O(2^{\secp^c})$ can distinguish two distributions with probability at most $\negl(2^{\secp^c})$. Let $c'>0$ be any constant. By replacing the security parameter $\secp$ with $\secp'\coloneqq\secp^\frac{c}{c'}$, $c'$-subexponential hardness of LWE can be converted to $c$-subexponential hardness of LWE. Thus, for any $c,c'>0$, assuming $c'$-subexponential hardness of LWE, $1$-of-$2^k$ puzzles with $c$-soundness exist.} 

\begin{lemma} \label{lem:LWEto12puz}
    Assuming the subexponential hardness of LWE, for any $c>0$ and for any polynomial $k$, $1$-of-$2^k$ puzzles with $c$-soundness exist.
\end{lemma}

We also use verifiable remote state preparations~\cite{FOCS:GheVid19,ITCS:Zhang25}.
In this paper, we use the formalism of \cite{ITCS:Zhang25}.

\begin{definition}[Remote State Preparations (RSPs)~\cite{ITCS:Zhang25}] \label{def:RSPV}
Let $n:\mathbb{N}\to\mathbb{N}$ be any polynomially bounded function. Let $p$ be polynomial.
    An $(n,\frac{1}{p})$-remote state preparation (RSP) is a set $(\cV,\cP)$ of interactive algorithms over a classical channel. $\cV$ is a PPT algorithm that takes $1^\secp$ as input and 
    outputs classical bit strings $(x,\theta)\in\bit^{n}\times\bit^{n}$ and $\mathsf{flag}\in\{\mathsf{pass},\mathsf{fail}\}$.
    $\cP$ is a QPT algorithm that takes $1^\secp$ as input and outputs a quantum state on the register $\mathbf{Q}$.
    We require the following two properties.
    \paragraph{\bf Completeness:}
    \begin{align}
        \TD(\phi_{\mathbf{F},\mathbf{D},\mathbf{Q}},\ket{\mathsf{pass}}\bra{\mathsf{pass}}_{\mathbf{F}}\otimes\eta_{\mathbf{D},\mathbf{Q}})\le\negl(\secp).
    \end{align}
    Here, for the notational simplicity, we consider that $\cV$'s classical outputs are encoded in a quantum state in the computational basis.
    $\cV$'s classical output $\mathsf{flag}$ is written in the register $\mathbf{F}$, and $(x,\theta)$ is written in the register $\mathbf{D}$.
    $\phi_{\mathbf{F},\mathbf{D},\mathbf{Q}}\gets\langle\cV(1^\secp),\cP(1^\secp)\rangle$ and 
     $\eta_{\mathbf{D},\mathbf{Q}}\coloneqq\frac{1}{4^{n}}\sum_{(x,\theta)\in\bit^{n}\times\bit^{n}}|x,\theta\rangle\langle x,\theta|_{\mathbf{D}}\otimes (\bigotimes^{n}_{i=1}H^{\theta_i}|x_i\rangle\langle x_i|H^{\theta_i})_{\mathbf{Q}}$.
    
    \paragraph{\bf $\frac{1}{p}$-soundness:} 
    For any non-uniform QPT adversary $\cP^*$ that outputs a quantum state on a register $\mathbf{Q'}$, there exists a non-uniform QPT algorithm $\mathsf{Sim}$ that maps a quantum state on the register $\mathbf{Q}$ to a quantum state on the registers $\mathbf{F}$ and $\mathbf{Q'}$ such that for any non-uniform QPT algorithm $\cD$,
    \begin{align}
        &\left|\Tr[\Pi^\mathsf{pass}_\mathbf{F}\sigma_\mathbf{F,D,Q'}]\Pr[\top\gets\cD\qty(\frac{\Pi^\mathsf{pass}_\mathbf{F}\sigma_\mathbf{F,D,Q'}\Pi^\mathsf{pass}_\mathbf{F}}{\Tr[\Pi^\mathsf{pass}_\mathbf{F}\sigma_\mathbf{F,D,Q'}]})] \right. \\
        & \left.-\Tr[\Pi^\mathsf{pass}_\mathbf{F}\mathsf{Sim}(\eta_{\mathbf{D,Q}})]\Pr[\top\gets\cD\qty(\frac{\Pi^\mathsf{pass}_\mathbf{F}\mathsf{Sim}(\eta_{\mathbf{D,Q}})\Pi^\mathsf{pass}_\mathbf{F}}{\Tr[\Pi^\mathsf{pass}_\mathbf{F}\mathsf{Sim}(\eta_{\mathbf{D,Q}})]})]\right| 
        \le\frac{1}{p(\secp)}. \label{def:RSPVsound}
    \end{align}
    Here, $\Pi^\mathsf{pass}_\mathbf{F}\coloneqq|\mathsf{pass}\rangle\langle\mathsf{pass}|_\mathbf{F}$ and $\sigma_\mathbf{F,D,Q'}\gets\langle\cV(1^\secp),\cP^*(1^\secp)\rangle$.
\end{definition}

The following lemma is shown in \cite{ITCS:Zhang25}:

\begin{lemma} \label{lem:LWEtoRSPs}
    Assuming the polynomial hardness of LWE, for any polynomially bounded function $n:\mathbb{N}\to\mathbb{N}$ and polynomial $p$, $r$-round $(n,\frac{1}{p})$-RSPs exist with a certain polynomial $r$.
\end{lemma}

% \begin{remark}
%     The soundness of RSP means that any adversary's output state condetioned on $\mathsf{flag}=\mathsf{pass}$ can be simulated from BB84 state.
%     \mor{soudan}
% \end{remark}

Finally, we explain quantum computation and classical communication key exchange (QCCC KE).

\begin{definition}[QCCC Key Exchange (QCCC KE)~\cite{cryptoeprint:2024/1707}]
    An $(\alpha,\beta)$-QCCC key exchange (KE) is a set $(\cA,\cB)$ of interactive algorithms over a classical channel. 
    $\cA$ ($\cB$) is a QPT algorithm that takes $1^\secp$ as input and outputs a bit string $a$ ($b$). We require the following properties.
    \paragraph{\bf $\alpha$-correctness:} For all sufficiently large $\secp\in\mathbb{N}$,
    \begin{align}
        \Pr[a=b:(a,b)\gets\langle\cA(1^\secp),\cB(1^\secp)\rangle]\ge\alpha(\secp).
    \end{align}
    Here, $(a,b)\gets\langle\cA(1^\secp),\cB(1^\secp)\rangle$ means that $\cA$'s output is $a$ and $\cB$'s output is $b$.
    \paragraph{\bf $\beta$-security:} For any non-uniform QPT adversary $\cE$ and for all sufficiently large $\secp\in\mathbb{N}$,
    \begin{align}
        \Pr[a=e:
        (a,b;\tau)\gets\langle\cA(1^\secp),\cB(1^\secp)\rangle,
        e\gets\cE(\tau)  
        ]\le\beta(\secp).
    \end{align}
    Here, $(a,b;\tau)\gets\langle\cA(1^\secp),\cB(1^\secp)\rangle$ means that $\cA$'s output is $a$, $\cB$'s output is $b$, and $\tau$ is the transcript.
    
If $(\cA,\cB)$ is a $(1-\negl,\negl)$-QCCC KE, then we simply say that $(\cA,\cB)$ is a QCCC KE. 
\end{definition}

%The following lemma was originally shown for classical KE, but we confirm that the proof also applies to QCCC KE.

We use the following security amplification lemma. 

\begin{lemma}\label{lem:KEamp}
    Let $p$ be any polynomial. If $(1-\negl,1-\frac{1}{p})$-QCCC KE exist, then QCCC KE exist.
\end{lemma}

The proof is given in \cref{sec:KEamp}.

\section{Proofs of Quantum Memory}
In this section, we define proofs of quantum memory (PoQM). 
We also observe that PoQM generalize the notion of PoQ.

\subsection{Definition}
We first define PoQM as follows.
\begin{definition}[Proofs of Quantum Memory (PoQM)]
Let $\alpha,\beta:\mathbb{N}\to[0,1]$ be any functions. Let $m_1,m_2:\mathbb{N}\to\mathbb{N}$ be any functions such that $m_1>m_2$.
An $(\alpha,\beta,m_1,m_2)$-proof of quantum memory ($(\alpha,\beta,m_1,m_2)$-PoQM) is a set $(\cV_1,\cP_1,\cV_2,\cP_2)$ of interactive algorithms 
over a classical channel
with the following syntax.
    \paragraph{Initialization Phase:}
    In the initialization phase, $\cV_1$ and $\cP_1$ interact over a classical channel.
    $\cV_1$ is a PPT algorithm that takes $1^\secp$ as input and outputs a bit string $v$.
    $\cP_1$ is a QPT algorithm that takes $1^\secp$ as input and outputs a bit string $\state$ and an $m_1$-qubit quantum state $\sigma_\state$.
    In other words,
    \begin{align}
    (v, (\state,\sigma_\state))\gets\langle \cV_1(1^\secp), \cP_1(1^\secp)\rangle. 
    \end{align}
    \paragraph{Execution Phase:} 
    In the execution phase, $\cV_2$ and $\cP_2$ interact over a classical channel.
    $\cV_2$ is a PPT algorithm that takes $v$ as input and outputs $\top$ or $\bot$.
    $\cP_2$ is a QPT algorithm that takes $\state$ and $\sigma_\state$ as input and outputs nothing.
    In other words,
    \begin{align}
        \top/\bot\gets\langle\cV_2(v),\cP_2(\state,\sigma_\state)\rangle.
    \end{align}
    %\mor{siki wa bun to onaji. nanode, ueno shiki no saigo ni . wo tsukeru}
We require the following two properties.
\paragraph{\bf{$\alpha$-completeness:}} For all sufficiently large $\secp\in\mathbb{N}$,
\begin{align}
    &\Pr[
    \top\gets \langle\cV_2(v),\cP_2(\state,\sigma_\state)\rangle
    : (v,(\state,\sigma_\state))\gets \langle \cV_1(1^\secp), \cP_1(1^\secp)\rangle 
    ] 
    \ge \alpha(\secp).
\end{align}

\paragraph{\bf{$(\beta,m_2)$-soundness:}} For any non-uniform QPT adversary $\cP_1^*$ that outputs a bit string $s$ and an $m_2$-qubit quantum state $\rho$, for any non-uniform \footnote{In our setting, it is more natural that the non-uniform QPT adversary $\cP_2^*$ takes only classical advice since we are interested in how much quantum memory the adversary can possess. However, we can construct PoQM with such stronger security, and therefore this only makes our results stronger.} QPT adversary $\cP_2^*$ that takes $s$ and $\rho$ as input, and for all sufficiently large $\secp\in\mathbb{N}$, \begin{align}
    \Pr[
    \top\gets \langle\cV_2(v),\cP_2^*(s,\rho)\rangle
    : (v,(s,\rho))\gets \langle \cV_1(1^\secp), \cP_1^*(1^\secp)\rangle 
    ] \le \beta(\secp).
\end{align}

We say that i.o.-$(\beta,m_2)$-soundness is satisfied if the above inequality is satisfied for
    infinitely-many $\secp\in\mathbb{N}$.
\end{definition}

We also introduce inefficient-verifier proofs of quantum memory (IV-PoQM), which is a variant of PoQM.
\begin{definition}[Inefficient-Verifier Proofs of Quantum Memory (IV-PoQM)]
    An $(\alpha,\beta,m_1,m_2)$-inefficient-verifier proof of quantum memory (IV-PoQM) is an $(\alpha,\beta,m_1,m_2)$-PoQM in which the verifier’s final computation to make the decision can be unbounded.
\end{definition}

\subsection{Amplification of $m_2$}
We show that $m_2$ can be increased by increasing $\beta$. 

\begin{lemma}  \label{thm:0_to_m}
Let $\alpha,\beta:\mathbb{N}\to[0,1]$ be any functions. Let $m_1,m_2:\mathbb{N}\to\mathbb{N}$ be any functions such that $m_1>m_2$.
An $(\alpha,\beta,m_1,0)$-PoQM is an $(\alpha,2^{m_2}\beta,m_1,m_2)$-PoQM.
\end{lemma}

\begin{proof}[Proof of \cref{thm:0_to_m}]
        Let $(\cV_1,\cP_1,\cV_2,\cP_2)$ be an $(\alpha,\beta,m_1,0)$-PoQM. We show that this is also an $(\alpha,2^{m_2}\beta,m_1,m_2)$-PoQM. 
        $\alpha$-completeness is straightforward.
        For the sake of contradiction, we assume that the PoQM is not $(2^{m_2}\beta,m_2)$-sound. This means that there exists an adversary $\cP_1^{*(m_2)}$ that outputs an $m_2$-qubit quantum state $\rho$ and a bit string $s$, and an adversary $\cP_2^{*(m_2)}$ that takes $\rho$ and $s$ as input such that
        \begin{align}
            &\Pr[\top\gets \langle\cV_2(v),\cP_2^{*(m_2)}(s,\rho)\rangle : (v,(s,\rho))\gets \langle \cV_1(1^\secp), \cP_1^{*(m_2)}(1^\secp)\rangle] 
            > 2^{m_2(\secp)}\beta(\secp) \label{Eq:notmPoQM}
        \end{align}
        for infinitely many $\secp\in\mathbb{N}$.
        
        From this $(\cP_1^{*(m_2)},\cP_2^{*(m_2)})$, we can construct a pair $(\cP_1^{*(0)},\cP_2^{*(0)})$ of adversaries that breaks $(\beta,0)$-soundness
        as follows.
        \begin{itemize}
            \item $\cP_1^{*(0)}:$ Run $(s,\rho)\gets\cP_1^{*(m_2)}(1^\secp)$. Measure $\rho$ in the computational basis to get a measurement result
            $p\in\{0,1\}^{m_2(\secp)}$. Output $s'\coloneqq(s,p)$.
            \item $\cP_2^{*(0)}:$ Get $s'=(s,p)$ as input. Run $\cP_2^{*(m_2)}(s,|p\rangle\langle p|)$.
        \end{itemize}
        
        By \cref{lemma:bz} and \cref{Eq:notmPoQM},
        \begin{align}
            &\Pr[\top\gets \langle\cV_2(v),\cP_2^{*(0)}(s')\rangle : (v,s')\gets \langle \cV_1(1^\secp), \cP_1^{*(0)}(1^\secp)\rangle] \\
            &\ge 2^{-m_2(\secp)}\Pr[\top\gets \langle\cV_2(v),\cP_2^{*(m_2)}(s,\rho)\rangle : (v,(s,\rho))\gets \langle \cV_1(1^\secp), \cP_1^{*(m_2)}(1^\secp)\rangle] \\
            &> \beta(\secp)
        \end{align}
        for infinitely many $\secp\in\mathbb{N}$. This contradicts $(\beta,0)$-soundness of the PoQM.
    \end{proof}

\subsection{Relation to PoQ}
We can show that PoQ is a special case of PoQM. 
\begin{lemma} \label{lem:PoQMtoPoQ}
    Let $\alpha,\beta:\mathbb{N}\to[0,1]$ be any functions. Let $m_1,m_2:\mathbb{N}\to\mathbb{N}$ be any functions such that $m_1>m_2$. If $(\alpha,\beta,m_1,m_2)$-PoQM exist, then $(\alpha,\beta)$-PoQ exist.
\end{lemma}

\cref{lem:PoQMtoPoQ} is obtained from the following lemma,
because an $(\alpha,\beta,m_1,m_2)$-PoQM is trivially an $(\alpha,\beta,m_1,m_2-1)$-PoQM for any $m_2\ge1$. 
\begin{lemma} \label{lem:0PoQMtoPoQ}
    Let $\alpha,\beta:\mathbb{N}\to[0,1]$ be any functions. Let $m_1:\mathbb{N}\to\mathbb{N}$ be any function. If $(\alpha,\beta,m_1,0)$-PoQM exist, then $(\alpha,\beta)$-PoQ exist.
\end{lemma}

\begin{proof}[Proof of \cref{lem:0PoQMtoPoQ}]
    Assume that $(\alpha,\beta,m_1,0)$-PoQM exist.
    Let $(\cV_1,\cP_1,\cV_2,\cP_2)$ be an $(\alpha,\beta,m_1,0)$-PoQM.
    From it, we construct an $(\alpha,\beta)$-PoQ $(\cV,\cP)$ as follows:
        \begin{enumerate}
            \item $\cV$ and $\cP$ get $1^\secp$ as input.
            \item 
            $\cV$ and $\cP$ interact over a classical channel. $\cV$ runs $v\gets\cV_1(1^\secp)$, and 
            $\cP$ runs $(\state,\sigma_\state)\gets\cP_1(1^\secp)$.  
            \item 
            $\cV$ and $\cP$ interact over a classical channel. $\cV$ runs $\top/\bot\gets\cV_2(v)$, and $\cP$ runs $\cP_2(\state,\sigma_\state)$. 
            $\cV$ outputs the output of $\cV_2(v)$. 
        \end{enumerate}
    $\alpha$-completeness of 
    thus
    constructed PoQ is clear. 
    Next we show $\beta$-soundness. For the sake of contradiction, we assume that the constructed PoQ is not $\beta$-sound.
    This means that there exists a non-uniform PPT prover $\cP^*$ such that
    \begin{align}
        \Pr[\top\gets\langle\cV(1^\secp),\cP^*(1^\secp)\rangle] > \beta(\secp)
    \end{align}
    for infinitely many $\secp\in\mathbb{N}$.
    We divide $\cP^*$ into two algorithms $\cP_1^*$ and $\cP_2^*$ such that $\cP_1^*$ interacts with $\cV_1$
    and $\cP_2^*$ interacts with $\cV_2$.
    Because $\cP^*$ is a PPT algorithm, $\cP_1^*$ outputs only a classical bit string, which we call $s$, and 
    $\cP_2^*$ takes only $s$ as input.
    We can show that thus constructed $(\cP_1^*,\cP_2^*)$ breaks $(\beta,0)$-soundness of the PoQM,
    because
    \begin{align}
        &\Pr[
        \top\gets \langle\cV_2(v),\cP_2^*(s)\rangle
        : (v,s)\gets \langle \cV_1(1^\secp), \cP_1^*(1^\secp)\rangle 
        ] \\
        &=\Pr[\top\gets\langle\cV(1^\secp),\cP^*(1^\secp)\rangle]  
        > \beta(\secp)
    \end{align}
    for infinitely many $\secp\in\mathbb{N}$.
    Hence the constructed PoQ is $\beta$-sound.
    \end{proof}

\section{Constructions of PoQM} \label{Sec:constructPoQM}
In this section, we provide two constructions of PoQM. 
The first construction is from $1$-of-$2^k$ puzzles.
The second one is from RSPs. 

\subsection{$1$-of-$2^k$ Puzzles Imply PoQM} \label{sec:puztoPoQM}
\begin{theorem} \label{thm:12puztoPoQM}
Let $m_2:\mathbb{N}\to\mathbb{N}$ be any polynomially-bounded function.
Let $c>0$ be any constant such that
$m_2(\secp)=O(\secp^c)$. 
Let $k$ be any polynomial
such that $k(\secp)=\omega(\secp^c)$. 
If $1$-of-$2^k$ puzzles with $c$-soundness exist, then, 
two-round\footnote{In this paper, a single round is a pair of two messages first one is from the verifier to the prover and the second one is from the prover to the verifier.} $(1-\negl,\negl,m_1,m_2)$-PoQM exist with some polynomial $m_1$. 
\end{theorem}

By combining this theorem with \cref{lem:LWEto12puz}, we obtain the following corollary.

\begin{corollary}
Let $m_2:\mathbb{N}\to\mathbb{N}$ be any polynomially-bounded function. Assuming the subexponential hardness of LWE, two-round $(1-\negl,\negl,m_1,m_2)$-PoQM exist with some polynomial $m_1$. 
\end{corollary}

\begin{proof}[Proof of \cref{thm:12puztoPoQM}] 
        Assume that $1$-of-$2^k$ puzzles with $c$-soundness exist.
        Let $(\mathsf{KeyGen, Obligate, Solve, Ver})$ be a $1$-of-$2^k$ puzzle with $c$-soundness. 
        We construct a PoQM $(\cV_1,\cP_1,\cV_2,\cP_2)$ as follows:
        
        \paragraph{Initialization Phase:}
        \begin{enumerate}
            \item $\cV_1$ and $\cP_1$ get $1^\secp$ as input.
            \item $\cV_1$ runs $(\mathsf{pk,sk})\gets\mathsf{KeyGen}(1^\secp)$ and sends $\mathsf{pk}$ to $\cP_1$.
            \item $\cP_1$ runs $(y,\rho)\gets\mathsf{Obligate(\mathsf{pk})}$ and sends $y$ to $\cV_1$. The number of qubits of $\rho$ is $m_1(\secp)$. 
            \item $\cV_1$ outputs $v\coloneqq(\mathsf{sk},y)$. 
            $\cP_1$ outputs $(\state,\sigma_\state)\coloneqq((\mathsf{pk},y), \rho)$.
        \end{enumerate}

        \paragraph{Execution Phase:}
        \begin{enumerate}
            \item $\cV_2$ takes $v$ as input. $\cP_2$ takes $(\state,\sigma_\state)$ as input.
            \item $\cV_2$ samples $\mathsf{ch}\gets\{0,1\}^{k(\secp)}$ and sends it to $\cP_2$.
            \item $\cP_2$ runs $\mathsf{ans}\gets\mathsf{Solve}(\mathsf{pk},y,\rho,\mathsf{ch})$ and sends $\mathsf{ans}$ to $\cV_2$.
            \item $\cV_2$ runs $\top/\bot\gets\mathsf{Ver}(\mathsf{sk},y,\mathsf{ch,ans})$ and outputs its output.
        \end{enumerate}
        Our goal is to show that the constructed $(\cV_1,\cP_1,\cV_2,\cP_2)$ is a $(1-\negl,\negl,\allowbreak m_1,m_2)$-PoQM. 
        We achieve this goal with the following three steps:
        \begin{enumerate}
            \item  We show that $(\cV_1,\cP_1,\cV_2,\cP_2)$ is  
        a $(1-\negl,\epsilon,m_1,0)$-PoQM, where $\epsilon(\secp)\coloneqq(2^{-k(\secp)}+\negl(2^{\secp^c}))^\frac{1}{2}$.
        \item 
        Using \cref{thm:0_to_m},
        a $(1-\negl,\epsilon,m_1,0)$-PoQM is a
        $(1-\negl,2^{m_2}\epsilon,m_1,m_2)$-PoQM.
        \item 
        We show that $2^{m_2(\secp)}\epsilon(\secp)=\negl(\secp)$.
        \end{enumerate}
        The second step is straightforward. In the following, we will explain the first and third steps.

\paragraph{First step.}
        $(1-\negl)$-completeness is straightforward.
        Let us show $(\epsilon,0)$-soundness. For the sake of contradiction, assume that it is not $(\epsilon,0)$-sound. 
        Then there exists a pair $(\cP_1^*,\cP_2^*)$ of adversaries such that
        \begin{align}
            \epsilon(\secp)&<\Pr[\top\gets \langle\cV_2(v),\cP_2^*(s)\rangle : (v,s)\gets \langle \cV_1(1^\secp), \cP_1^*(1^\secp)\rangle] \\
            &=\Pr[\top\gets\mathsf{Ver}(\mathsf{sk},y,\mathsf{ch},\mathsf{ans}):
            \begin{array}{l}
                (\mathsf{pk},\mathsf{sk})\gets\mathsf{KeyGen}(1^\secp) \\
                (y,s')\gets\cP_1^*(\mathsf{pk}) \\
                \mathsf{ch}\gets\bit^{k(\secp)} \\
                \mathsf{ans}\gets\cP_2^*(\mathsf{ch},s')
            \end{array}
            ]\label{Eq:notsound}
        \end{align}
        for infinitely many $\secp\in\mathbb{N}$.
        From this $(\cP_1^*,\cP_2^*)$, we can construct a set $(\cA,\cB,\cC)$ of adversaries that breaks $c$-soundness of the $1$-of-$2^k$ puzzle as follows:
        \begin{itemize}
            \item $\cA$ : Run $(y,s')\gets\cP_1^*(\mathsf{pk})$. Send $y$ to $\mathsf{Chal}$ and send $s'$ to $\cB$ and $\cC$.
            \item $\cB$ : Run $\mathsf{ans}_\cB\gets\cP_2^*(\mathsf{ch},s')$ and send $\mathsf{ans}_\cB$ to $\mathsf{Chal}$.
            \item $\cC$ : Run $\mathsf{ans}_\cC\gets\cP_2^*(\mathsf{ch},s')$ and send $\mathsf{ans}_\cC$ to $\mathsf{Chal}$.
        \end{itemize}
        $(\cA,\cB,\cC)$ can break $c$-soundness as follows:
        \begin{align}
            &\Pr[\top\gets\mathsf{Chal}]\\
            &=\Pr[\top\gets\mathsf{Ver}(\mathsf{sk},y,\mathsf{ch},\mathsf{ans}_\cB)\land\top\gets\mathsf{Ver}(\mathsf{sk},y,\mathsf{ch},\mathsf{ans}_\cC):
            \begin{array}{l}
                (\mathsf{pk},\mathsf{sk})\gets\mathsf{KeyGen}(1^\secp) \\
                (y,s')\gets\cP_1^*(\mathsf{pk}) \\
                \mathsf{ch}\gets\bit^{k(\secp)} \\
                \mathsf{ans}_\cB\gets\cP_2^*(\mathsf{ch},s') \\
                \mathsf{ans}_\cC\gets\cP_2^*(\mathsf{ch},s')
            \end{array} \label{Eq:puzadv}
            ] \\
            &\ge\Pr[\top\gets\mathsf{Ver}(\mathsf{sk},y,\mathsf{ch},\mathsf{ans}):
            \begin{array}{l}
                (\mathsf{pk},\mathsf{sk})\gets\mathsf{KeyGen}(1^\secp) \\
                (y,s')\gets\cP_1^*(\mathsf{pk}) \\
                \mathsf{ch}\gets\bit^{k(\secp)} \\
                \mathsf{ans}\gets\cP_2^*(\mathsf{ch},s')
            \end{array}
            ]^2 \label{Eq:causch} \\
            &> \epsilon(\secp)^2=2^{-k(\secp)}+\negl(2^{\secp^c})
        \end{align}
        for infinitely many $\secp\in\mathbb{N}$. The first inequality follows from the Jensen's inequality.
        Hence, $(\cA,\cB,\cC)$ break $c$-soundness of the $1$-of-$2^k$ puzzle, contradicting the assumption.

        \paragraph{Third step.}
        Since $k(\secp)=\omega(\secp^c)$, we have $2^{-k(\secp)}=\negl(2^{\secp^c})$, and thus $\epsilon(\secp)=(2^{-k(\secp)}+\negl(2^{\secp^c}))^\frac{1}{2}=\negl(2^{\secp^c})$. From $m_2(\secp)=O(\secp^c)$, there exists a polynomial $p$ such that $2^{m_2(\secp)}\le p(2^{\secp^c})$ for all sufficiently large $\secp\in\mathbb{N}$. Hence, we obtain
            $2^{m_2(\secp)}\epsilon(\secp)\le p(2^{\secp^c})\negl(2^{\secp^c})
            =\negl(2^{\secp^c})=\negl(\secp).$
    \end{proof}

\subsection{RSPs Imply PoQM} \label{Sec:RSPtoPoQM}
\begin{theorem} \label{thm:RSPtoPoQM}
    Let $p$ be any polynomial. Let $m_2:\mathbb{N}\to\mathbb{N}$ be any polynomially-bounded function such that $m_2(\secp)=\omega(\log(\secp))$.
    If $(\lceil9.1m_2\rceil,\frac{1}{2p})$-RSPs exist, then $(1-\negl,1/p,\lceil9.1m_2\rceil,m_2)$-PoQM exist.
\end{theorem}

    By combining this theorem with \cref{lem:LWEtoRSPs}, we obtain the following corollary.
%    \minki{In \cref{lem:LWEtoRSPs}, $n$ ($=\Theta(m_2)$ here) is polynomial but here $m_2=\omega(\log \lambda)$. Can we choose, say, $m_2=\log^2 \lambda$?}
    \begin{corollary}
        Let $p$ be any polynomial. Let $m_2:\mathbb{N}\to\mathbb{N}$ be any polynomially-bounded function such that $m_2(\secp)=\omega(\log(\secp))$.  Assuming the polynomial hardness of LWE, $r$-round $(1-\negl,1/p,\lceil9.1m_2\rceil,m_2)$-PoQM exist with a certain polynomial $r$.
    \end{corollary}

\begin{proof}[Proof of \cref{thm:RSPtoPoQM}]
        Assume that $(\lceil9.1m_2\rceil,\frac{1}{2p})$-RSPs exist. Let $(\cV,\cP)$ be a $(\lceil9.1m_2\rceil,\frac{1}{2p})$-RSP. We construct a $(1-\negl,1/p,\lceil9.1m_2\rceil,m_2)$-PoQM $(\cV_1,\cP_1,\cV_2,\cP_2)$ as follows:
        \paragraph{Initialization Phase:}
        \begin{enumerate}
            \item $\cV_1$ and $\cP_1$ take $1^\secp$ as input.
            \item $\cV_1$ runs $v\gets\cV(1^\secp)$ where $v\in\{(\mathsf{pass},x,\theta),\mathsf{fail}\}$ and $\cP_1$ runs $\phi\gets\cP(1^\secp)$.
            Here, $x$ and $\theta$ are $\lceil9.1m_2\rceil$-bit strings,
            and $\phi$ is a $\lceil9.1m_2\rceil$-qubit state.
            \item $\cV_1$'s output is $v$. $\cP_1$'s output is $(\state,\sigma_\state)\coloneqq(1^\secp,\phi)$.
        \end{enumerate}

        \paragraph{Execution Phase:}
        \begin{enumerate}
            \item $\cV_2$ takes $v\in\{(\mathsf{pass},x,\theta),\mathsf{fail}\}$ as input. 
            $\cP_2$ takes $(\state,\sigma_\state)=(1^\secp,\phi)$ as input.
            \item If $v=\mathsf{fail}$, then $\cV_2$ samples $\theta\gets\bit^{\lceil9.1m_2\rceil}$ and sends it to $\cP_2$. 
            If 
            $v=(\mathsf{pass},x,\theta)$,  
            $\cV_2$ sends $\theta$ to $\cP_2$.
            \item For each $i\in[\lceil9.1m_2\rceil]$, $\cP_2$ measures $i$th qubit of $\phi$ in the computational basis if $\theta_i=0$ or 
            in the Hadamard basis if $\theta_i=1$.
            Let $x_i'$ be the measurement result on the $i$th qubit. $\cP_2$ sets $x'\coloneqq x_1'\|...\|x_{\lceil9.1m_2\rceil}'$.
            \item $\cP_2$ sends $x'$ to $\cV_2$. 
           \item 
           If $v=\mathsf{fail}$ or $x\neq x'$, $\cV_2$ outputs $\bot$. Otherwise, $\cV_2$ outputs $\top$.
        \end{enumerate}
        
        Now we show that the constructed $(\cV_1,\cP_1,\cV_2,\cP_2)$ is a $(1-\negl,1/p,\allowbreak\lceil9.1m_2\rceil,m_2)$-PoQM. 
        
        $(1-\negl)$-completeness is straightforward.
        Let us next show $(1/p,m_2)$-soundness. We define $\mathbf{Hybrid}_0$ as follows, which is the original security game for
        $(1/p,m_2)$-soundness. 

        \paragraph{$\mathbf{Hybrid}_0$:}
        \begin{enumerate}
            \item $\cV_1$ and $\cP_1^*$ take $1^\secp$ as input.
            \item They run the RSP. $\cV_1$ outputs $v\in\{(\mathsf{pass},x,\theta),\mathsf{fail}\}$. 
            $\cP_1^*$ outputs a quantum state $\sigma_\mathbf{Q'}$ on the register $\mathbf{Q}'$.\label{step:RSP}
            \item 
            $\cP_1^*$ runs a certain QPT algorithm $E$ on 
            $\sigma_\mathbf{Q'}$ to get
            $(s,\rho)$, where $s$ is a classical bit string and $\rho$ is an 
            $m_2$-qubit quantum state: $(s,\rho)\gets E(\sigma_\mathbf{Q'})$.
            $\cP_1^*$ outputs $(s,\rho)$.
            \label{encodehyb0}
            \item 
            $\cV_2$ takes $v$ as input.
            $\cP_2^*$ takes $(s,\rho)$ as input.\label{V2P2hyb0}
            \item If $v=\mathsf{fail}$, then $\cV_2$ samples $\theta\gets\bit^{\lceil9.1m_2\rceil}$ and sends it to $\cP_2^*$. 
            If $v=(\mathsf{pass},x,\theta)$, $\cV_2$ sends $\theta$ to $\cP_2^*$.
            \item $\cP_2^*$ sends $x'$ to $\cV_2$.
            \item If $v=\mathsf{fail}$ or $x\neq x'$, $\cV_2$ outputs $\bot$. 
            Otherwise, $\cV_2$ outputs $\top$.
        \end{enumerate}
        
        Because of $\frac{1}{2p}$-soundness of the RSP (\cref{def:RSPVsound}), there exists a non-uniform QPT algorithm $\mathsf{Sim}$.
        By using it, we define $\mathbf{Hybrid}_1$, which is the same as $\mathbf{Hybrid}_0$ except for the step 2:

        \paragraph{$\mathbf{Hybrid}_1$:}
        \begin{enumerate}
%            \item $\cC$ and $\cA_1$ get $1^\secp$ as input.
            \item[2.] $\cV_1$ samples $(x,\theta)\gets\bit^{\lceil9.1m_2\rceil}\times\bit^{\lceil9.1m_2\rceil}$ and generates $\mathsf{Sim}(\bigotimes_{i=1}^{\lceil 9.1m_2\rceil} \allowbreak H^{\theta_i}|x_i\rangle\langle x_i|H^{\theta_i})$, 
            which consists of two registers $\mathbf{F}$ and $\mathbf{Q}'$. 
            $\cV_1$ gets $\mathsf{flag}\in\{\mathsf{pass,fail}\}$ by measuring the register $\mathbf{F}$ and sends the register $\mathbf{Q'}$ of the post-measurement state
            %$\Pi_\mathsf{flag}(\mathsf{Sim}(\bigotimes_{i=1}^{\lceil 9.1m_2(\secp)\rceil} H^{\theta_i}|x_i\rangle\langle x_i|H^{\theta_i}))$ 
            to $\cP_1^*$. $\cV_1$ sets $v\coloneqq\mathsf{fail}$ 
            if $\mathsf{flag}=\mathsf{fail}$. Otherwise, $\cV_1$ sets $v\coloneqq(\mathsf{pass},x,\theta)$. 
%            \item $\cC$ sets $v=(\mathsf{pass},x,\theta)$ if $\mathsf{flag}=\mathsf{pass}$. Otherwise, $\cC$ sets $v=\mathsf{fail}$.
%            \item $\cA_1$ outputs a bit string $\state$ and a $\lfloor 0.11\lceil9.1m(\secp)\rceil\rfloor$-qubit quantum state $\rho$.
%            \item $\cA_1$ sends $(\state,\rho)$ to $\cA_2$.
%            \item If $v=(\mathsf{pass},x,\theta)$, $\cC$ sends $\theta$ to $\cA_2$. Otherwise, $\cC$ outputs $\bot$.
%            \item $\cA_2$ sends $x'$ to $\cC$.
%            \item If $x=x'$, $\cC$ outputs $\top$. Otherwise, $\cC$ outputs $\bot$.
        \end{enumerate}

        \begin{lemma} \label{lem:hyb0.1}
        \begin{align}
       |\Pr[\top\gets \mathbf{Hybrid}_0(\secp)]-\Pr[\top\gets\mathbf{Hybrid}_1(\secp)]|\le\frac{1}{2p(\secp)}
        \end{align}
        for all sufficiently large $\secp\in\mathbb{N}$.
            \begin{proof}[Proof of \cref{lem:hyb0.1}]
                For the sake of contradiction, we assume that 
                \begin{align}
                    |\Pr[\top\gets \mathbf{Hybrid}_0(\secp)]-\Pr[\top\gets\mathbf{Hybrid}_1(\secp)]|>\frac{1}{2p(\secp)} \label{Eq:nothyb01}
                \end{align}
                for infinitely many $\secp\in\mathbb{N}$.
                Then, we can construct a non-uniform QPT algorithm $\cD$ that breaks $\frac{1}{2p}$-soundness of $(\lceil9.1m_2\rceil,\frac{1}{2p})$-RSP as follows:

                %\begin{itemize}
                    %\item $\top/\bot\gets\cD\qty(\frac{\Pi^\mathsf{pass}_\mathbf{F}\sigma_\mathbf{F,D,Q'}\Pi^\mathsf{pass}_\mathbf{F}}{\Tr[\Pi^\mathsf{pass}_\mathbf{F}\sigma_\mathbf{F,D,Q'}]})$:
                \begin{enumerate}
                    \item Get a quantum state over registers $\mathbf{F},\mathbf{D}, \mathbf{Q'}$, where registers are defined as in \cref{def:RSPV}. 
                    \item Get $(x,\theta)\in\bit^{\lceil9.1m_2\rceil}\times\bit^{\lceil9.1m_2\rceil}$ by measuring register $\mathbf{D}$. Set $v=(\mathsf{pass},x,\theta)$.
                    \item Run $(s,\rho)\gets E(\xi_{\mathbf{Q'}})$, where $E$ is the algorithm of step \ref{encodehyb0} of $\mathbf{Hybrid}_0$ and $\xi_\mathbf{Q'}$ is the reduced state on the register $\mathbf{Q'}$ of the post-measurement state.
                    \item Simulate the interaction between $\cV_2$ and $\cP_2^*$ from the step \ref{V2P2hyb0} of $\mathbf{Hybrid}_0$ to the last step.
                    \item Output $\top$ if $\cV_2$ outputs $\top$. Otherwise, output $\bot$.
                \end{enumerate}
                %\end{itemize}

                It is clear that 
                \begin{align}
                    &\Pr[\top\gets\mathbf{Hybrid}_0(\secp)] \\
                    &=\Tr[\Pi^\mathsf{pass}_\mathbf{F}\sigma_\mathbf{F,D,Q'}]\Pr[\top\gets\cD
                    \left(\frac{\Pi^\mathsf{pass}_\mathbf{F}\sigma_\mathbf{F,D,Q'}\Pi^\mathsf{pass}_\mathbf{F}}{\Tr[\Pi^\mathsf{pass}_\mathbf{F}\sigma_\mathbf{F,D,Q'}]}\right)] ,
                \end{align}
                where $\sigma_\mathbf{F,D,Q'}$ is the output of $\cV_1$ and $\cP_1^*$ at step \ref{step:RSP} of $\mathbf{Hybrid}_0$, and
                \begin{align}
                    &\Pr[\top\gets\mathbf{Hybrid}_1(\secp)] \\
                    &=\Tr[\Pi^\mathsf{pass}_\mathbf{F}\mathsf{Sim}(\eta_{\mathbf{D,Q}})]\Pr[\top\gets\cD\qty(\frac{\Pi^\mathsf{pass}_\mathbf{F}\mathsf{Sim}(\eta_{\mathbf{D,Q}})\Pi^\mathsf{pass}_\mathbf{F}}{\Tr[\Pi^\mathsf{pass}_\mathbf{F}\mathsf{Sim}(\eta_{\mathbf{D,Q}})]})].
                \end{align}
                By \cref{Eq:nothyb01}, 
                \begin{align}
                    &\left|\Tr[\Pi^\mathsf{pass}_\mathbf{F}\sigma_\mathbf{F,D,Q'}]\Pr[\top\gets\cD\qty(\frac{\Pi^\mathsf{pass}_\mathbf{F}\sigma_\mathbf{F,D,Q'}\Pi^\mathsf{pass}_\mathbf{F}}{\Tr[\Pi^\mathsf{pass}_\mathbf{F}\sigma_\mathbf{F,D,Q'}]})]\right. \\
                    &\left.-\Tr[\Pi^\mathsf{pass}_\mathbf{F}\mathsf{Sim}(\eta_{\mathbf{D,Q}})]\Pr[\top\gets\cD\qty(\frac{\Pi^\mathsf{pass}_\mathbf{F}\mathsf{Sim}(\eta_{\mathbf{D,Q}})\Pi^\mathsf{pass}_\mathbf{F}}{\Tr[\Pi^\mathsf{pass}_\mathbf{F}\mathsf{Sim}(\eta_{\mathbf{D,Q}})]})]\right| \\
                    &=|\Pr[\top\gets \mathbf{Hybrid}_0(\secp)]-\Pr[\top\gets \mathbf{Hybrid}_1(\secp)]|\\
                    &>\frac{1}{2p(\secp)}
                \end{align}
                for infinitely many $\secp\in\mathbb{N}$. This contradicts $\frac{1}{2p}$-soundness of the RSP.
            \end{proof}
        \end{lemma}

        Let us define $\mathbf{Hybrid}_2$, which is the same as $\mathbf{Hybrid}_1$ except for the step 2:

        \paragraph{$\mathbf{Hybrid}_2$:}
        \begin{enumerate}
            %\item $\cC$ and $\cA_1$ get $1^\secp$ as input.
            \item[2.] $\cV_1$ samples $(x,\theta)\gets\bit^{\lceil9.1m_2\rceil}\times\bit^{\lceil9.1m_2\rceil}$, generates $\bigotimes_{i=1}^{\lceil 9.1m_2\rceil}\allowbreak H^{\theta_i}|x_i\rangle\langle x_i|H^{\theta_i}$ and sends it to $\cP_1^*$. $\cV_1$ sets $v\coloneqq(\mathsf{pass},x,\theta)$.
            %\item $\cA_1$ outputs a bit string $\state$ and a $\lfloor 0.11\lceil9.1m(\secp)\rceil\rfloor$-qubit quantum state $\rho$.
            %\item $\cC$ sends $\theta$ to $\cA_2$.
            %\item $\cA_2$ sends $x'$ to $\cC$.
            %\item If $x=x'$, $\cC$ outputs $\top$. Otherwise, $\cC$ outputs $\bot$. 
        \end{enumerate}

        As shown below, the acceptance probability of $\mathbf{Hybrid}_2$ is at least that of $\mathbf{Hybrid}_1$. 

        \begin{lemma} \label{lem:hyb1.2}
        \begin{align}
        \Pr[\top\gets\mathbf{Hybrid}_1(\secp)]\le\Pr[\top\gets\mathbf{Hybrid}_2(\secp)]
        \end{align}
            for all $\secp\in\mathbb{N}$.
        \end{lemma}   

        \begin{proof}[Proof of \cref{lem:hyb1.2}]
        We start by expanding the acceptance probability of $\mathbf{Hybrid}_1$.
        \begin{align}
            &\Pr[\top\gets\mathbf{Hybrid}_1(\secp)] \\
            &=\Pr[\mathsf{pass}] \\
            &\times\Pr[\top\gets\langle\cV_2(\mathsf{pass},x,\theta),\cP_2^*(s,\rho)\rangle:
            \begin{array}{l}
                (x,\theta)\gets\bit^{\lceil9.1m_2\rceil}\times\bit^{\lceil9.1m_2\rceil} \\
                (s,\rho)\gets E(\zeta^{x,\theta}_\mathbf{Q'})
            \end{array}
            ].
        \end{align}
            Here, $\zeta^{x,\theta}_\mathbf{Q'}$ is the reduced state on register $\mathbf{Q}'$ of
            \begin{align}
                \frac{\Pi^\mathsf{pass}_\mathbf{F}\mathsf{Sim}(\bigotimes_{i=1}^{\lceil 9.1m_2\rceil} H^{\theta_i}|x_i\rangle\langle x_i|H^{\theta_i})\Pi^\mathsf{pass}_\mathbf{F}}{\Tr[\Pi^\mathsf{pass}_\mathbf{F}\mathsf{Sim}(\bigotimes_{i=1}^{\lceil 9.1m_2\rceil} H^{\theta_i}|x_i\rangle\langle x_i|H^{\theta_i})]},
            \end{align} 
            and 
            \begin{align}
                \Pr[\mathsf{pass}]=\Tr[\Pi^\mathsf{pass}_\mathbf{F}\mathsf{Sim}(\bigotimes_{i=1}^{\lceil 9.1m_2\rceil} H^{\theta_i}|x_i\rangle\langle x_i|H^{\theta_i})].
            \end{align}

            Next, we write down explicitly the acceptance probability of $\mathbf{Hybrid}_2$. Then we bound this probability by considering the special case where $\cP_1^*$, as its first step after applying $\mathsf{Sim}$, measures register $\mathbf{F}$, and conditions on the measurement outcome being $\mathsf{pass}$. We further restrict by replacing the state with its reduced version. Since these restrictions can only reduce the acceptance probability, the resulting experiment provides a lowerbound for $\mathbf{Hybrid}_2$, which coincides exactly with the acceptance probability of $\mathbf{Hybrid}_1$.
            \begin{align}
                &\Pr[\top\gets\mathbf{Hybrid}_2(\secp)] \\
                &=\Pr[\top\gets\langle\cV_2(\mathsf{pass},x,\theta),\cP_2^*(s,\rho)\rangle:
            \begin{array}{l}
                (x,\theta)\gets\bit^{\lceil9.1m_2\rceil}\times\bit^{\lceil9.1m_2\rceil} \\
                (s,\rho)\gets E\qty(\bigotimes_{i=1}^{\lceil 9.1m_2\rceil} H^{\theta_i}|x_i\rangle\langle x_i|H^{\theta_i})
            \end{array}
            ] \\
            &\ge\Pr[\mathsf{pass}]\\
            &\times\Pr[\top\gets\langle\cV_2(\mathsf{pass},x,\theta),\cP_2^*(s,\rho)\rangle:
            \begin{array}{l}
                (x,\theta)\gets\bit^{\lceil9.1m_2\rceil}\times\bit^{\lceil9.1m_2\rceil} \\
                (s,\rho)\gets E\qty(\frac{\Pi^\mathsf{pass}_\mathbf{F}\mathsf{Sim}(\bigotimes_{i=1}^{\lceil 9.1m_2\rceil} H^{\theta_i}|x_i\rangle\langle x_i|H^{\theta_i})\Pi^\mathsf{pass}_\mathbf{F}}{\Tr[\Pi^\mathsf{pass}_\mathbf{F}\mathsf{Sim}(\bigotimes_{i=1}^{\lceil 9.1m_2\rceil} H^{\theta_i}|x_i\rangle\langle x_i|H^{\theta_i})]})
            \end{array}
            ] \\
            &\ge\Pr[\mathsf{pass}] \\
            &\times\Pr[\top\gets\langle\cV_2(\mathsf{pass},x,\theta),\cP_2^*(s,\rho)\rangle:
            \begin{array}{l}
                        (x,\theta)\gets\bit^{\lceil9.1m_2\rceil}\times\bit^{\lceil9.1m_2\rceil}  \\
                        (s,\rho)\gets E(\zeta^{x,\theta}_\mathbf{Q'})
            \end{array}
            ] \\
            &=\Pr[\top\gets\mathbf{Hybrid}_1(\secp)].
            \end{align}
            
        \end{proof}

        We define $\mathbf{Hybrid}_3$, which is the same as $\mathbf{Hybrid}_2$ except for steps 3 and 4:

        \paragraph{$\mathbf{Hybrid}_3$:}
        \begin{enumerate}
            \item[3.]
            $\cP_1^*$ runs a certain QPT algorithm $E$ on 
            $\bigotimes_{i=1}^{\lceil 9.1m_2\rceil} H^{\theta_i}|x_i\rangle\langle x_i|H^{\theta_i}$ to get
            $(s,\rho)$, where $s$ is a classical bit string and $\rho$ is an 
            $m_2$-qubit quantum state: $(s,\rho)\gets E(\bigotimes_{i=1}^{\lceil 9.1m_2\rceil} H^{\theta_i}|x_i\rangle\langle x_i|H^{\theta_i})$. Get a bit string $p\in\bit^{m_2}$ by measuring $\rho$ in the computational basis. Set $s'\coloneqq(s,p)$. 
            $\cP_1^*$ outputs $s'$.
            \label{encode}
            \item [4.]
            $\cV_2$ takes $v$ as input.
            $\cP_2^*$ takes $s'=(s,p)$ as input, and uses it as $(s,|p\rangle\langle p|)$.
        \end{enumerate}
        
        By \cref{lemma:bz}, we can obtain the following lemma.

        %Precisely, Hybrid 2 and Hybrid 3 are not the PoQM's security game since $\cC$ and $\cA_1$ interact over a quantum channel. But this quantum communication does not affect the proof of \cref{thm:0_to_m}. Thus, by \cref{thm:0_to_m}, we can obtain the following lemma.

        \begin{lemma} \label{lem:hyb2.3}
        \begin{align}
                \Pr[\top\gets\mathbf{Hybrid}_2(\secp)]\le2^{m_2(\secp)}\Pr[\top\gets\mathbf{Hybrid}_3(\secp)]
        \end{align}
            for all $\secp\in\mathbb{N}$.
        \end{lemma}

        To conclude the theorem, we show the following lemma.

        \begin{lemma} \label{lem:QCbound}
        For all sufficiently large $\secp\in\mathbb{N}$,
        \begin{align}
            \Pr[\top\gets\mathbf{Hybrid}_3(\secp)]\le2^{-\frac{\xi}{2}\cdot \lceil9.1m_2(\secp)\rceil + 2^{-\lceil9.1m_2(\secp)\rceil}}.
        \end{align}
        Here, $\xi=-\log(\frac{1}{2}+\frac{1}{2\sqrt{2}})>0.22$.
            \begin{proof}[Proof of \cref{lem:QCbound}]
                For the sake of contradiction, we assume that there exists a pair $(\cP_1^*,\cP_2^*)$ of adversaries such that
                \begin{align}
                    2^{-\frac{\xi}{2}\cdot \lceil9.1m_2(\secp)\rceil + 2^{-\lceil9.1m_2(\secp)\rceil}} 
                    &<\Pr[\top\gets\mathbf{Hybrid}_3(\secp)]\\
                    &=\Pr[x=x':
                    \begin{array}{l}
                        (x,\theta)\gets\bit^{\lceil9.1m_2(\secp)\rceil}\times \bit^{\lceil9.1m_2(\secp)\rceil}  \\
                        (s,p)\gets\cP_1^*(\bigotimes_{i=1}^{\lceil 9.1m_2(\secp)\rceil}H^{\theta_i}|x_i\rangle\langle x_i|H^{\theta_i}) \\
                        x'\gets\cP_2^*(\theta,s,|p\rangle\langle p|)
                    \end{array}
                    ]
                \end{align}
                for infinitely many $\secp\in\mathbb{N}$.
                From this $(\cP_1^*,\cP_2^*)$, we can construct a non-uniform QPT adversary $\cA$ that breaks \cref{lem:LOCC_leak} as follows:
                \begin{enumerate}
                    \item Send the classical description of $\cP_1^*$ to $\cC$. 
                    \item
                    $\cC$ runs $(s,p)\gets\cP_1^*(\bigotimes_{i=1}^{\lceil 9.1m_2(\secp)\rceil}H^{\theta_i}|x_i\rangle\langle x_i|H^{\theta_i})$ and returns 
                    $(s,p)$ to $\cA$.
                    \item Receive $\theta$ from $\cC$, run $x'\gets\cP_2^*(\theta,s,|p\rangle\langle p|)$ and send $x'$ to $\cC$.
                \end{enumerate}
                Then, for infinitely many $\secp\in\mathbb{N}$,
                \begin{align}
                \Pr[\top\gets\cC]&=
                    \Pr[x=x':
                    \begin{array}{l}
                        (x,\theta)\gets\bit^{\lceil9.1m_2(\secp)\rceil}\times \bit^{\lceil9.1m_2(\secp)\rceil}  \\
                        (s,p)\gets\cP_1^*(\bigotimes_{i=1}^{\lceil 9.1m_2(\secp)\rceil}H^{\theta_i}|x_i\rangle\langle x_i|H^{\theta_i}) \\
                        x'\gets\cP_2^*(\theta,s,|p\rangle\langle p|)
                    \end{array}
                    ] \\
                    &=\Pr[\top\gets\mathbf{Hybrid}_3(\secp)]\\
                   & >2^{-\frac{\xi}{2}\cdot \lceil9.1m_2(\secp)\rceil + 2^{-\lceil9.1m_2(\secp)\rceil}}.
                \end{align}
                This contradicts \cref{lem:LOCC_leak}.
            \end{proof}
        \end{lemma}
        By combining \cref{lem:hyb0.1,lem:hyb1.2,lem:hyb2.3,lem:QCbound}, we have 
        \begin{align}
        \Pr[\top\gets\mathbf{Hybrid}_0(\secp)]<\frac{1}{2p(\secp)}+\negl(\secp)<\frac{1}{p(\secp)} 
        \end{align}
        for all sufficiently large $\secp\in\mathbb{N}$.
    \end{proof}

\section{Lowerbounds of (IV-)PoQM}
In this section, we show that IV-PoQM imply StatePuzzs (and therefore OWPuzzs).
We also show that predictable PoQM imply QCCC KE.

\subsection{IV-PoQM imply StatePuzzs} \label{sec:PoQMtospuz}
We first show that IV-PoQM imply StatePuzzs.
\begin{theorem}\label{thm:PoQMtoStatepuz}
Let $\alpha,\beta:\mathbb{N}\to[0,1]$ be any functions such that $\alpha(\secp)-\beta(\secp)\ge1/\poly(\secp)$ for all sufficiently large $\secp\in\mathbb{N}$. Let $m_1:\mathbb{N}\to\mathbb{N}$ be any function.
    If $(\alpha,\beta,m_1,0)$-IV-PoQM exist, then StatePuzzs exist.
\end{theorem}

\begin{proof}[Proof of \cref{thm:PoQMtoStatepuz}]   
    Let $(\cV_1,\cP_1,\allowbreak\cV_2,\cP_2)$ be an $(\alpha,\beta,m_1,0)$-IV-PoQM. 
     The final state before the measurement of $\cP_1$ is written as $\sum_\state c_\state |\phi_{\mathsf{state}}\rangle|\state\rangle$ 
     with some complex coefficients $\{c_\state\}$, where
        $|\phi_\state\rangle$ is a pure $m_1'$-qubit state. 
        $\cP_1$ measures the second register to get the result $\state$. 
        $\cP_1$ outputs $(\state,\sigma_\state)$, where $\sigma_\state$ is the first $m_1$ qubits of $|\phi_\state\rangle$.
    
    From $(\cV_1,\cP_1,\cV_2,\cP_2)$, we construct an $(\alpha,\beta,m_1',0)$-IV-PoQM $(\cV_1,\cP_1',\cV_2,\cP_2')$ as follows:
%    where $\cP_1'$ outputs the $m_1'$-qubit pure state $|\phi_\state\rangle$ and the classical bit string $\state$, and
%    $\cP_2'$ takes $(\state,|\phi_\state\rangle)$ as input as follows: 
    
    \begin{enumerate}
        \item 
        $\cP_1'$ runs $\cP_1$ to generate
        $\sum_\state c_\state |\phi_{\mathsf{state}}\rangle|\state\rangle$, measures the second register,
        and outputs
        $(\state,|\phi_\state\rangle)$.
        \item 
        $\cP_2'$ takes
        $(\state,|\phi_\state\rangle)$ as input, and runs $\cP_2(\state,\sigma_\state)$,
        where $\sigma_\state$ is the first $m_1$ qubits of $|\phi_\state\rangle$.
    \end{enumerate}
         From $(\cV_1,\cP_1',\cV_2,\cP_2')$, we construct a StatePuzz, $\Samp$, as follows:
        \begin{enumerate}
        \item
        Take $1^\secp$ as input.
            \item Run $(v,(\state,\ket{\phi_\state}))\gets\langle\cV_1(1^\secp),\cP_1'(1^\secp)\rangle$. Let $\tau$ be the transcript.
            \item Output $s\coloneqq(\state,\tau)$ and $\ket{\psi_s}\coloneqq\ket{\phi_{\state}}$.
        \end{enumerate}
        
        Now we show that thus constructed $\Samp$ is a $1/p$-StatePuzz with a certain polynomial $p$.
        From \cref{thm:StatePuzz}, such a $1/p$-StatePuzz can be amplified to obtain a StatePuzz.

        Let $p$ be a polynomial such that $p(\secp)>(\alpha(\secp)-\beta(\secp))^{-2}$ for all sufficiently large $\secp\in\mathbb{N}$. For the sake of contradiction, we assume that $\Samp$ is not a $1/p$-StatePuzz. Then there exists a non-uniform QPT algorithm $\mathcal{A}$ such that for infinitely many 
        $\secp\in\mathbb{N}$,
            \begin{align}\label{eq:notStatePuzz}
                \underset{(s,\ket{\psi_s})\gets \Samp(1^\secp)}{\mathbb{E}}\langle\psi_s|\cA(s)|\psi_s\rangle> 1 - \frac{1}{p(\secp)}.
            \end{align}

From this $\cA$, we construct a pair  
        $(\cP_1^*,\cP_2^*)$ of adversaries that breaks $(\beta,0)$-soundness of the IV-PoQM as follows:
        \begin{itemize}
            \item $\cP_1^*$:
            Run $(\state,|\phi_\state\rangle)\gets\cP_1'(1^\secp)$. Let $\tau$ be the transcript. Output $(\state,\tau)$.
            \item $\cP_2^*$:
                 Run $\cA(\state,\tau)$. 
                 Run $\cP_2'(\state,\cA(\state,\tau))$.
        \end{itemize}
        
        $(\cP_1^*,\cP_2^*)$ can break $(\beta,0)$-soundness of the IV-PoQM as follows. 
        \begin{align}
            &\Pr[\top\gets \langle\cV_2(v),\cP_2^*(\state,\tau)\rangle: (v,(\state,\tau))\gets \langle\cV_1(1^\secp), \cP_1^*(1^\secp)\rangle] \\
            &=\Pr[\top\gets\langle\cV_2(v),\cP_2'(\state,\cA(\state,\tau))\rangle:(v,(\state,\ket{\phi_\state}))\gets\langle\cV_1(1^\secp),\cP_1'(1^\secp)\rangle] \\
            &\ge \Pr[\top\gets \langle\cV_2(v),\cP_2'(\state,\ket{\phi_\state})\rangle: (v,(\state,\ket{\phi_\state}))\gets \langle \cV_1(1^\secp), \cP_1'(1^\secp)\rangle] \\
            &- \underset{((\state,\tau),\ket{\phi_\state})\gets \Samp(1^\secp)}{\mathbb{E}}\TD(\ket{\phi_\state},\cA(\state,\tau)) \\
            &\ge \alpha(\secp) - \underset{((\state,\tau),\ket{\phi_\state})\gets \Samp(1^\secp)}{\mathbb{E}}\sqrt{1-\langle\phi_\state|\cA(\state,\tau)|\phi_\state\rangle}            
        \end{align}
        for all sufficiently large $\secp\in\mathbb{N}$.
        By Jensen's inequality and \Cref{eq:notStatePuzz},
        \begin{align}
            &\underset{((\state,\tau),\ket{\phi_\state})\gets \Samp(1^\secp)}{\mathbb{E}}\sqrt{1-\langle\phi_\state|\cA(\state,\tau)|\phi_\state\rangle} \\
            &\le \sqrt{1-\underset{((\state,\tau),\ket{\phi_\state})\gets \Samp(1^\secp)}{\mathbb{E}}\langle\phi_\state|\cA(\state,\tau)|\phi_\state\rangle} \\
            &< \frac{1}{p(\secp)^{\frac{1}{2}}}
        \end{align}
        for infinitely many $\secp\in\mathbb{N}$.
        Therefore,
        \begin{align}
            &\Pr[\top\gets \langle\cV_2(v),\cP_2^*(\state,\tau)\rangle: (v,(\state,\tau))\gets \langle\cV_1(1^\secp), \cP_1^*(1^\secp)\rangle] \\
            &> \alpha(\secp) - \frac{1}{p(\secp)^{\frac{1}{2}}} \\
            &>\beta(\secp)
        \end{align}
        for infinitely many $\secp\in\mathbb{N}$.
        This contradicts $(\beta,0)$-soundness of the IV-PoQM.
    \end{proof}

\subsection{Predictable PoQM Imply QCCC KE}
We next show that
a restricted version of PoQM, which we call 
predictable PoQM, implies QCCC KE. 

Predictable PoQM are defined as follows.

\begin{definition}[Predictable PoQM] \label{def:prePoQM}
    We call 
    an $(\alpha,\beta,m_1,m_2)$-PoQM 
    an $(\alpha,\beta,m_1,m_2)$-predictable PoQM if 
    the execution phase is of a single round, and has the following specific form:
   \begin{enumerate}
    \item 
    $\cV_2$ takes $v$ as input, where $v$ is the output of $\cV_1$. 
   \item 
    $\cP_2$ takes $(\state,\sigma_\state)$ as input, where $(\state,\sigma_\state)$ is the output of $\cP_1$. 
  \item 
   $\cV_2$ sends a bit string $x$ to $\cP_2$. 
   \item
   $\cP_2$ sends a bit string $y$ to $\cV_2$. 
  \item 
   $\cV_2$ computes $y'$ from $v$ and $x$. 
   If $y=y'$, $\cV_2$ outputs $\top$.
   If $y\neq y'$, $\cV_2$ outputs $\bot$. 
   \end{enumerate} 
   \if0
    Moreover, we require that 
    \begin{align}  \label{eq:pre}
        \Pr[y\gets\cV_2(v,x):
       \begin{array}{l} 
        (v,(\state,\sigma_\state);\tau)\gets\langle\cV_1(1^\secp),\cP_1(1^\secp)\rangle\\ 
     x\gets \cV_2(v)\\
        y\gets\cP_2(\state,\sigma_\state,x)
       \end{array}
        ]\ge\gamma(\secp).
    \end{align}
     Here, $(v,(\state,\sigma_\state);\tau)\gets\langle\cV_1(1^\secp),\cP_1(1^\secp)\rangle$ means that
        $\cV_1$'s output is $v$, $\cP_1$'s output is $(\state,\sigma_\state)$, and $\tau$ is the transcript.
        \fi
\end{definition}

The construction of PoQM from RSPs in \cref{Sec:RSPtoPoQM} 
realizes the predictable PoQM. Thus, we obtain the following lemma.
\begin{lemma}
    Let $p$ be any polynomial. Let $m_2$ be any polynomially bounded function such that $m_2(\secp)=\omega(\log(\secp))$. Assuming the polynomial hardness of LWE, $r$-round $(1-\negl,1/p,\lceil9.1m_2\rceil,m_2)$-predictable PoQM exist with a certain polynomial $r$.
\end{lemma}

%\mor{The above PoQM is 1/poly soundness, but the following result is for negl-soundness PoQM. Is it OK?}
We show that predictable PoQM imply QCCC KE. 
\begin{theorem} \label{thm:urPoQMtoKE}
    Let $m_1:\mathbb{N}\to\mathbb{N}$ be any function. Let $p$ be any polynomial. If $(1-\negl,1-\frac{1}{p},m_1,0)$-predictable PoQM exist, then QCCC KE exist.
\end{theorem}

\if0
\begin{theorem}\label{thm:urPoQMtoKE}
    Let $m_1:\mathbb{N}\to\mathbb{N}$ be any function. Let $\alpha:\mathbb{N}\to[0,1]$ be any function. Let $c_1$ and $c_2$ be any constants such that $c_1>c_2>0$. Let $p(\secp)\coloneqq\secp^{c_1}$ and $q(\secp)\coloneqq\secp^{c_2}$. If $(\alpha,\alpha-\frac{1}{q},m_1,0)$-predictable PoQM with extraction probability $1-\frac{1}{p}$ exist, then QCCC KE exist.
\end{theorem}
\fi

\begin{proof}[Proof of \cref{thm:urPoQMtoKE}]
    Let $(\cV_1,\cP_1,\cV_2,\cP_2)$ be an $(1-\negl,1-\frac{1}{p},m_1,0)$-predictable PoQM. From it, we construct a QCCC KE $(\cA,\cB)$ as follows:
    \begin{enumerate}
        \item $\cA$ and $\cB$ take $1^\secp$ as input.
        \item $\cA$ runs $\cP_1(1^\secp)$, 
        and $\cB$ runs $\cV_1(1^\secp)$. Let $(\state,\sigma_\state)$ be $\cP_1$'s output. Let $v$ be $\cV_1$'s output.
        \item $\cB$ runs $x\gets\cV_2(v)$ and sends $x$ to $\cA$.
        \item $\cA$ runs $y\gets\cP_2(\state,\sigma_\state,x)$, but does not send $y$ to $\cB$.
        %\item  
        %$\cA$ runs $y\gets\cP_2(\state,\sigma_\state)$, and $\cB$ runs $\cV_2(v)$, but $\cA$ does not send $y$ to $\cB$.
        \item $\cB$ runs $y'\gets\cV_2(v,x)$.
        \item $\cA$ outputs $a\coloneqq y$, and $\cB$ outputs $b\coloneqq y'$.
    \end{enumerate}
    Now we show that thus constructed $(\cA,\cB)$ is a $(1-\negl,1-\frac{1}{2p})$-QCCC KE. From \cref{lem:KEamp}, such a $(1-\negl,1-\frac{1}{2p})$-QCCC KE can be amplified to obtain a QCCC KE.
    
    $(1-\negl)$-correctness is clear from $(1-\negl)$-completeness of predictable PoQM. Next, we show $(1-\frac{1}{2p})$-security. 
    For the sake of contradiction, we assume that $(\cA,\cB)$ is not $(1-\frac{1}{2p})$-secure. This means that there exists a non-uniform QPT adversary $\cE$ such that 
        \begin{align} 
        \Pr[y\gets\cE(\tau,x):
       \begin{array}{l} 
        (v,(\state,\sigma_\state);\tau)\gets\langle\cV_1(1^\secp),\cP_1(1^\secp)\rangle\\ 
        x\gets \cV_2(v)\\
        y\gets\cP_2(\state,\sigma_\state,x)
        \end{array}
        ]> 1-\frac{1}{2p(\secp)}
    \end{align}
    for infinitely many $\secp\in\mathbb{N}$. Here, $(v,(\state,\sigma_\state);\tau)\gets\langle\cV_1(1^\secp),\cP_1(1^\secp)\rangle$ means that $\cV_1$'s output is $v$, $\cP_1$'s output is $(\state,\sigma_\state)$, and $\tau$ is the transcript.
    From this $\cE$, we can construct 
    a pair $(\cP_1^*,\cP_2^*)$ of adversaries
    that breaks $(1-\frac{1}{p},0)$-soundness of the predictable PoQM as follows:
    \begin{itemize}
        \item $\cP_1^*:$ Run $(\state,\sigma_\state)\gets\cP_1(1^\secp)$. Let $\tau$ be the transcript. Output $\tau$.
        \item $\cP_2^*:$ Take $\tau$ and $x$ as input,  
        and run $e\gets\cE(\tau,x)$. Send $e$ to $\cV_2$.
    \end{itemize}
    $(\cP_1^*,\cP_2^*)$ breaks $(1-\frac{1}{p},0)$-soundness as follows:
    \begin{align}
        &\Pr[\top\gets\langle\cV_2(v),\cP_2^*(\tau)\rangle:(v,\tau)\gets\langle\cV_1(1^\secp),\cP_1^*(1^\secp)\rangle] \\
  &=\Pr\left[\top\gets\cV_2(v,x,e):
       \begin{array}{l} 
        (v,(\state,\sigma_\state);\tau)\gets\langle\cV_1(1^\secp),\cP_1(1^\secp)\rangle\\ 
        x\gets \cV_2(v)\\
        e\gets\cE(\tau,x)
        \end{array}
        \right]\\
   &=\Pr\left[\top\gets\cV_2(v,x,e):
       \begin{array}{l} 
        (v,(\state,\sigma_\state);\tau)\gets\langle\cV_1(1^\secp),\cP_1(1^\secp)\rangle\\ 
        x\gets \cV_2(v)\\
        e\gets\cE(\tau,x)\\
        y\gets\cP_2(\state,\sigma_\state,x)
        \end{array}
        \right]\\
    &\ge\Pr\left[\top\gets\cV_2(v,x,e)\wedge e=y:
       \begin{array}{l} 
        (v,(\state,\sigma_\state);\tau)\gets\langle\cV_1(1^\secp),\cP_1(1^\secp)\rangle\\ 
        x\gets \cV_2(v)\\
        e\gets\cE(\tau,x)\\
        y\gets\cP_2(\state,\sigma_\state,x)
        \end{array}
        \right]\\
     &=\Pr\left[\top\gets\cV_2(v,x,y)\wedge y\gets\cE(\tau,x):
       \begin{array}{l} 
        (v,(\state,\sigma_\state);\tau)\gets\langle\cV_1(1^\secp),\cP_1(1^\secp)\rangle\\ 
        x\gets \cV_2(v)\\
        y\gets\cP_2(\state,\sigma_\state,x)
        \end{array}
        \right]\\
        &>1-\frac{1}{2p(\secp)}-\negl(\secp) \label{union2} \\
        &>1-\frac{1}{p(\secp)}
    \end{align}
    for infinitely many $\secp\in\mathbb{N}$. 
    Here, in \cref{union2}, we have used the union bound.
    This contradicts $(1-\frac{1}{p},0)$-soundness of the predictable PoQM.
\end{proof}

\if0
\begin{remark}
One might think that we could also define ``predictable IV-PoQM''.
However, they actually imply PoQM. See App.?
\mor{App kaku}\takashi{I don't think this holds under the current definition. In fact, predictable IV-PoQM is unlikely to imply key exchange.}
\end{remark}
\fi

\begin{remark}
We can actually define a weaker notion of predictable PoQM, and can show that they also imply QCCC KE.
See \cref{sec:weak}.
\end{remark}
\section{Constant-Round PoQ Imply Single-Round QV-PoQ or PoQM}
\label{sec:nontrivialPoQ}
In this section, we show the following result.\footnote{
Note that in this theorem we consider only uniform PPT (QPT) adversaries, because
as we can see in the proof, 
we can show only the uniform setting.
}

\if0
Without loss of generality, we can assume that the total number of messages exchanged is even, because we can add dummy messages.
In this paper 
a single round consists of two messages, first one being from the verifier to the prover
and the second one from the prover to the verifier.
Therefore, $\ell$-round means that the total number of messages exchanged is $2\ell$. 
\fi

\begin{theorem}
\label{thm:main}
Let $\ell:\mathbb{N}\to\mathbb{N}$ be any polynomial.
Let $\alpha,\beta:\mathbb{N}\to[0,1]$ be any functions such that $\alpha(\secp)-\beta(\secp)\ge\frac{1}{\poly(\secp)}$ for all $\secp\in\mathbb{N}$.  
Let $\epsilon,\delta:\mathbb{N}\to[0,1]$ be any functions such that 
$\epsilon(\secp)\ge\frac{1}{\poly(\secp)}$, $\delta(\secp)\ge\frac{1}{\poly(\secp)}$
and
$\left(\alpha(\secp)-\epsilon(\secp)\right)-\left(\beta(\secp)+\delta(\secp)\right)\ge\frac{1}{\poly(\secp)}$ for all $\secp\in\mathbb{N}$.
If $\ell$-round $(\alpha,\beta)$-PoQ exist, 
then
one of the following primitives exists:
\begin{itemize}
\item 
$\ell$-round $(\alpha,\alpha-\epsilon,\poly,0)$-PoQM (with i.o.-soundness).
%\item 
%$(\ell-1)$-round $(c-\epsilon,s)$-PoQ exist.
\item
single-round $(\alpha-\epsilon,\beta+\delta)$-QV-PoQ (with i.o.-soundness). 
\item 
$(\ell-1)$-round $(\beta+\delta,\beta)$-PoQ.
\end{itemize}
Note that the above single-round QV-PoQ is a special one where the verifier's second computation (i.e., the computation to make the decision) is PPT.
\end{theorem}

\begin{remark}
With a similar proof, we also obtain the same theorem except that
the single-round QV-PoQ is a special one where the verifier's first computation (i.e., the computation that generates the challenge) is PPT.
\end{remark}

\begin{proof}[Proof of \cref{thm:main}]
Let $(\cV,\cP)$ be an $\ell$-round $(\alpha,\beta)$-PoQ.    
We consider the following two cases.

\paragraph{Case 1.}
%Assume that 
For any QPT algorithm $\cP^*$ whose internal state after sending its $(\ell-1)$th message is classical,
\begin{align}\label{eq:P*}
\Pr[\top\gets\langle \cV(1^\secp),\cP^*(1^\secp)\rangle]\le \alpha(\secp)-\epsilon(\secp)
\end{align}
for infinitely many $\secp\in\mathbb{N}$.
Then we can construct an $(\alpha,\alpha-\epsilon,\poly,0)$-PoQM with i.o.-soundness as follows:
\begin{itemize}
    \item 
    {\it Initialization phase:}
    The verifier and the prover run $\cV$ and $\cP$, respectively, until the end of the $(\ell-1)$th round. 
    \item 
    {\it Execution phase:}
    The verifier and the prover run the $\ell$th round of $\cV$ and $\cP$.
    Finally, the verifier outputs the output of $\cV$.
\end{itemize}
Its $\alpha$-completeness immediately follows from that of the original PoQ. 
Its i.o.-$(\alpha-\epsilon,0)$-soundness is also clear from \Cref{eq:P*}. 

\paragraph{Case 2.}
%Assume that 
Case 1 does not hold.
In other words,
there exists a QPT algorithm $\cP^*$ whose internal state after sending its $(\ell-1)$th message is classical
such that
\begin{align}\label{eq:P*_break}
\Pr[\top\gets\langle \cV(1^\secp),\cP^*(1^\secp)\rangle]> \alpha(\secp)-\epsilon(\secp)
\end{align}
for all sufficiently large $\secp\in\mathbb{N}$.

\if0
Assume that
$\cP^*$'s $\ell$th computation is PPT.
Then we can construct an $(\ell-1)$-round $(\alpha-\epsilon,\beta)$-PoQ as follows.
\begin{enumerate}
    \item 
    The verifier and the prover run $\cV$ and $\cP^*$, respectively, until the end of the $(\ell-1)$th round.
    \item 
    When prover sends its $(\ell-1)$th message, it also sends its classical internal state.
    \item 
    The verifier simulates the the $\ell$th round between $\cV$ and $\cP^*$ by itself. 
    \item 
    The verifier outputs the output of $\cV$.
\end{enumerate}
Its i.o.-$(\alpha-\epsilon)$-completeness and
$\beta$-soundness are clear.
\fi

Let us construct a single-round $(\alpha-\epsilon,\beta+\delta)$-QV-PoQ with i.o.-soundness as follows:
\begin{enumerate}
    \item 
    The verifier simulates the interaction between $\cV$ and $\cP^*$ until the end of the $\cV$'s $\ell$th computation (i.e., by the step when $\cV$ outputs its $\ell$th message).
    \item 
    The verifier sends the transcript and $\cP^*$'s classical internal state
    to the prover.
    \item 
    The prover runs the $\ell$th computation of $\cP^*$ (i.e., the prover computes $\cP^*$'s $\ell$th message), and sends the $\ell$th message to the verifier.
    \item 
    The verifier makes the decision.
\end{enumerate}
Its $(\alpha-\epsilon)$-completeness follows from \Cref{eq:P*_break}. 
If it is i.o.-$(\beta+\delta)$-sound, we obtain single-round $(\alpha-\epsilon,\beta+\delta)$-QV-PoQ with i.o.-soundness.

Assume that it is not i.o.-$(\beta+\delta)$-sound. 
Then, there exists a PPT algorithm $\cA$ that breaks the QV-PoQ.
Let us define the following QPT algorithm $\cP^{**}$:
\begin{enumerate}
\item
Run $\cP^*$ until the end of the $(\ell-1)$th round.
\item
Run the PPT algorithm $\cA$ for its $\ell$th computation.
\end{enumerate}
Note that $\cP^{**}$'s internal state after sending $(\ell-1)$th message is classical.
$\cP^{**}$ satisfies
\begin{align}\label{eq:P**}
\Pr[\top\gets\langle \cV(1^\secp),\cP^{**}(1^\secp)\rangle]> \beta(\secp)+\delta(\secp)
\end{align}
for all sufficiently large $\secp\in\mathbb{N}$.
Then, we can construct $(\ell-1)$-round $(\beta+\delta,\beta)$-PoQ as follows. 
\begin{enumerate}
    \item 
    The verifier and the prover run $\cV$ and $\cP^{**}$, respectively, until the end of the $(\ell-1)$th round.
    \item 
    When prover sends its $(\ell-1)$th message, it also sends its classical internal state.
    \item 
    The verifier simulates the $\ell$th round between $\cV$ and $\cP^{**}$ by itself. 
    \item 
    The verifier outputs the output of $\cV$.
\end{enumerate}
Its 
$(\beta+\delta)$-completeness follows from \Cref{eq:P**} and  
$\beta$-soundness follows from that of the original PoQ.

\if0
Assume that there exist a polynomial $p$ and a QPT algorithm $\cP^*$ such that $\cP^*$'s last computation is PPT and
\begin{align}
\Pr[\top\gets\langle \cV(1^\secp),\cP^*(1^\secp)\rangle]\ge c-\frac{1}{p(\secp)}.
\end{align}
Then we obtain an $(\ell-1)$-round PoQ $(\cV^*,\cP^*)$, 
where $\cP^*$ sends its classical internal state when it sends the $(\ell-1)$th message,
and $\cV^*$ simulates the last computation of $\cP^*$.

Assume that for any polynomial $p$ and any QPT algorithm $\cP^*$ whose last computation is PPT,
\begin{align}
\Pr[\top\gets\langle \cV(1^\secp),\cP^*(1^\secp)\rangle]< c-\frac{1}{p(\secp)}.
\end{align}
\fi

\end{proof}
 
By repeating \cref{thm:main}, we obtain the following corollary.
\begin{corollary}
\label{coro:main}
If constant-round PoQ exist, then one of the following primitives exists.
\begin{itemize}
    \item 
constant-round PoQM with $m_2=0$ (with i.o.-soundness).
\item 
single-round QV-PoQ (with i.o.-soundness).
%\item 
%single-round PoQ.
\end{itemize} 
%\takashi{The second item is redundant since it implies the third item.} \oki{revised}
%Note that the above 1-round QV-PoQ is a special one where the verifier's second computation (i.e., the computation to make the decision) is PPT.
\end{corollary}

\begin{proof}[Proof of Corollary~\ref{coro:main}]
Let $\ell\ge2$.
Assume that $\ell$-round $(\alpha_0,\beta_0)$-PoQ exist with some $\alpha_0,\beta_0$
such that $\alpha_0(\secp)-\beta_0(\secp)\ge\frac{1}{\poly(\secp)}$
for all $\secp\in\mathbb{N}$.
\begin{enumerate}
\item 
Define 
$\epsilon_1\coloneq\frac{(\alpha_0-\beta_0)}{n}$
and
$\delta_1\coloneq\frac{(\alpha_0-\beta_0)(n-2)}{n}$,
where $n\coloneq 2\ell(\ell-1)$.
Then from \cref{thm:main}, we obtain $(\ell-1)$-round $(\alpha_1,\beta_1)$-PoQ
with $\alpha_1\coloneq \beta_0+\delta_1$ and $\beta_1\coloneq \beta_0$, 
or other primitives (i.e., PoQM or single-round QV-PoQ) with $\frac{1}{\poly}$ completeness-soundness gap.
\item
If we obtain $(\ell-1)$-round $(\alpha_1,\beta_1)$-PoQ, 
define
$\epsilon_2\coloneq\frac{(\alpha_1-\beta_1)}{n}$
and
$\delta_2\coloneq\frac{(\alpha_1-\beta_1)(n-2)}{n}=\delta_1\frac{n-2}{n}=(\alpha_0-\beta_0)\left(\frac{n-2}{n}\right)^2$.
Then from \cref{thm:main}, we obtain $(\ell-2)$-round $(\alpha_2,\beta_2)$-PoQ
with $\alpha_2\coloneq \beta_1+\delta_2=\beta_0+\delta_2$ and $\beta_2\coloneq \beta_1=\beta_0$, or other primitives with $\frac{1}{\poly}$ completeness-soundness gap.
\item
If we obtain $(\ell-2)$-round $(\alpha_2,\beta_2)$-PoQ, 
define
$\epsilon_3\coloneq\frac{(\alpha_2-\beta_2)}{n}$
and
$\delta_3\coloneq\frac{(\alpha_2-\beta_2)(n-2)}{n}=\delta_2\frac{n-2}{n}=(\alpha_0-\beta_0)\left(\frac{n-2}{n}\right)^3$.
Then from \cref{thm:main}, we obtain $(\ell-3)$-round $(\alpha_3,\beta_3)$-PoQ
with $\alpha_3\coloneq \beta_2+\delta_3=\beta_0+\delta_3$ and $\beta_3\coloneq \beta_2=\beta_0$, or other primitives with $\frac{1}{\poly}$ completeness-soundness gap.
\end{enumerate}

If we repeat this, we finally obtain 
single-round $(\alpha_{\ell-1},\beta_{\ell-1})$-PoQ
with $\alpha_{\ell-1}= \beta_0+\delta_{\ell-1}$ 
and $\beta_{\ell-1}= \beta_0$, or other primitives with $\frac{1}{\poly}$ completeness-soundness gap.
Because 
\begin{align}
\delta_{\ell-1}
&=(\alpha_0-\beta_0)\left(\frac{n-2}{n}\right)^{\ell-1}
=(\alpha_0-\beta_0)\left(1-\frac{2}{n}\right)^{\ell-1}\\
&=(\alpha_0-\beta_0)\left(1-\frac{1}{\ell(\ell-1)}\right)^{\ell-1}
\ge
(\alpha_0-\beta_0)\left(1-\frac{1}{\ell}\right)
\ge\frac{1}{\poly},    
\end{align}
we have $\alpha_{\ell-1}(\secp)-\beta_{\ell-1}(\secp)\ge\frac{1}{\poly(\secp)}$
for all $\secp\in\mathbb{N}$.

\if0
Assume that $(\ell-i)$-round $(\alpha_i,\beta_i)$-PoQ exist with some $\alpha_i,\beta_i$   
such that $\alpha_i(\secp)-\beta_i(\secp)\ge\frac{1}{\poly(\secp)}$
for all $\secp\in\mathbb{N}$.
Take 
$\epsilon_i\coloneq\frac{(\alpha_i-\beta_i)}{n}$
and
$\delta_i\coloneq\frac{(\alpha_i-\beta_i)(n-2)}{n}$,
where $n\coloneq 2\ell(\ell-1)$.
Then, by induction, we obtain
\begin{align}
\delta_{\ell-1}
&=(\alpha_{\ell-1}-\beta_{\ell-1})\frac{n-2}{n}\\
&=\delta_{\ell-2}\frac{n-2}{n}\\
&=(\alpha_{\ell-2}-\beta_{\ell-2})\left(\frac{n-2}{n}\right)^2\\
&=\delta_{\ell-3}\left(\frac{n-2}{n}\right)^2\\
&\cdots\\
&=(\alpha_0-\beta_0)\left(\frac{n-2}{n}\right)^{\ell-1}\\
&\ge
(\alpha_0-\beta_0)\left(1-\frac{1}{\ell}\right).    
\end{align}
Hence we obtain 1-round $(\alpha,\beta)$-PoQ with
$\alpha(\secp)-\beta(\secp)\ge\frac{1}{\poly(\secp)}$
for all $\secp\in\mathbb{N}$.
\fi

\end{proof}

%We will first show \cref{coro:main}, and then show \cref{thm:main}.

\begin{remark}
\label{remark:Kabir}
In the above proof, we can see that the new verifier of the $(\ell-1)$-round PoQ runs
$\cA$.
If $\cA$'s run-time is proportional to the square of the run-time of the verifier of the $\ell$-round PoQ,
for example,
the run-time of the new verifier is proportional to the square of the run-time of the old verifier.
Hence, if we apply \cref{thm:main} polynomially-many times, the verifier's run-time can blowup exponentially.
This is why \cref{coro:main} holds only for constant $\ell$.\footnote{We thank Kabir Tomer for pointing out this fact.}
\end{remark}

\section{Constant-Round IV-PoQ with Black-Box Reduction Imply IV-PoQM}
\label{sec:IVPoQM}
%In this section, we show that $\ell$-round PoQ with black-box reduction to a falsifiable assumption imply $\ell$-round PoQM or $(\ell-1)$-round PoQ with black-box reduction to a falsifiable assumption.

\subsection{Preliminaries}
Before showing the main result of this section, we give definitions of several concepts.

\if0
First, we clarify a quantum query to a classical oracle.
\begin{definition}
    Let $F:\bit^*\to\bit^*$ be a function. Quantum queries to $F$ are defined as follows. 
    The query prepares registers $\mathbf{XYZ}$ where $\mathbf{X}$ contains the query string, $\mathbf{Y}$ is the response register, 
    and $\mathbf{Z}$ contains auxiliary information. A query consists of applying the unitary
    \begin{align}
        \sum_x|x\rangle\langle x|_\mathbf{X}\otimes X_\mathbf{Y}^{F(x)}\otimes I_{\mathbf{Z}},
    \end{align}
    where $\mathbf{X}\coloneqq \bigotimes_{i=1}^$
\end{definition}
\fi

First,
we recall falsifiable assumptions \cite{C:naor03,STOC:GenWic11}. 
While \cite{STOC:GenWic11} considers efficient challengers, we also allow inefficient challengers.

\begin{definition}[Falsifiable Assumptions \cite{C:naor03,STOC:GenWic11,STOC:TomZha26}]
    A falsifiable assumption $(\cC,\mathsf{thres})$ consists of a (potentially-inefficient) challenger $\cC$ and a constant $\mathsf{thres}\in[0,1)$. The challenger interacts with an adversary $\cA$ and outputs $\top$ or $\bot$. We say that an adversary $\cA$ breaks the assumption if there exists a polynomial $p$ such that for infinitely many $\secp\in\mathbb{N}$,
    \begin{align}
        \Pr[\top\gets\langle \cC(1^\secp),\cA(1^\secp)\rangle]\ge\mathsf{thres}+\frac{1}{p(\secp)}.
    \end{align}
\end{definition}

 \cite{STOC:TomZha26} introduced the notion of \emph{IV-PoQ with black-box reduction}, which is defined as follows.
\begin{definition}[IV-PoQ with Black-Box Reduction \cite{STOC:TomZha26}]
    An $(\alpha,\beta)$-IV-PoQ with black-box reduction to a falsifiable assumption $(\cC,\mathsf{thres})$ is a set $(\cV,\cP,\cR)$ of algorithms. $\cV$ and $\cP$ interact over a classical channel. $\cV$ is an algorithm that takes $1^\secp$ as input and outputs $\top$ or $\bot$. It is PPT during the interaction, and the final computation to make the decision can be unbounded. $\cP$ is a QPT algorithm that takes $1^\secp$ as input and outputs nothing. $\cR$ is a QPT reduction algorithm. We require the following properties:
    \paragraph{\bf{$\alpha$-completeness:}} For all sufficiently large $\secp\in\mathbb{N}$,
    \begin{align}
        \Pr[\top\gets \langle\cV(1^\secp),\cP(1^\secp)\rangle] \ge \alpha(\secp).
    \end{align}
    \paragraph{\bf $\beta$-soundness by reduction.} For any polynomial $p'$, there exists a polynomial $p$ such that for any (not-necessarily-efficient) classical prover $\cA$ that satisfies
\begin{align}
    \Pr[\top\gets\langle\cV(1^\secp),\cA(1^\secp)\rangle]\ge\beta(\secp)+\frac{1}{p'(\secp)}
\end{align}
for infinitely many $\secp\in\mathbb{N}$, it holds that
\begin{align}
    \Pr[\top\gets\langle\cC(1^\secp),\cR^{\cA}(1^\secp)\rangle]\ge\mathsf{thres}+\frac{1}{p(\secp)}
\end{align}
for infinitely many $\secp\in\mathbb{N}$. 
%Here, $\cR$ makes queries by mapping $\ket{x,y}\to\ket{x,y\oplus\cA(x;r)}$, 
%where the internal randomness $r$ is sampled uniformly at random and fixed it for all queries.
Here, $\cR^\cA$ is defined as follows.
Let $x\coloneqq(i,\tau,st,k)$, where $i\in[\ell]$, $\tau$ and $st$ are bit strings, and $k\in\mathbb{N}$.
$\cD_x$ is the distribution of the output of $\cA$
when we run the $i$th computation of $\cA$ on input a transcript $\tau$ and a state $st$.\footnote{The index $k$ is included so that we can sample the output of $\cA$ many times on the same input.}
Sample a classical oracle $\cO$ in such a way that the $x$th row of its truth table is sampled from $\cD_x$.
Once one $\cO$ is sampled, we fix $\cO$, and $\cR$ can quantumly query $\cO$.
\end{definition}

We use the following lemma that is shown in \cite{STOC:TomZha26}.\footnote{\cite{STOC:TomZha26} formulates the lemma in the non-uniform setting, 
but here we state its uniform version, because our results hold only for the uniform setting.}
\begin{lemma}[\cite{STOC:TomZha26}] \label{lem:TZ}
Let $\{U_x\}_{x\in\bit^*}$ be a set of efficient unitaries. 
For each $x\in\bit^*$, let $\cD_x$ be the distribution of the computational-basis measurement results
on $U_x\ket{0...0}$.
    %Let $\{\cD_x\}_{x\in\bit^*}$ be a collection of distributions over bit strings. 
    Let $\cD\coloneqq \bigotimes_{x\in\bit^*} \cD_x$. Let $\cO$ be a classical oracle.
We say an oracle $\cO$ is sampled from $\cD$, and denote by $\mathcal{O}\gets\cD$, 
when the $x$th row of its truth table is sampled from $\cD_x$ for each $x\in\bit^*$. 
%Suppose there exists an efficient unitary $U=\sum_x|x\rangle\langle x| \otimes U_x$ such that $\cD_x$ is the distribution that results from measuring $U_x\ket{0}$ in the computational basis.
Then for any polynomial $q$ and for any QPT algorithm $\cA$, there exists a QPT simulator $\Sim$ such that 
    \begin{align}
        \left|\Pr_{\cO\gets\cD}[\top\gets\cA^\cO(1^\secp)]-\Pr[\top\gets\Sim(1^\secp)]\right|\le\frac{1}{q(\secp)}
    \end{align}
    for all $\secp\in\mathbb{N}$.
\end{lemma}

\subsection{Results}
%We show that $\ell$-round PoQ with black-box reduction to a falsifiable assumption imply $\ell$-round PoQM or $(\ell-1)$-round PoQ with black-box reduction to a falsifiable assumption.
Now we show our main result of this section.\footnote{
Note that in this theorem we consider only uniform PPT (QPT) adversaries, because
as we can see in the proof, 
we can show only the uniform setting.
}

\begin{theorem} \label{thm:IVPoQM}
Let $\ell$ be any polynomial. Let $\alpha,\beta:\mathbb{N}\to[0,1]$ be any functions such that $\alpha(\secp)-\beta(\secp)\ge\frac{1}{\poly(\secp)}$ for all $\secp\in\mathbb{N}$. Let $\epsilon:\mathbb{N}\to[0,1]$ be any function such that 
$\epsilon(\secp)\ge\frac{1}{\poly(\secp)}$ 
and
$\left(\alpha(\secp)-\epsilon(\secp)\right)-\beta(\secp)\ge\frac{1}{\poly(\secp)}$ for all $\secp\in\mathbb{N}$.  Suppose that there exist $\ell$-round $(\alpha,\beta)$-IV-PoQ with black-box reduction to a falsifiable assumption $(\cC,\mathsf{thres})$. 
Then one of the following primitives exists. 
\begin{itemize}
   \item $\ell$-round $(\alpha,\alpha-\epsilon,\poly,0)$-IV-PoQM (with i.o.-soundness).
   \item $(\ell-1)$-round $(\alpha-\epsilon,\beta)$-IV-PoQ with black-box reduction to a falsifiable assumption $(\cC,\mathsf{thres})$. 
\end{itemize}
\end{theorem}
\if0
\begin{proof}(sketch.)
 Construct IV-PoQM in a similar manner to the one in the proof of  \Cref{thm:main}. Assuming that the IV-PoQM is not sound, we construct $(\ell-1)$-round IV-PoQ where the verifier internally runs the final round of the assumed adversary against the IV-PoQM. Then we have an obvious quantum reduction from the soundness of the $(\ell-1)$-round IV-PoQ to that of the original $\ell$-round IV-PoQ. Combining this with the reduction from the soundness of the original $\ell$-round IV-PoQ to X, we obtain quantum reduction from the soundness of the $(\ell-1)$-round IV-PoQ to X.
\end{proof}
\fi

\begin{proof}[Proof of \cref{thm:IVPoQM}]
    %Assume that $\ell$-round IV-PoQ whose soundness is reduced to a falsifiable assumption $(\cC,\mathsf{thres})$ exist. 
    Let $(\cV,\cP,\cR)$ be an $\ell$-round $(\alpha,\beta)$-IV-PoQ with black-box reduction to a falsifiable assumption $(\cC,\mathsf{thres})$. 
    This means that for all sufficiently large $\secp\in\mathbb{N}$,
    \begin{align} \label{eq:ivPoQcom}
        \Pr[\top\gets \langle\cV(1^\secp),\cP(1^\secp)\rangle] \ge \alpha(\secp),
    \end{align}
         and for any polynomial $p'$, there exists a polynomial $p$ such that for any classical prover $\cA$ that satisfies
\begin{align}
    \Pr[\top\gets\langle\cV(1^\secp),\cA(1^\secp)\rangle]\ge\beta(\secp)+\frac{1}{p'(\secp)}
\end{align}
for infinitely many $\secp\in\mathbb{N}$, it holds that
\begin{align} \label{eq:bbr}
    \Pr[\top\gets\langle\cC(1^\secp),\cR^{\cA}(1^\secp)\rangle]\ge\mathsf{thres}+\frac{1}{p(\secp)}
\end{align}
for infinitely many $\secp\in\mathbb{N}$.
    
    We consider the following two cases.

\paragraph{Case 1.}
For any QPT algorithm $\cP^*$ whose internal state after sending its $(\ell-1)$th message is classical,
\begin{align}\label{eq:ivP*}
\Pr[\top\gets\langle \cV(1^\secp),\cP^*(1^\secp)\rangle]\le \alpha(\secp)-\epsilon(\secp)
\end{align}
for infinitely many $\secp\in\mathbb{N}$.
Then we can construct an $(\alpha,\alpha-\epsilon,\poly,0)$-IV-PoQM as follows:
\begin{itemize}
    \item 
    {\it Initialization phase:}
    The verifier and the prover run $\cV$ and $\cP$, respectively, until the end of the $(\ell-1)$th round. 
    \item 
    {\it Execution phase:}
    The verifier and the prover run the $\ell$th round of $\cV$ and $\cP$.
    Finally, the verifier outputs the output of $\cV$.
\end{itemize}
Its $\alpha$-completeness immediately follows from \cref{eq:ivPoQcom}.
Its i.o.-$(\alpha-\epsilon,0)$-soundness is also clear from \Cref{eq:ivP*}.

\paragraph{Case 2.}
%Assume that 
Case 1 does not hold.
In other words,
there exists a QPT algorithm $\cP^*$ whose internal state after sending its $(\ell-1)$th message is classical
such that
\begin{align}\label{eq:ivP*_break}
\Pr[\top\gets\langle \cV(1^\secp),\cP^*(1^\secp)\rangle]> \alpha(\secp)-\epsilon(\secp)
\end{align}
for all sufficiently large $\secp\in\mathbb{N}$. 

%Let $st$ be $\cP^{*}$'s classical internal state after sending its $(\ell-1)$th message.
%Let $m_\ell$ be $\cV$'s $\ell$th message. 
%Let $\cB(s,m_\ell)$ be a (potentially inefficient) classical algorithm whose output message distribution is identical 
%to the output distribution of the last round of $\cP^*$, 
By using $\cP^*$, we construct an $(\ell-1)$-round $(\alpha-\epsilon,\beta)$-IV-PoQ $(\widetilde{\cV},\widetilde{\cP})$ as follows:
\begin{enumerate}
    \item $\widetilde{\cV}$ and $\widetilde{\cP}$ run $\cV$ and $\cP^*$, respectively, until the end of the $(\ell-1)$th round.
    \item When $\widetilde{\cP}$ sends its $(\ell-1)$th message, it also sends its classical internal state $s$. 
%    \item $\widetilde{\cV}$ computes $\cV$'s $\ell$th message $m_\ell$, runs $\cB(s,m_\ell)$, and uses its output to compute the last step of $\cV$. 
\item $\widetilde{\cV}$ simulates the $\ell$th round between $\cV$ and $\cP^*$ by itself,
    and outputs the output of $\cV$.
\end{enumerate}
Its $(\alpha-\epsilon)$-completeness follows from \cref{eq:ivP*_break}. 

\if0
By assumption, there exists a QPT reduction $\cR$ such that for any classical prover $\cP^*$ that satisfies
\begin{align}
    \Pr[\top\gets\langle\cV(1^\secp),\cP^*(1^\secp)\rangle]>\beta(\secp)
\end{align}
for infinitely many $\secp\in\mathbb{N}$, it holds that
\begin{align}
    \Pr[\top\gets\langle\cC(1^\secp),\cR^{\cP^*}(1^\secp)\rangle]>\mathsf{thres}(\secp)
\end{align}
for infinitely many $\secp\in\mathbb{N}$.
\fi

Let us next show $\beta$-soundness by reduction. 
Our goal is to show the existence of a QPT reduction algorithm $\cQ$
such that
for any polynomial $p'$ there exists a polynomial $p$ such that
for any classical algorithm $\widetilde{\cP}^{*}$ that 
satisfies
\begin{align} \label{eq:breakl-1}
    \Pr[\top\gets\langle\widetilde{\cV}(1^\secp),\widetilde{\cP}^{*}(1^\secp)\rangle]\ge\beta(\secp)+\frac{1}{p'(\secp)}
\end{align}
for infinitely many $\secp\in\mathbb{N}$,
it holds that
\begin{align}
    \Pr[\top\gets\langle\cC(1^\secp),\cQ^{\widetilde{\cP}^{*}}(1^\secp)\rangle]\ge\mathsf{thres}+\frac{1}{p(\secp)}
\end{align}
for infinitely many $\secp\in\mathbb{N}$.

For that goal, let us define
a classical (inefficient) algorithm $\cA$ as follows.

\paragraph{$\cA$:} Run $\widetilde{\cP}^{*}$ until the $(\ell-1)$th round.
In the $\ell$th round, simulate the $\ell$th round of $\cP^*$.
%run, but do not send the classical internal state $s$ in the last round. Upon receiving the $\ell$th message $m_\ell$, run $\cB(s,m_\ell)$ and sends its output.

By \cref{eq:breakl-1},
%and $\delta(\secp)\ge\frac{1}{\poly(\secp)}$

\begin{align}
    \Pr[\top\gets\langle\cV(1^\secp),\cA(1^\secp)\rangle]\ge\beta(\secp)+\frac{1}{p'(\secp)}
\end{align}
for infinitely many $\secp\in\mathbb{N}$. 
By \cref{eq:bbr}, there exists a polynomial $p$ such that
\begin{align} \label{eq:breakfa}
    \Pr[\top\gets\langle\cC(1^\secp),\cR^{\cA}(1^\secp)\rangle]\ge\mathsf{thres}+\frac{1}{p(\secp)}
\end{align}
for infinitely many $\secp\in\mathbb{N}$.

%Since $\cB$ is equivalent to the last round of $\cP^*$ (i.e., the output distributions of $\cB(s,m_\ell)$ and  $\cP^*(s,m_\ell)$ are identical) and $\cP^*$ is a QPT algorithm, 

Let us define the QPT reduction algorithm $\cQ^{\widetilde{\cP}^{*}}$ as follows:
\paragraph{$\cQ^{\widetilde{\cP}^{*}}:$} It is identical to $\cR^{\cA}$ except for the following:
when $\cR$ queries $\cA$'s inefficient part (namely, simulating the $\ell$th round of $\cP^*$),
the query is replaced with the efficient simulation of it that is guaranteed by \cref{lem:TZ}. 

By \cref{lem:TZ}, we have
\begin{align}
    \Pr[\top\gets\langle\cC(1^\secp),\cQ^{\widetilde{\cP}^{*}}(1^\secp)\rangle]&\ge\Pr[\top\gets\langle\cC(1^\secp),\cR^{\cA}(1^\secp)\rangle] -\frac{1}{2p(\secp)} \\
    &\ge\mathsf{thres}+\frac{1}{2p(\secp)}
\end{align}
for infinitely many $\secp\in\mathbb{N}$.

% By using $\cR$, we construct a QPT reduction $\cQ^{\cP^{**}}$ as follows:
% \paragraph{$\cQ^{\cP^{**}}:$} It is identical to $\cR^{\cA'}$ except that each query to $\cB$ is replaced by the QPT simulator $\Sim$ in \cref{lem:TZ}.
% %Note that each quantum query to $\cP^*$ runs in polynomial time since $\cP^*$ is a QPT algorithm.
% Since the output distributions of $\cB(s,m_\ell)$ and $\cP^*(s,m_\ell)$ are identical, $\cR^{\cA'}$ and $\cQ^{\cP^{**}}$ generate the transcript according to the same distribution. By \cref{eq:breakfa}, we have
% \begin{align}
%     \Pr[\top\gets\langle\cC(1^\secp),\cQ^{\cP^{**}}(1^\secp)\rangle]\ge\mathsf{thres}+\frac{1}{p(\secp)}
% \end{align}
% for infinitely many $\secp\in\mathbb{N}$. This breaks the falsifiable assumption.
\end{proof}

Repeatedly applying \cref{thm:IVPoQM}, we have the following corollary. 
\begin{corollary} \label{coro:IVPoQM}
If there exist constant-round IV-PoQ with black-box reduction to a falsifiable assumption, then one of the following exists.
\begin{itemize}
   \item constant-round IV-PoQM (with i.o.-soundness).
   \item single-round IV-PoQ with black-box reduction to a falsifiable assumption.  
\end{itemize}
\end{corollary}

\begin{proof}[Proof of \cref{coro:IVPoQM}]
    Let $\ell\ge2$. Assume that $\ell$-round $(\alpha_0,\beta_0)$-IV-PoQ with black box reduction to a falsifiable assumption exist with some $\alpha_0,\beta_0$ such that $\alpha_0(\secp)-\beta_0(\secp)\ge\frac{1}{\poly(\secp)}$ for all $\secp\in\mathbb{N}$.
\begin{enumerate}
    \item Define $\epsilon\coloneqq\frac{\alpha_0-\beta_0}{\ell}\left(\ge\frac{1}{\poly}\right)$. Then from \cref{thm:IVPoQM}, we obtain $(\ell-1)$-round $(\alpha_1,\beta_1)$-IV-PoQ with black-box reduction to a falsifiable assumption with $\alpha_1\coloneqq\alpha_0-\epsilon$ and $\beta_1\coloneqq\beta_0$ or IV-PoQM with $\frac{1}{\poly}$ completeness-soundness gap.
    \item If we obtain $(\ell-1)$-round IV-PoQ with black-box reduction to a falsifiable assumption, then from \cref{thm:IVPoQM}, we obtain $(\ell-2)$-round $(\alpha_2,\beta_2)$-IV-PoQ with black-box reduction to a falsifiable assumption with $\alpha_2\coloneqq\alpha_1-\epsilon=\alpha_0-2\epsilon$ and $\beta_2\coloneqq\beta_1=\beta_0$ or IV-PoQM with $\frac{1}{\poly}$ completeness-soundness gap.
    \item If we obtain $(\ell-2)$-round IV-PoQ with black-box reduction to a falsifiable assumption, then from \cref{thm:IVPoQM}, we obtain $(\ell-3)$-round $(\alpha_3,\beta_3)$-IV-PoQ with black-box reduction to a falsifiable assumption with $\alpha_3\coloneqq\alpha_2-\epsilon=\alpha_0-3\epsilon$ and $\beta_3\coloneqq\beta_2=\beta_0$ or IV-PoQM with $\frac{1}{\poly}$ completeness-soundness gap.
\end{enumerate}
If we repeat this, we finally obtain single-round $(\alpha_{\ell-1},\beta_{\ell-1})$-IV-PoQ with black-box reduction to a falsifiable assumption with $\alpha_{\ell-1}=\alpha_0-(\ell-1)\epsilon$ and $\beta_{\ell-1}=\beta_0$ or IV-PoQM with $\frac{1}{\poly}$ completeness-soundness gap. Because 
\begin{align}
    \alpha_{\ell-1}(\secp)-\beta_{\ell-1}(\secp)=\alpha_0-(\ell-1)\epsilon-\beta_0=\frac{\alpha_0(\secp)-\beta_0(\secp)}{\ell}\ge\frac{1}{\poly(\secp)}
\end{align}
for all $\secp\in\mathbb{N}$, the single-round IV-PoQ with black-box reduction to a falsifiable assumption have $\frac{1}{\poly}$ completeness-soundness gap.
\end{proof}

On the other hand, Tomer and Zhandry~\cite{STOC:TomZha26} showed that single-round IV-PoQ cannot have a quantum black-box reduction to a falsifiable assumption,
otherwise the falsifiable assumption is quantumly broken. 
Combining \cref{coro:IVPoQM} and the result, 
we obtain the following corollary:

\begin{corollary}\label{cor:IV-PoQ_to_IV-PoQM}
If there exist constant-round IV-PoQ with black-box reduction to a falsifiable assumption, 
then there exist constant-round IV-PoQM (with i.o.-soundness) or the falsifiable assumption is quantumly broken.     
\end{corollary}
In particular, \cite{C:MorYam24} constructed constant-round IV-PoQ from (classically-secure) constant-round statistically-hiding commitments in a black-box way, 
and therefore we obtain the following corollary.
\begin{corollary}
    If quantumly-secure constant-round statistically-hiding commitments exist,  then there exist constant-round IV-PoQM (with i.o.-soundness). 
\end{corollary}

Furthermore, constant-round statistically-hiding commitments can be obtained from 
collision-resistant hash functions \cite{C:HalMic96}, thus we also obtain the following corollary.
\begin{corollary}
If quantumly-secure collision-resistant hash functions exist,  then there exist constant-round IV-PoQM (with i.o.-soundness).     
\end{corollary}   
%\takashi{I also thought a bit about if we can show an efficiently-verifiable version of \Cref{cor:IV-PoQ_to_IV-PoQM}, but that doesn't seem to work. (In fact, if we could prove that, then the main results of the paper would be shown as immediate corollaries.)}

\ifnum\anonymous=0
\paragraph{Acknowledgements.}
TM is supported by
JST CREST JPMJCR23I3,
JST Moonshot R\verb|&|D JPMJMS2061-5-1-1, 
JST FOREST, 
MEXT QLEAP, 
the Grant-in Aid for Transformative Research Areas (A) 21H05183,
and 
the Grant-in-Aid for Scientific Research (A) No.22H00522.
\else
\fi

\ifnum\llncs=1
\bibliographystyle{alpha} 
\bibliography{abbrev0,abbrev3,crypto_new,reference}
\else
\bibliographystyle{alpha} 
\bibliography{abbrev0,abbrev3,crypto_new,reference}
\fi

\appendix

\section{Weak Predictable PoQM}
We define the following weaker version of predictable PoQM.
\label{sec:weak}
\begin{definition}[Weak Predictable PoQM]
\label{def:weak}
    % Let $\gamma:\mathbb{N}\to[0,1]$ be any function. 
    A weak predictable $(\alpha,\beta,m_1,m_2)$-PoQM is an
    $(\alpha,\beta,m_1,m_2)$-PoQM with the following single-round execution phase.
    % We call 
    % an $(\alpha,\beta,m_1,m_2)$-PoQM 
    % an $(\alpha,\beta,m_1,m_2)$-predictable PoQM with extraction probability $\gamma$ if 
    % the execution phase is the following.
   \begin{enumerate}
    \item 
    $\cV_2$ takes $v$ as input, where $v$ is the output of $\cV_1$. 
   \item 
    $\cP_2$ takes $(\state,\sigma_\state)$ as input, where $(\state,\sigma_\state)$ is the output of $\cP_1$. 
  \item 
   $\cV_2$ sends a bit string $x$ to $\cP_2$. 
   \item
   $\cP_2$ sends a bit string $y$ to $\cV_2$. 
  \item 
   $\cV_2$ outputs $\top$ or $\bot$. 
   \end{enumerate} 
    Moreover, we require the existence of a QPT algorithm $\mathsf{Ext}$ such that
    \begin{align} \label{eq:ur}
        \Pr[y\gets\mathsf{Ext}(v,\tau,x):
       \begin{array}{l} 
        (v,(\state,\sigma_\state);\tau)\gets\langle\cV_1(1^\secp),\cP_1(1^\secp)\rangle\\ 
        x\gets \cV_2(v)\\
        y\gets\cP_2(\state,\sigma_\state,x)
        \end{array}
        ]\ge\alpha(\secp).
    \end{align}
        Here, $(v,(\state,\sigma_\state);\tau)\gets\langle\cV_1(1^\secp),\cP_1(1^\secp)\rangle$ means that
        $\cV_1$'s output is $v$, $\cP_1$'s output is $(\state,\sigma_\state)$, and $\tau$ is the transcript.
\end{definition}

Predictable PoQM (\cref{def:prePoQM}) satisfy the above definition, since we can construct a QPT algorithm $\mathsf{Ext}$ that satisfies \cref{eq:ur} as follows: On input $(v,\tau,x)$, run $y'\gets\cV_2(v,x)$ and output its output. Thus, we have the following lemma.
\begin{lemma}
    Let $\alpha,\beta:\mathbb{N}\to[0,1]$ be any functions. Let $m_1,m_2$ be any polynomials.  
    If $(\alpha,\beta,m_1,m_2)$-predictable PoQM exist, then
    $(\alpha,\beta,m_1,m_2)$-weak predictable PoQM exist.
\end{lemma}

We show that weak predictable PoQM also imply QCCC KE.
\begin{theorem} \label{thm:wprePoQMtoKE}
    Let $m_1:\mathbb{N}\to\mathbb{N}$ be any function.  Let $p$ be any polynomial. If $(1-\negl,1-\frac{1}{p},m_1,0)$-weak predictable PoQM exist, then QCCC KE exist.
\end{theorem}

\begin{proof}[Proof of \cref{thm:wprePoQMtoKE}]
    Let $(\cV_1,\cP_1,\cV_2,\cP_2)$ be an $(1-\negl,1-\frac{1}{p},m_1,0)$-weak predictable PoQM. From it, we construct a QCCC KE $(\cA,\cB)$ as follows:
    \begin{enumerate}
        \item $\cA$ and $\cB$ take $1^\secp$ as input.
        \item $\cA$ runs $\cP_1(1^\secp)$, 
        and $\cB$ runs $\cV_1(1^\secp)$. Let $(\state,\sigma_\state)$ be $\cP_1$'s output. Let $v$ be $\cV_1$'s output.
        \item $\cB$ runs $x\gets\cV_2(v)$ and sends $x$ to $\cA$.
        \item $\cA$ runs $y\gets\cP_2(\state,\sigma_\state,x)$, but does not send $y$ to $\cB$.
        %\item  
        %$\cA$ runs $y\gets\cP_2(\state,\sigma_\state)$, and $\cB$ runs $\cV_2(v)$, but $\cA$ does not send $y$ to $\cB$.
        \item $\cB$ runs $y'\gets\mathsf{Ext}(v,\tau,x)$.
        \item $\cA$ outputs $a\coloneqq y$, and $\cB$ outputs $b\coloneqq y'$.
    \end{enumerate}
    Now we show that thus constructed $(\cA,\cB)$ is a $(1-\negl,1-\frac{1}{2p})$-QCCC KE. From \cref{lem:KEamp}, such a $(1-\negl,1-\frac{1}{2p})$-QCCC KE can be amplified to obtain a QCCC KE.
    
    $(1-\negl)$-correctness is clear from \cref{eq:ur}. Next, we show $(1-\frac{1}{2p})$-security. 
    For the sake of contradiction, we assume that $(\cA,\cB)$ is not $(1-\frac{1}{2p})$-secure. This means that there exists a non-uniform QPT adversary $\cE$ such that 
        \begin{align} 
        \Pr[y\gets\cE(\tau,x):
       \begin{array}{l} 
        (v,(\state,\sigma_\state);\tau)\gets\langle\cV_1(1^\secp),\cP_1(1^\secp)\rangle\\ 
        x\gets \cV_2(v)\\
        y\gets\cP_2(\state,\sigma_\state,x)
        \end{array}
        ]> 1-\frac{1}{2p(\secp)}
    \end{align}
    for infinitely many $\secp\in\mathbb{N}$. Here, $(v,(\state,\sigma_\state);\tau)\gets\langle\cV_1(1^\secp),\cP_1(1^\secp)\rangle$ means that $\cV_1$'s output is $v$, $\cP_1$'s output is $(\state,\sigma_\state)$, and $\tau$ is the transcript.
    From this $\cE$, we can construct 
    a pair $(\cP_1^*,\cP_2^*)$ of adversaries
    that breaks $(1-\frac{1}{p},0)$-soundness of the PoQM as follows:
    \begin{itemize}
        \item $\cP_1^*:$ Run $(\state,\sigma_\state)\gets\cP_1(1^\secp)$. Let $\tau$ be the transcript. Output $\tau$.
        \item $\cP_2^*:$ Take $\tau$ and $x$ as input,  
        and run $e\gets\cE(\tau,x)$. Send $e$ to $\cV_2$.
    \end{itemize}
    $(\cP_1^*,\cP_2^*)$ breaks $(1-\frac{1}{p},0)$-soundness as follows:
    \begin{align}
        &\Pr[\top\gets\langle\cV_2(v),\cP_2^*(\tau)\rangle:(v,\tau)\gets\langle\cV_1(1^\secp),\cP_1^*(1^\secp)\rangle] \\
  &=\Pr\left[\top\gets\cV_2(v,x,e):
       \begin{array}{l} 
        (v,(\state,\sigma_\state);\tau)\gets\langle\cV_1(1^\secp),\cP_1(1^\secp)\rangle\\ 
        x\gets \cV_2(v)\\
        e\gets\cE(\tau,x)
        \end{array}
        \right]\\
   &=\Pr\left[\top\gets\cV_2(v,x,e):
       \begin{array}{l} 
        (v,(\state,\sigma_\state);\tau)\gets\langle\cV_1(1^\secp),\cP_1(1^\secp)\rangle\\ 
        x\gets \cV_2(v)\\
        e\gets\cE(\tau,x)\\
        y\gets\cP_2(\state,\sigma_\state,x)
        \end{array}
        \right]\\
    &\ge\Pr\left[\top\gets\cV_2(v,x,e)\wedge e=y:
       \begin{array}{l} 
        (v,(\state,\sigma_\state);\tau)\gets\langle\cV_1(1^\secp),\cP_1(1^\secp)\rangle\\ 
        x\gets \cV_2(v)\\
        e\gets\cE(\tau,x)\\
        y\gets\cP_2(\state,\sigma_\state,x)
        \end{array}
        \right]\\
     &=\Pr\left[\top\gets\cV_2(v,x,y)\wedge y\gets\cE(\tau,x):
       \begin{array}{l} 
        (v,(\state,\sigma_\state);\tau)\gets\langle\cV_1(1^\secp),\cP_1(1^\secp)\rangle\\ 
        x\gets \cV_2(v)\\
        y\gets\cP_2(\state,\sigma_\state,x)
        \end{array}
        \right]\\
        &>1-\frac{1}{2p(\secp)}-\negl(\secp)\label{union21} \\
        &>1-\frac{1}{p(\secp)}
    \end{align}
    for infinitely many $\secp\in\mathbb{N}$. 
    Here, in \cref{union21}, we have used the union bound.
    This contradicts $(1-\frac{1}{p},0)$-soundness of the weak predictable PoQM.
\end{proof}
\section{QCCC KE Security Amplification} \label{sec:KEamp}
%\oki{This section is newly added. As pointed out in \cite{cryptoeprint:2026/1215}, the proof of classical KE amplification does not apply to QCCC KE.} 
%\oki{This section will not be included in the camera-ready version.}
In this section, we give a proof of \cref{lem:KEamp}. We use the following lemma.
\begin{lemma}[\cite{STOC:BQSY24}, Theorem 1] \label{parallel}
    Let $k$ be a polynomial. Let $s:\mathbb{N}\to[0,1]$ be any function. Let $\pi$ be a 3-message quantum interactive protocol with soundness error $s(\secp)$ against any non-uniform QPT adversary. Then, the $k$-fold parallel repetition $\pi^{\otimes k}$ has soundness error at most $s(\secp)^{k(\secp)}+\negl(\secp)$ against any non-uniform QPT adversary.
\end{lemma}

We now show \cref{lem:KEamp}.
\begin{proof}[Proof of \cref{lem:KEamp}]
    Let $k(\secp)\coloneqq \secp p(\secp)$. Let $(\cA,\cB)$ be a $(1-\negl,1-\frac{1}{p})$-QCCC KE. We construct a QCCC KE as follows:
    \begin{itemize}
        \item $\cA':$ On input $1^\secp$, run $\cA^{\otimes k}(1^\secp)$. For each $i\in[k]$, let $a_i$ denote the output of $i$-th copy of $\cA$. Output $a\coloneqq a_1\|...\|a_k$.
        \item $\cB':$ On input $1^\secp$, run $\cB^{\otimes k}(1^\secp)$. For each $i\in[k]$, let $b_i$ denote the output of $i$-th copy of $\cB$. Output $b\coloneqq b_1\|...\|b_k$.
    \end{itemize}
    First, we show correctness. By the correctness of $(1-\negl,1-\frac{1}{p})$-QCCC KE,
    \begin{align}
        \Pr[a=b:(a,b;\tau)\gets\langle\cA'(1^\secp),\cB'(1^\secp)\rangle]&\ge 1-\sum_{i=1}^k\Pr[a_i\neq b_i:(a_i,b_i,\tau_i)\gets\langle\cA(1^\secp),\cB(1^\secp)\rangle] \label{ampunion} \\ 
        &\ge 1-\negl(\secp),
    \end{align}
    where \cref{ampunion} follows from the union bound. Thus, $(\cA',\cB')$ satisfies correctness. 

    We next show security. Consider the following $3$-message quantum interactive protocol $\pi$.
    \begin{enumerate}
        \item The adversary $\cE$ sends $1^\secp$ to the challenger $\cC$.
        \item $\cC$ runs $(a,b;\tau)\gets\langle\cA(1^\secp),\cB(1^\secp)\rangle$ and sends $\tau$ to $\cE$.
        \item $\cE$ outputs a bit string $e$ and sends it to $\cC$.
        \item  If $a=e$, $\cC$ outputs $\top$. If $a\neq e$, $\cC$ outputs $\bot$.
    \end{enumerate}
    By the security of $(1-\negl,1-\frac{1}{p})$-QCCC KE, for any non-uniform QPT adversary $\cE$,
    \begin{align}
        \Pr[\top\gets\langle\cC(1^\secp),\cE(1^\secp)\rangle]=\Pr[a=e:(a,b;\tau)\gets\langle\cA(1^\secp),\cB(1^\secp)\rangle,e\gets\cE(\tau)]\le1-\frac{1}{p(\secp)},
    \end{align} 
    for all sufficiently large $\secp\in\mathbb{N}$.
    The security game for $(\cA',\cB')$ is the $k(\secp)$-fold parallel repetition of $\pi$. Therefore, by \cref{parallel}, for any non-uniform QPT adversary $\cE$,
    \begin{align}
        \Pr[a=e:(a,b;\tau)\gets\langle\cA'(1^\secp),\cB'(1^\secp)\rangle,e\gets\cE(\tau)] 
        &\le \left(1-\frac{1}{p(\secp)}\right)^{k(\secp)}+\negl(\secp) \\
        &\le \exp({-\frac{k(\secp)}{p(\secp)}})+\negl(\secp) \label{exp} \\
        &\le\negl(\secp) \label{ampfin}
    \end{align}
    where \cref{exp} follows from $1+x\le \exp(x)$ and \cref{ampfin} follows from $k(\secp)=\secp p(\secp)$. This completes the proof.
\end{proof}

\ifnum\cameraready=1
\else
\ifnum\submission=1
\newpage
\setcounter{tocdepth}{1}
\tableofcontents
\else
\fi
\fi

\end{document}